\documentclass[sigconf]{aamas}

\usepackage{balance} 

\makeatletter
\gdef\@copyrightpermission{
  \begin{minipage}{0.2\columnwidth}
   \href{https://creativecommons.org/licenses/by/4.0/}{\includegraphics[width=0.90\textwidth]{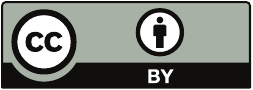}}
  \end{minipage}\hfill
  \begin{minipage}{0.8\columnwidth}
   \href{https://creativecommons.org/licenses/by/4.0/}{This work is licensed under a Creative Commons Attribution International 4.0 License.}
  \end{minipage}
  \vspace{5pt}
}
\makeatother

\setcopyright{ifaamas}
\acmConference[AAMAS '25]{Proc.\@ of the 24th International Conference
on Autonomous Agents and Multiagent Systems (AAMAS 2025)}{May 19 -- 23, 2025}
{Detroit, Michigan, USA}{Y.~Vorobeychik, S.~Das, A.~Nowé  (eds.)}
\copyrightyear{2025}
\acmYear{2025}
\acmDOI{}
\acmPrice{}
\acmISBN{}

\usepackage{booktabs} 
\usepackage[ruled]{algorithm2e} 

\SetAlFnt{\small}
\SetAlCapFnt{\small}
\SetAlCapNameFnt{\small}
\SetAlCapHSkip{0pt}
\IncMargin{-\parindent}

\usepackage{lineno}
\usepackage{bm}
\usepackage{bbold}
\usepackage{amsmath}
\usepackage{hyperref}
\usepackage [mathscr] {eucal}
\usepackage{comment}
\usepackage{tabularx,ragged2e}
\usepackage{mathtools}
\usepackage{tikz}
\usepackage{pgfplots}
\usepackage{subcaption}
\usepackage{enumitem}
\usepackage{dsfont}

\SetKwInput{KwInput}{Input}
\SetKwInput{KwInitialize}{Initialize}

\newtheorem{assumption}{Assumption}

\newtheorem{rthm}{Theorem}
\newenvironment{reptheorem}[1]
  {\renewcommand\therthm{\ref*{#1}}\rthm}
  {\endrthm}

\DeclareSymbolFontAlphabet{\mathbb}{AMSb}
\DeclareMathOperator*{\Exp}{\mathbb{E}}
\newcolumntype{C}{>{\Centering\arraybackslash}X}

\newcommand{\vertiii}[1]{{\left\vert\kern-0.25ex\left\vert\kern-0.25ex\left\vert #1 
    \right\vert\kern-0.25ex\right\vert\kern-0.25ex\right\vert}}

\usepackage{xparse}
\makeatletter
\NewDocumentCommand{\tagx}{om}{%
  \IfNoValueTF{#1}
   {
    \refstepcounter{equation}(\theequation)\label{#2}%
   }
   {
    (#1)\def\@currentlabel{#1}\label{#2}%
   }%
}
\makeatother

\pgfplotsset{compat=1.18}

\acmSubmissionID{73}

\title[To Spend or To Gain]{To Spend or to Gain: Online Learning in Repeated Karma Auctions}

\author{Damien Berriaud}
\orcid{0009-0002-0450-0339}
\affiliation{
  \institution{ETH Zürich}
  \streetaddress{Physikstrasse 3}
  \city{Zürich}
  \postcode{8092}
  \country{Switzerland}}
\email{dberriaud@ethz.ch}

\author{Ezzat Elokda}
\orcid{0000-0002-2415-6494}
\affiliation{
  \institution{ETH Zürich}
  \streetaddress{Physikstrasse 3}
  \city{Zürich}
  \postcode{8092}
  \country{Switzerland}}
\email{elokdae@ethz.ch}

\author{Devansh Jalota}
\orcid{0000-0002-7956-8525}
\affiliation{
  \institution{Stanford}
  \streetaddress{475 Via Ortega}
  \city{Stanford}
  \state{California}
  \postcode{94305}
  \country{USA}}
  \email{djalota@stanford.edu}

\author{Emilio Frazzoli}
\orcid{0000-0002-0505-1400}
\affiliation{
  \institution{ETH Zürich}
  \streetaddress{Sonneggstrasse 3}
  \city{Zürich}
  \postcode{8092}
  \country{Switzerland}}
  \email{emilio.frazzoli@idsc.mavt.ethz.ch}
  
\author{Marco Pavone}
\orcid{0000-0002-0206-4337}
\affiliation{
  \institution{Stanford}
  \streetaddress{475 Via Ortega}
  \city{Stanford}
  \state{California}
  \postcode{94305}
  \country{USA}}
  \email{pavone@stanford.edu}
  
\author{Florian Dörfler}
\orcid{0000-0002-9649-5305}
\affiliation{
  \institution{ETH Zürich} 
  \streetaddress{Physikstrasse 3}
  \city{Zürich}
  \postcode{8092}
  \country{Switzerland}}
  \email{dorfler@ethz.ch}

\begin{abstract}
Recent years have seen a surge of artificial currency-based mechanisms in contexts where monetary instruments are deemed unfair or inappropriate, e.g., in allocating food donations to food banks, course seats to students, and, more recently, even for traffic congestion management.
Yet the applicability of these mechanisms remains limited in repeated auction settings, as it is challenging for users to learn how to bid an artificial currency that has no value outside the auctions.
Indeed, users must jointly learn the value of the currency in addition to how to spend it optimally.
Moreover, in the prominent class of \emph{karma mechanisms}, in which artificial \emph{karma payments}
are redistributed to users at each time step,
users do not only \emph{spend karma} to obtain public resources but also \emph{gain karma} for yielding them.
For this novel class of karma auctions, we propose an \emph{adaptive karma pacing strategy} that learns to bid optimally, and show that this strategy a) is asymptotically optimal for a single user bidding against competing bids drawn from a stationary distribution; b) leads to convergent learning dynamics when all users adopt it; and c) constitutes an approximate Nash equilibrium as the number of users grows.
Our results require a novel analysis in comparison to adaptive pacing strategies in monetary auctions, since we depart from the classical assumption that the currency has known value outside the auctions, and consider that the currency is both spent \emph{and} gained through the redistribution of payments.
\end{abstract}



\newcommand{\BibTeX}{\rm B\kern-.05em{\sc i\kern-.025em b}\kern-.08em\TeX}

\begin{document}


\pagestyle{fancy}
\fancyhead{}


\maketitle 


\section{Introduction}
In 
public resource allocation contexts, the use of monetary instruments to regulate resource consumption is often deemed inequitable (e.g., to manage traffic congestion~\cite{evans1992road,arnott1994welfare,taylor2010addressing,brands2020tradable}), inappropriate (e.g., for organ and food donations~\cite{sonmez2020incentivized,kim2021organ,prendergast2022allocation} or course allocations~\cite{budish2012multi}), or simply undesired (e.g., for peer-to-peer file sharing~\cite{vishnumurthy2003karma,friedman2006efficiency} or babysitting services~\cite{johnson2014analyzing}).
As a consequence, significant attention has been devoted to the study of non-monetary mechanism design~\cite{roth1990two,schummer2007mechanism}, which is known to be challenging due to 
interpersonal comparability~\cite{roberts1980interpersonal} and the lack of a general instrument to manipulate incentives~\cite{gibbard1973manipulation,satterthwaite1975strategy}.

However, a number of mechanisms have seen some recent success in jointly achieving the objectives of fairness, efficiency, and strategy-proofness when resources are allocated \emph{repeatedly over time}~\cite{guo2009competitive,balseiro2019multiagent,guo2020dynamic,gorokh2021monetary,gorokh2021remarkable,siddartha2023robust}.
The core principle of these mechanisms is to restrict the number of times the resource can be consumed and let the users trade off when it is most beneficial for them to do so.
To achieve these goals, many of these mechanisms employ \emph{artificial currencies}~\cite{guo2009competitive,johnson2014analyzing,gorokh2021monetary,gorokh2021remarkable,prendergast2022allocation,siddartha2023robust,elokda2023self}, which involves issuing a budget of non-tradable credits or currency to users which they may use to repeatedly bid for resources.
In artificial currency mechanisms, users, who may have time-varying and stochastic valuations for the resources, must be strategic in their bidding to not deplete the budget too quickly, and to spare currency for periods when they have the highest valuation for the resources.
Thus, artificial currencies serve the dual purpose of monitoring resource consumption and providing a means for users to express their time-varying preferences, resulting in fair and efficient allocations over time.

The literature on artificial currency mechanisms for repeated resource allocation can be broadly categorized in two classes.
In the first class, artificial currency is issued at the beginning of a finite episode only to be spent during the episode~\cite{gorokh2021monetary,gorokh2021remarkable,siddartha2023robust}.
In the second class, artificial currency is issued at the beginning of the episode but can also be gained throughout it, typically by means of peer-to-peer exchanges~\cite{vishnumurthy2003karma,friedman2006efficiency,johnson2014analyzing,elokda2023self} or by redistributing the payments collected in each time step~\cite{prendergast2022allocation,elokda2024carma}.
Some works have referred to this class of artificial currencies as \emph{karma}~\cite{vishnumurthy2003karma,elokda2023self,vuppalapati2023karma}:
when users yield resources to others they gain karma, and when instead they consume resources they lose karma.
Karma mechanisms offer some advantages in comparison to mechanisms relying on initial endowments of artificial currency only because from the system perspective, the resource allocation can be infinitely repeated with no central intervention other than the initial endowment of karma; and from the user perspective, karma is forgiving as it enables users to immediately replenish their budgets by yielding resources.

In both classes of artificial currency mechanisms, the focus thus far has been on analyzing equilibrium properties, including existence~\cite{elokda2023self}, strategy-proofness~\cite{guo2009competitive,gorokh2021monetary}, and robustness~\cite{gorokh2021remarkable,siddartha2023robust}.
However, the problem of \emph{learning how to optimally bid artificial currency in repeated auction settings}, and whether such a learning procedure converges to a Nash equilibrium, remains unaddressed.
This problem holds both significant importance and challenge.
The importance is two-fold:
on one hand, the equilibrium-based analysis of previous studies is only meaningful if an equilibrium is reached;
on the other, devising simple learning rules that align with users' self-interest is crucial to implement these mechanisms in practice.
The challenge stems from the fact that, unlike traditional monetary instruments, artificial currency does not have any value outside the resource allocation context for which it has been issued.
Therefore, users must jointly learn the value of the currency as well as how to spend it optimally.
Moreover, the possibility of gaining currency in the class of karma mechanisms leads to new challenges that are particular to these novel mechanisms.

The problem of learning how to bid optimally in \emph{monetary} auctions is a classical one~\cite{nisan2007algorithmic,krishna2009auction}.
This problem has gained recent traction in the context of repeated, budget-constrained auctions~\cite{balseiro2015repeated,balseiro2019learning,castiglioni2022online,gaitonde2022budget,balseiro2023best,wang2023learning,lucier2024autobidders}, most famously to automate the bidding in multi-period online ad campaigns.
We draw inspiration from these works, and in particular~\cite{balseiro2019learning}, to derive \emph{adaptive pacing strategies} in artificial currency-based auctions.
The first class of artificial currency auctions in which users are issued an initial endowment of currency only is most closely related to works considering \emph{value maximization}~\cite{balseiro2021landscape,gaitonde2022budget,lucier2024autobidders}.
In these works, monetary payment costs are not included in the users' optimization objective,
but these monetary costs still enter the optimization explicitly in constraints, either on individual bids~\cite{gaitonde2022budget} or on the total expenditure (also known as \emph{return of investment constraints})~\cite{lucier2024autobidders}.
In contrast, the cost of spending currency must be learned and does not appear in the optimization objective nor constraints in artificial currency auctions.

Moreover, the second class of karma auctions in which karma is gained throughout the auction campaign leads to new strategic opportunities that are not typically considered in monetary settings.
For instance, even in a second-price auction users have an incentive to bid non-truthfully to maximize the karma gained upon losing.
Furthermore, the preservation of total karma held by the users in this class of auctions leads to challenges in the simultaneous adoption of classical adaptive pacing strategies,
since it becomes impossible for all users to simultaneously deplete their budgets.

\subsection{Contribution}

In this paper, we adopt techniques from adaptive pacing in budget-constrained monetary auctions~\cite{balseiro2019learning} to the two aforementioned classes of artificial currency mechanisms.
The first class of artificial currency mechanisms in which users are issued an initial endowment of currency only is addressed in Appendix~\ref{sec:A}.
This appendix serves the dual purpose of reviewing the standard results~\cite{balseiro2019learning}, as well as modifying the standard framework to address that the cost of spending artifical currency is not known a-priori.
The main technical novelty in Appendix~\ref{sec:A} is to include a lower-bound on the pacing multipliers, since otherwise standard adaptive pacing 
could lead to unbounded bids that quickly deplete the initial budget.

The main paper is thus devoted to the second and more challenging class of karma mechanisms.
We specifically consider mechanisms in which the karma payments are redistributed uniformly in each time step.
We derive an \emph{adaptive karma pacing} strategy for these mechanisms, and show that:
 a) adaptive karma pacing is asymptotically optimal for a single user bidding against competing bids drawn from a stationary distribution; b) when all users adopt adaptive karma pacing, the expected dynamics converge asymptotically to a unique stationary point; and c) adaptive karma pacing constitutes an approximate Nash equilibrium
under the additional assumption that there is a large number of parallel auctions for which the matching probability of any two particular users decays asymptotically to zero.
The novel technical challenges of the karma-based setting in comparison to the monetary setting, and how they are tackled in our paper, are summarized as follows:

\begin{itemize}
    
    \item[$\bullet$] The possibility of gaining karma through redistribution leads to non-truthfulness of the second-price auction, and a complex dependency of the karma budget dynamics on the user's bid.
    This requires relaxing the primal problem’s budget constraint and using non-perfect gradient information in the associated relaxed dual problem;

    \item[$\bullet$] The preservation of karma due to redistribution leads to several complications requiring a novel adaptive karma pacing strategy and analysis.
    It is impossible for all agents to simultaneously deplete their budgets in this setting, which requires removing the target expenditure rate from the multiplier updates.
    This leads to non-uniqueness of the optimal stationary multipliers, requiring further strategy modifications to converge to a unique multiplier profile.
\end{itemize}
As a consequence of these modifications to standard adaptive pacing, our paper performs a novel regret analysis in order to extend previous performance guarantees to the karma-based setting.

The remainder of the paper is organized as follows.
In Section~\ref{sec:model} we introduce the problem formulation including key definitions and notations.
We derive our adaptive karma pacing strategy in Section~\ref{sec:G-derivation}.
Our main results are then included in Section~\ref{sec:G} which establishes performance guarantees for this strategy.
Finally, Section~\ref{sec:disc} discusses the results, shedding light on the key assumptions made and providing directions for future work.
Review of related works, results for artificial currency mechanisms with no redistribution, numerical computation results, 
detailed assumptions and proofs are included in the appendix.
\section{Problem setup}\label{sec:model}

This section introduces the setting studied in the paper, including notations and important definitions.

\subsection{Notation}

We denote by $[N]$ the set $\{1, \dots, N\}$, by $\mathds{1}\{\cdot\}$ the indicator function,
by $(\cdot)^+$ the function $x \mapsto \max \{x,0\}$, and by $P_{[a,b]}(\cdot)$ the projection $x \mapsto \min\{ \max \{ x,a \}, b\}$. 
Scalars $x$ are distinguished from vectors $\bm{x} = (x_i)_{i \in I}$, for some index set $I$, through the use of boldface.
If $x$ is a scalar, then $\underline{x}$ (respectively, $\overline{x}$) is a lower bound (respectively, upper bound) of $x$.
If $\bm{x}$ is a vector, then $\underline{x} = \min_{i \in I} x_i$ (respectively, $\overline{x} = \max_{i \in I} x_i $).
Finally, for the vector $\bm{x}$ and an index $i \in I$, the vector $\bm{x_{-i}} = (x_j)_{j \in I, \: i \neq j}$ is constructed by dropping component $i$. 

\subsection{Setting}
\label{sec:setting}

\begin{figure*}[!tb]
 \centering
 \includegraphics[width = 0.86\textwidth]{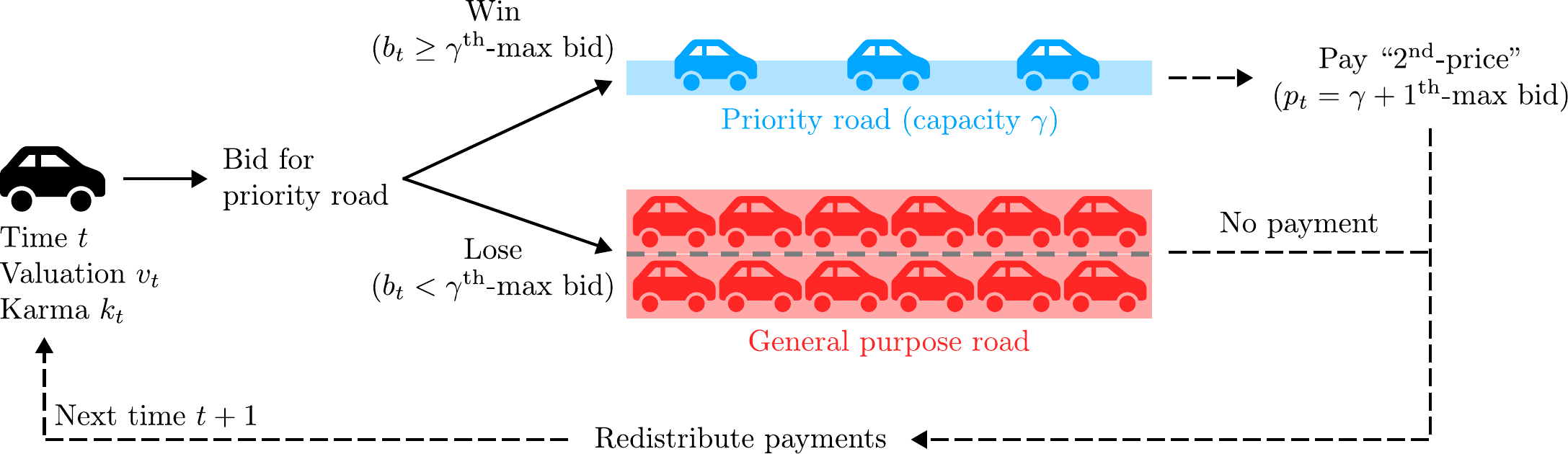}
 \caption{Schematic representation of repeated resource allocation using karma.}
 \Description{At each period, users discover their current valuation and place a bid depending on their remaining karma budget. The highest bidders then get access to the priority road by paying the generalized second price, while the other users are redirected to the general purpose lane. Payments are then potentially redistributed before the start of the new period.}
 \label{fig:karma-scheme}
\end{figure*}

We study a general class of repeated resource allocation problems in which a limited number of resources must be repeatedly allocated to a population $\mathcal{N} = [N]$ of agents.
For the sake of presentation, we instantiate this class of problems using a stylized morning commute setting~\cite{vickrey1969congestion,arnott1990economics}, which is schematically illustrated in Figure~\ref{fig:karma-scheme}.
At discrete time steps $t \in \mathbb{N}$ (e.g., days), the agents seek to commute from the suburb to the city center using one of two roads.
The \emph{general purpose road} is subject to congestion, while access to the \emph{priority road} is limited to its free-flow capacity of $\gamma \in [N-1]$ agents per time step, and therefore it remains free of congestion.
Traveling on the general purpose road takes unit time, 
while traveling on the priority road takes a shorter time $0 \leq 1 - \Delta < 1$.
This model can be interpreted as an abstraction of a multi-lane highway with a governed express lane and un-governed, congested general purpose lanes.
At each time step $t \in \mathbb{N}$, each agent $i \in \mathcal{N}$ is associated with a private \emph{valuation of time} $v_{i,t} \in [0,1]$ drawn independently across time from fixed, exogenous distributions $\mathcal{V}_i$.
The valuations represent the agents' time-varying sensitivities to travel delays, e.g., because they have flexible schedules on some days but must be punctual on other days, and are normalized to the interval $[0,1]$ without loss of generality.
We denote by $\bm{v_t} = \left(v_{i,t} \right)_{i \in \mathcal{N}}$ the vector of agents' valuations at time $t$, which are distributed according to $\bm{\mathcal{V}} =  \prod_{i \in \mathcal{N}}\mathcal{V}_i$ with support over $[0,1]^N$.
As is common in the literature~\cite{balseiro2019learning,gorokh2021monetary,castiglioni2022online}, we assume that the valuation distribution $\bm{\mathcal{V}}$ is absolutely continuous with bounded density $\nu : [0,1]^N \mapsto \left[\underline{\nu}, \overline{\nu}\right]^N \subset \mathbb{R}_{>0}^N.$

Access to the priority road is governed by means of an artificial currency called \emph{karma}.
Each agent $i \in \mathcal{N}$ is endowed with an initial karma budget $k_{i,1} \in \mathbb{R}_+$.
Then, at each time step $t \in \mathbb{N}$, each agent places a sealed bid $b_{i,t} \in \mathbb{R}_+$ smaller than its current budget $k_{i,t}\in \mathbb{R}_+$.
The $\gamma$-highest bidders, referred to as the `auction winners', are granted access to the priority road, and must pay $p_t^{\gamma+1} :=  \gamma+1\textsuperscript{th}\mbox{-}\max\{ b_{i,t}\}_{i=1}^N$ in karma.
The $N-\gamma$ remaining agents, referred to as the `auction losers', must instead take the general purpose road and make no payments.
The price $p_t^{\gamma+1}$ is set by the highest bid among the auction losers, i.e., it corresponds to the second price auction if $\gamma = 1$.
After payments are settled, they are redistributed uniformly to the agents such that each agent gains $g_{i,t} = \gamma \, p_t^{\gamma+1} / N$ units of karma in the next time step.
Notice that under this redistribution scheme, agents that access the priority road have a net decrease in karma, while those using the general purpose road have a net increase in karma.
Meanwhile, the total amount of karma in the system set by the initial endowments $\bm{k_1}$ is preserved over time.

Let $\bm{b_{-i,t}}=(b_{j,t})_{j \neq i}$ be the bid profile of agents other than $i$, and $d_{i,t}^\gamma = \gamma\textsuperscript{th}\mbox{-}\max_{j:j\neq i} \{ b_{j,t}\}$ be the associated \emph{competing bid}, since agent $i$ must bid higher than $d_{i,t}^\gamma$ to be among the auction winners.
We assume that ties in bids do not occur (as is common in the literature~\cite{balseiro2019learning}; in practice ties could be settled randomly).
Let $x_{i,t} = \mathds{1}\{b_{i,t} > d^\gamma_{i,t}\} \in \{0,1\}$ indicate whether agent $i$ is an auction winner at time $t$.
Then the agent suffers a cost $c_{i,t} = v_{i,t} \: (1 - x_{i,t} \Delta)$ 
and pays $z_{i,t} = x_{i,t} \: d^\gamma_{i,t}$ at that time.
Their budget for the next time step is hence determined by $k_{i,t+1} = k_{i,t} - z_{i,t} + g_{i,t}$.

We consider rational agents that aim to minimize their expected total cost over a time horizon $T$. At time $t$, the information available to agent $i$ to formulate its bid is the \emph{history} \\
$\mathcal{H}_{i,t} = \left\{T, (v_{i,s}, k_{i,s}, b_{i,s}, z_{i,s}, c_{i,s})_{s=1}^{t-1} , v_{i,t}, k_{i,t} \right\}.$
A bidding strategy $\beta_i \in \mathcal{B}$ for agent $i$ is thus a sequence of functions $\beta_i = (\beta_{i,1}, \dots, \beta_{i,T})$, where  $\beta_{i,t}$ maps the history $\mathcal{H}_{i,t}$ to a probability distribution over the set of feasible bids $[0, k_{i,t}]$.
A profile of strategies for all agents is accordingly denoted by $\bm{\beta}=\left(\beta_i\right)_{i \in \mathcal{N}}$.
For a fixed agent $i$ following strategy $\beta_i$, it will be convenient to define three notions of expected total cost, given by
\begin{align}
\small
    \mathcal{C}_i^{\beta_i} \left(\bm{v_i},\bm{d_i} \right) &= \Exp_{\bm{b_i} \sim \beta_i} \left[ \sum_{t=1}^T c_{i,t} \right] 
    = \Exp_{\bm{b_i} \sim \beta_i} \left[ \sum_{t=1}^T v_{i,t} \left(1 -   \mathds{1}\left\{ b_{i,t} > d_{i,t}^\gamma \right\}  \Delta \right) \right], \label{equ:C-sample-path} \\
    \mathcal{C}_i^{\beta_i} &= \Exp_{\substack{\bm{v_i} \sim \prod\limits_{t=1}^T \mathcal{V}_i, \; \bm{d_i} \sim \prod\limits_{t=1}^T \mathcal{D}_i, \; \bm{b_i} \sim \beta_i } }
    \left[ \mathcal{C}_i^{\beta_i} (\bm{v_i} ,\bm{d_i}) \right], \label{equ:C-stat-comp} \\
    \mathcal{C}_i^{\bm{\beta}} &=
    \Exp_{\substack{\bm{v} \sim \prod\limits_{t=1}^T \bm{\mathcal{V}}, \; \bm{b_{-i}} \sim \bm{\beta_{-i}}}} \left[\mathcal{C}_i^{\beta_i} \left(\bm{v_i},\bm{d_i}(\bm{b_{-i}}) \right) \right]. \label{equ:C-strategic}
\end{align}
Equation~\eqref{equ:C-sample-path} defines the \emph{sample path cost} $\mathcal{C}_i^{\beta_i} (\bm{v_i},\bm{d_i^\gamma} )$ for a fixed realization of valuations $\bm{v_i} = (v_{i,t})_{t \in [T]}$ and competing bids $\bm{d_i} := \big(d^\gamma_{i,t}, d^{\gamma+1}_{i,t}\big)_{t \in [T]}$.
Equation~\eqref{equ:C-stat-comp} defines the \emph{stationary competition cost} $\mathcal{C}_i^{\beta_i}$, in which
the competing bids $d_{i,t}$ are assumed to be drawn independently across time from a stationary distribution $\mathcal{D}_i$.
Finally, Equation~\eqref{equ:C-strategic} defines the \emph{strategic competition cost} $\mathcal{C}_i^{\bm{\beta}}$, which is agent $i$'s expected total cost when all agents follow strategy profile $\bm{\beta}$.

Note that the above notations implicitly rely on $T$ and we should write  $\bm{\beta}^T$ for a strategy profile over the time horizon $T.$ We now overload notations and consider an infinite series of strategy profiles $\bm{\beta} = \big(\bm{\beta}^T \big)_{T\in \mathbb{N}}. $
With this, we will say that strategy profile $\bm{\beta}$ constitutes an \emph{approximate Nash equilibrium} if its strategic competition cost satisfies, for all agents $i \in \mathcal{N}$,
\begin{equation}\label{equ:Nash}
\small
    \lim_{T \to \infty} \: \frac{1}{T} \left( C_i^{\bm{\beta}^T} - \inf_{\tilde{\beta}_i \in \mathcal{B}} C_i^{\tilde{\beta}_i,\bm{\beta_{-i}}^T}  \right) =0.
\end{equation}
Equation~\eqref{equ:Nash} implies that under strategy profile $\bm{\beta}$, no single agent $i$ can asymptotically improve its expected average cost per time step by unilaterally deviating to a strategy $\tilde{\beta}_i \neq \beta_i$.

\section{Derivation of Adaptive Karma Pacing} \label{sec:G-derivation}

In the remainder of the paper, our main goal is to devise a bidding strategy that constitutes an approximate Nash equilibrium when all agents follow it.
As a first step towards this goal, this section derives a candidate optimal bidding strategy using an \emph{online dual gradient ascent scheme}. This classical optimization technique has gained recent traction in the context of budget-constrained auctions~\cite{balseiro2019learning} and other related problems~\cite{hazan2023introduction,balseiro2023best,pmlr-v206-jalota23a,jalota2023stochastic}.
To elucidate our bidding strategy, we first introduce the optimization problem of a single agent $i \in \mathcal{N}$ who has the \emph{benefit of hindsight}, i.e., who can make optimal bidding decisions with prior knowledge of the future realizations of valuations $\bm{v_i}$ and competing bids $\bm{d_i}$.
Thus, the optimal cost of this problem serves as a theoretical benchmark for the lowest cost that the agent can hope to achieve.
Then, since in practice the agent only observes the stochastic valuations and competing bids online as the auctions progress, we introduce a candidate bidding strategy, based on online gradient ascent, to approximate the agent's optimal bidding strategy with the benefit of hindsight.
\smallskip

\textbf{Optimal Cost with the Benefit of Hindsight.}
We construct hereafter a lower bound on the optimal cost with the benefit of hindsight.
For a fixed realization of valuations $\bm{v_i}$ and competing bids $\bm{d_i}
$, agent $i$'s optimal cost with the benefit of hindsight is given by the following optimization problem
\begin{equation}
\label{equ:G_stat_comp_def_hindsight}
\small
    \begin{aligned}
        &\mathcal{C}_i^H (\bm{v_i},\bm{d_i}) 
        = \min_{\bm{b_i}\in \mathbb{R}_+^T}  \sum_{t=1}^T v_{i,t}\left(1 - \mathds{1} \left\{ b_{i,t} > d_{i,t}^\gamma \right\} \Delta\right), \mbox{ such that}\\
         &\sum_{t=1}^s \mathds{1} \left\{ b_{i,t} > d_{i,t}^\gamma \right\} d_{i,t}^\gamma  \leq \rho_i T +  \sum_{t=1}^{s-1} g_{i,t} \left(b_{i,t}, d_{i,t}\right), \; 
         \forall s \in [T],
    \end{aligned}
\end{equation}
where, following the standard literature~\cite{balseiro2019learning}, we define $\rho_i = k_{i,1} / T $ as the \emph{target expenditure rate}, i.e., the average expenditure per time step that would fully deplete the initial budget by the end of the time horizon if no karma was gained.
Notice that Problem~\eqref{equ:G_stat_comp_def_hindsight} bears significant complexity in comparison to its counterpart in the standard setting with no budget gains, c.f. Problem~\eqref{equ:A_stat_comp_def_hindsight} in Appendix~\ref{sec:A}.
First, the possibility of gaining karma requires a budget constraint for each time step $s \in [T]$ instead of only one at the end of the horizon.
Second, the outcomes $x_{i,t}$ cannot be used directly as decision variables since the gains $g_{i,t}$ depend non-trivially on the bids $b_{i,t}$, as given by
\begin{equation}
\label{equ:karma-gain}
\small
g_{i,t}\left(b_{i,t},d_{i,t}\right) = \frac{\gamma}{N} \: p_t^{\gamma+1}\left(b_{i,t},d_{i,t}\right) = \frac{\gamma}{N} \: \begin{cases}
d_{i,t}^\gamma, &  d_{i,t}^\gamma < b_{i,t} , \\
b_{i,t}, & d_{i,t}^{\gamma+1} < b_{i,t} \leq d_{i,t}^\gamma, \\
d_{i,t}^{\gamma+1}, & b_{i,t} \leq d_{i,t}^{\gamma+1}.
\end{cases}
\end{equation}
Finally, notice that with respect to the standard setting, Problem~\eqref{equ:G_stat_comp_def_hindsight} has an additional dependency on $d_{i,t}^{\gamma+1}$, i.e., the $\gamma+1^\textup{th}$-highest competing bid, c.f. Equation~\eqref{equ:karma-gain}.
If agent $i$ is among the auction winners, the price $p^{\gamma+1}_t$ and thereby the gain $g_{i,t}$ is determined by $d^\gamma_{i,t}$; if the agent is among the auction losers and is not the \emph{price setter}, the gain is determined by $d^{\gamma+1}_{i,t}$; and if the agent is among the auction losers but sets the price the gain is determined by its own bid $b_{i,t}$.
To address the complexity of Problem~\eqref{equ:G_stat_comp_def_hindsight}, we perform a relaxation that forms a lower bound on the optimal cost with the benefit of hindsight, given by
\begin{equation}\label{equ:G_lowerbound_cost_H}
\small
    \begin{aligned}
        \mathcal{C}_i^H (\bm{v_i},\bm{d_i}) 
        \geq \underline{\mathcal{C}}_i^H \big(\bm{v_i},\bm{d_i^\gamma}\big) 
        = &\min_{\bm{x_i}\in \{0,1\}^T}  \sum_{t=1}^T v_{i,t}( 1 - x_{i,t} \Delta), \\
        &\mbox{s.t.} \sum_{t=1}^T \left(x_{i,t} -  \frac{\gamma}{N} \right) d_{i,t}^\gamma  \leq \rho_i T. 
    \end{aligned}
\end{equation}
This lower bound is obtained by a) allowing temporary negative balances of karma as long as it is non-negative at the end of the horizon $T$; and b) eliminating the dependency of $g_{i,t}$ on $b_{i,t}$ and $d^{\gamma+1}_{i,t}$ by assuming that when the agent is among the auction losers, it always manages to be the price setter and impose the maximum gain $\frac{\gamma}{N} \: d^\gamma_{i,t}$.
The Lagrangian dual problem associated with 
Problem~\eqref{equ:G_lowerbound_cost_H} is 
\begin{equation}
\small
\label{equ:G_stat_comp_hindsight-dual}
\begin{aligned}
    &\mathcal{C}_i^H (\bm{v_i},\bm{d_i}) 
    \geq \delta_i^H \big(\bm{v_i},\bm{d_i^\gamma}\big) \\
    &:= \sup\limits_{\mu_i \geq 0} \min\limits_{\bm{x}_i \in \{0,1\}^T} \sum_{t=1}^T  x_{i,t} \left(\mu_i d_{i,t}^\gamma - \Delta v_{i,t}\right) + v_{i,t}  -\mu_i \left(\rho_i + \frac{\gamma}{N} d_{i,t}^\gamma  \right) \\
\end{aligned}
\end{equation}
The relaxation in Problem~\eqref{equ:G_lowerbound_cost_H} results in dual Problem~\eqref{equ:G_stat_comp_hindsight-dual} with similar structure as its counterpart in the standard setting with no budget gains, c.f. Problem~\eqref{equ:A_stat_comp_hindsight-dual-all} in Appendix~\ref{sec:A}.
The main difference is that the target expenditure rate $\rho_i$ is replaced by the time-varying term $\rho_i + \frac{\gamma}{N} \: d^\gamma_{i,t}$.
This is intuitive as the agent now aims to deplete both its initial budget as well as the gains it receives.
Notice that for a fixed multiplier $\mu_i \geq 0$, the inner minimum in~\eqref{equ:G_stat_comp_hindsight-dual} is obtained by winning all auctions satisfying $\Delta v_{i,t} > \mu_i d_{i,t}^\gamma$.
This can be achieved by bidding $b_{i,t}=\Delta v_{i,t} / \mu_i$, yielding
\begin{equation}
\small
\label{equ:G_stat_comp_hindsight-dual-2}
    \begin{aligned}
    \delta_i^H \big(\bm{v_i},\bm{d_i^\gamma}\big)
    &= \sup_{\mu_i \geq 0} \sum_{t=1}^T  v_{i,t}  - \mu_i \left(\rho_i + \frac{\gamma}{N} d_{i,t}^\gamma  \right) - \left(\Delta v_{i,t} - \mu_i d_{i,t}^\gamma \right)^+  \\
    &:=\sup\limits_{\mu_i \geq 0} \sum_{t=1}^T \delta_{i,t}^H \big(v_{i,t}, d_{i,t}^\gamma, \mu_i\big). 
    \end{aligned}
\end{equation}

\smallskip

\textbf{Adaptive Karma Pacing.}
We perform a stochastic gradient ascent scheme in order to approximately solve the relaxed dual Problem~\eqref{equ:G_stat_comp_hindsight-dual-2} using online observations.
Namely, the agent considers a candidate optimal multiplier $\mu_{i,t}$ and places its bid accordingly with $b_{i,t}= \Delta v_{i,t} / \mu_{i,t}$.
In an ideal case, it would then update $\mu_{i,t+1}$ using the subgradient given by 
$
    \frac{\partial \delta_{i,t}^H}{ \partial \mu_{i,t}} \big(v_{i,t}, d_{i,t}^\gamma, \mu_{i,t}\big)  = z_{i,t} - \rho_i - \frac{\gamma}{N} d_{i,t}^\gamma$.
However, adopting this subgradient raises two issues.
First, the term $\frac{\gamma}{N} \: d_{i,t}^\gamma$ is only observed 
by auction winners, hence we use the observed gain $g_{i,t}$ as a proxy instead.
Second, if all agents are to adopt this subgradient, it is impossible for them to simultaneously track $\rho_i$ and fully deplete their karma by the end of the horizon, as the total amount of karma in the system is preserved.
For this reason, we will omit the term $\rho_i$, such that each agent attempts to match its expenditures to its gains.
This yields the 
\emph{adaptive karma pacing strategy},
which is denoted by $K$ and summarized in Algorithm~\ref{alg:G}.

\begin{algorithm}[ht]
\caption{\textbf{Adaptive Karma Pacing $K$}}
\label{alg:G}

\KwInput{ Time horizon $T$, initial budget $k_{i,1}>0$, multiplier bounds $\overline{\mu} > \underline{\mu} > 0$, gradient step size $\epsilon > 0$.}

\KwInitialize{ Initial multiplier $\mu_{i,1} \in [\underline{\mu}, \overline{\mu} ]$.}

\For{ $t=1,\dots,T$}{
\begin{enumerate}
    \item Observe the realized valuation $v_{i,t}$ and place bid $ b_{i,t} = \min \left\{ \dfrac{\Delta v_{i,t}}{ P_{[\underline{\mu}, \overline{\mu} ]} ( \mu_{i,t}) }, k_{i,t} \right\};$ \hfill\tagx[$K$-$b$]{equ:G-b}
    
    \item Observe the expenditure $z_{i,t}$ and the gain $g_{i,t}$. 
    Update the multiplier $ \mu_{i,t+1} =  \mu_{i,t} + \epsilon ( z_{i,t} - g_{i,t}),$ \hfill\tagx[$K$-$\mu$]{equ:G-mu}
    
    as well as the karma budget
    $k_{i,t+1} = k_{i,t} - z_{i,t} + g_{i,t} .$
\end{enumerate}}
\end{algorithm}


The term `adaptive karma pacing' is in line with the literature on budget-constrained monetary auctions~\cite{balseiro2019learning,pmlr-v119-fang20a}, 
but instead of trying to \emph{pace} the budget depletion rate to match the target rate $\rho_i$, strategy $K$ attempts to match the time-varying expenses $z_{i,t}$  with the gains $g_{i,t}.$
Indeed, karma losses $z_{i,t} - g_{i,t} > 0$ increase $\mu_{i,t+1}$, effectively reducing future bids; and vice versa for $z_{i,t} - g_{i,t} < 0$.
Another important novelty in strategy $K$ is that the denominator in the bid~\eqref{equ:G-b} is $\mu_{i,t}$ instead of $1 + \mu_{i,t} $, as common in the standard monetary setting.
This is a consequence of 
the fact that the valuation in karma is not known a-priori; and could lead to a rapid depletion of the budget if $\mu_{i,t}$ becomes small during the learning process even for a short transient period.
For this reason, it is necessary to introduce the lower bound $\underline{\mu}$ in Algorithm~\ref{alg:G}, which we note is unlike 
adaptive pacing algorithms in the monetary setting~\cite{balseiro2019learning}.
Moreover, the projection of multiplier $\mu_i$ on $[\underline{\mu},\overline{\mu}]$ now occurs in the bid~\eqref{equ:G-b} instead of the multiplier update~\eqref{equ:G-mu}.
The importance of this technical difference will be discussed in Section~\ref{sec:G_sim_lear}.

\section{Analysis of Adaptive Karma Pacing}\label{sec:G}

In this section, we analyze the previously derived adaptive karma pacing strategy $K$, with the main goal of establishing that it constitutes an $\varepsilon$-Nash equilibrium when adopted by all agents.
To achieve this goal, this section proceeds as follows.
In section~\ref{sec:G_stat_comp}, we establish that this strategy is asymptotically optimal for a single agent bidding against competing bids drawn from a stationary distribution (Theorem~\ref{thm:G_stat_comp}).
Section~\ref{sec:G_sim_lear} then establishes that the learning dynamics converge to a unique stationary point when all agents follow strategy $K$ (Theorem~\ref{thm:G_sim_lear_cv}).
Finally, Section~\ref{sec:G_eps_NE} combines these results to achieve the main goal of proving that the strategy constitutes an $\varepsilon$-Nash equilibrium under suitable conditions (Theorem~\ref{thm:G_eps_NE}).

\subsection{Asymptotic Optimality under Stationary Competition}\label{sec:G_stat_comp}

In this section, we establish that strategy $K$ is \emph{asymptotically optimal in a stationary competition setting}, where a single agent $i$ bids against competing bids $\bm{d_i}= \big(d_{i,t}^{\gamma}, d_{i,t}^{\gamma+1}\big)_{t \in [T]}$  drawn independently across time from a fixed distribution $\mathcal{D}_i$.
This section, as well as Sections~\ref{sec:G_sim_lear} and~\ref{sec:G_eps_NE}, are organized similarly as follows.
We first state the required definitions, and the new assumptions needed in the karma-based setting\footnote{Standard assumptions in the literature, as well as minor technical assumptions, are included in Appendix~\ref{sec:A} and~\ref{sec:G_add_ass}, respectively.}.
We then present a brief statement of the main result of the section (Theorem~\ref{thm:G_stat_comp} in this case)\footnote{Full theorem statements are included in Appendix~\ref{sec:G_add_ass}.},
and focus our discussion on the differences to the standard setting with no budget gains (c.f. Appendix~\ref{sec:A}).

To state the main result of this section, we must first define the \emph{expected dual objective}, \emph{expected gain}, \emph{expected expenditure}, and \emph{expected karma loss} when agent $i$ follows strategy $K$.
For a fixed multiplier $\mu_i > 0$, these quantities are respectively given by 
\begin{equation} \label{equ:G-dual-objective} 
\small
    \begin{aligned}
        \Psi_i^{0}(\mu_i) &= \Exp_{v_i, d_i} \left[v_i -\mu_i g_i - (\Delta v_i - \mu_i d_i^\gamma )^+\right], \quad
     G_i(\mu_i) = \Exp_{v_i, d_i} \left[g_i \right], \\
     Z_i(\mu_i) &= \Exp_{v_i, d_i} \left[d_i^\gamma \mathds{1}\{\Delta v_i  > \mu_i d_i^\gamma \} \right], 
     \quad
     L_i(\mu_i) = Z_i(\mu_i) - G_i(\mu_i),
    \end{aligned}
\end{equation}
where the expectation is with respect to the stationary distributions $\mathcal{V}_i$ and $\mathcal{D}_i.$ 
We make two observations on the definition of the expected dual objective.
First, in line with the multiplier update~\eqref{equ:G-mu},
the dual objective $\Psi_i^0$ that strategy $K$ aims to maximize artificially considers a target expenditure rate of zero.
Notice however that $\rho_i$ appears in the expected optimal dual objective $\Psi_i^H(\mu_i) = \Exp_{v_i, d_i} \big[\delta_i^H \big(v_{i}, d_{i}^\gamma, \mu_i\big)\big]$ of Problem~\eqref{equ:G_stat_comp_hindsight-dual-2}.
For this reason, we require the initial budget to grow sublinearly with the time horizon,
so as to control the difference between $\Psi_i^0$ and $\Psi_i^H.$
\begin{assumption}[Initial Budget $k_{i,1}(T)$] \label{ass:G_stat_comp_limit_T}
    The initial budget $k_{i,1}$ is a function of $T$ satisfying $\lim_{T\to \infty} k_{i,1}(T) /T = 0.$
\end{assumption}

Second, 
Equation~\eqref{equ:G-dual-objective} is defined in terms of the actual gain $g_i$ rather than the maximum possible gain $\frac{\gamma}{N} \: d^\gamma_i$ used in the relaxed problems~\eqref{equ:G_lowerbound_cost_H}--\eqref{equ:G_stat_comp_hindsight-dual-2}.
This discrepancy implies yet another gap between $\Psi_i^0$ and $\Psi_i^H$,
since
if agent $i$ is not among the auction winners, it can 
gain more by bidding $d^{\gamma+1}_i < b_i \leq d^\gamma_i$.
For this reason, it is convenient to define the \emph{residual gain} 
$\hat{\varepsilon} = \frac{\gamma}{N} \Exp_{d_i}\big[d^\gamma_i - d^{\gamma+1}_i \big]$, that is, the expected maximum additional gain that agent $i$ can get by becoming the price setter. 

We moreover denote by $\mu^{\star0}_i >0$ the \emph{stationary multiplier} that satisfies $L_i \big(\mu_i^{\star0}\big)=0$ and causes the expected expenditure to equal the expected gain.
Finally, we define the notion of \emph{hitting time} as
\begin{equation}\label{equ:G-hit-time}
\small
    \begin{aligned}
        \mathscr{T}_i &= \min \left\{\mathscr{T}_i^k, \mathscr{T}_i^{\underline{\mu}}, \mathscr{T}_i^{\overline{\mu}} \right\},  \quad \text{where }
         \mathscr{T}_i^k = \mathop{\mathrm{argmax}}_{t\in[T]} \left\{ \forall s \in [t], k_{i,s} \geq \Delta/\underline{\mu} \right\}, \\
    \mathscr{T}_i^{\underline{\mu}} &= \mathop{\mathrm{argmax}}\limits_{t\in[T]} \left\{ \forall s \in [t], \mu_{i,s} \geq \underline{\mu} \right\}, 
    \quad
    \mathscr{T}_i^{\overline{\mu}} = \mathop{\mathrm{argmax}}\limits_{t\in[T]} \Big\{ \forall s \in [t], \mu_{i,s} \leq \overline{\mu} \Big\}.
    \end{aligned}
\end{equation}
This is the latest time step which guarantees that $b_{i,t} = \Delta v_{i,t} / \mu_{i,t}$
in~\eqref{equ:G-b} for any valuation $v_{i,t} \in [0,1]$.
By definition, the hitting time is a stricter notion than the \emph{budget depletion time} $\mathscr{T}_i^k$ used in the standard setting with no budget gains (c.f. Appendix~\ref{sec:A}).
This modification is needed since the projection occurs in the bid~\eqref{equ:G-b} instead of the multiplier update~\eqref{equ:G-mu} in strategy $K$, and in turn requires the following additional assumption to establish our result.

\begin{assumption}[Control of Hitting Time] \label{ass:G_stat_comp_hitting_time_main}

    The following 
    holds:
    \begin{enumerate}[label=\ref{ass:G_stat_comp_hitting_time_main}.\arabic*]
      
        \item Distribution $\mathcal{V}_i$ has support in $\big[\underline{v_i}, 1\big]$, where $0 < \underline{v_i}< 1$;\label{ass:G_stat_comp_hitting_time-1}
        
        \item Distribution $\mathcal{D}_i$ has support in $\big[\underline{d_i}, \overline{d_i}\big]^2$, where $0 < \underline{d_i} < \overline{d_i}$.\label{ass:G_stat_comp_hitting_time-2}
        
    \end{enumerate}
\end{assumption}

We are now ready to state the main result regarding the asymptotic optimality of strategy $K$ under stationary competition.

\begin{theorem}[Asymptotic Optimality under Stationary Competition] \label{thm:G_stat_comp}
    There exists a constant $C \in \mathbb{R}_+$ such that the average expected regret of an agent $i \in \mathcal{N}$ for following strategy $K$ in the stationary competition setting satisfies
\begin{equation*} 
\begin{aligned}
\small
      &\frac{1}{T} \Exp_{\bm{v_i}, \bm{d_i}} \left[ \mathcal{C}_i^{K} (\bm{v_i}, \bm{d_i}) - \mathcal{C}_i^H (\bm{v_i}, \bm{d_i})  \right] \\
        &\leq   C \left(\epsilon  +  \frac{ 1 }{\epsilon T} + \frac{ k_{i,1}  }{ T}
        + \frac{\mathbb{E}_{\bm{v_i}, \bm{d_i}} \left[T - \mathscr{T}_i  \right]}{T} + \hat{\varepsilon}\right)
    \end{aligned}
\end{equation*}
Moreover, for suitably chosen parameters,
strategy $K$ asymptotically converges to an $O(\hat{\varepsilon})$-neighborhood of the optimal expected cost with the benefit of hindsight, i.e., 
\begin{equation*}
    \small
    \lim_{T \to \infty} \: \frac{1}{T} \: \Exp_{\bm{v_i},\bm{d_i}} \left[ \mathcal{C}_i^K (\bm{v_i},\bm{d_i}) - \mathcal{C}_i^H (\bm{v_i},\bm{d_i}) \right] = O(\hat{\varepsilon}).
\end{equation*}
\end{theorem}

The full statement of Theorem~\ref{thm:G_stat_comp} with the required technical assumptions is included in Appendix~\ref{sec:ass_G_stat}.
The detailed proof of the theorem is moreover included in Appendix~\ref{sec:proof_G_stat_comp}, and we provide here a sketch of the proof.
The average expected optimal cost with the benefit of hindsight is lower-bounded by the maximum of the expected dual objective 
$\Psi_i^H \big(\mu_i^{\star H}\big)$ derived from the Lagrangian dual problem in Equation~\eqref{equ:G_stat_comp_hindsight-dual-2}.
Meanwhile, the average stationary competition cost of strategy $K$ is upper-bounded in terms of $\Psi^0_i \big(\mu_i^{\star0} \big)$ through a Taylor expansion in $\mu_i^{\star0}$.
Controlling this upper bound
requires to control a) the hitting time such that the budget is depleted, or the multiplier leaves its bounds, only towards the end of the horizon $T$,
as assumed in the definition of $\Psi^0_i$; and b) the expected distance of the multiplier iterates to the optimal multiplier $\mu_{i,t} - \mu^{\star0}_i$ such that $\mu_{i,t}$ converges asymptotically to $\mu^{\star0}_i$.
These two objectives are achieved with a suitable choice of the gradient step size $\epsilon .$
Finally, we bound the difference of dual objectives $\Psi^0_i \big(\mu_i^{\star0} \big) - \Psi_i^H \big(\mu_i^{\star H} \big)$ in terms of the target expenditure rate $\rho_i$ and the residual gain $\hat{\varepsilon}.$

This last step constitutes a fundamental difference to the standard setting with no budget gains, c.f., Theorem~\ref{thm:A_stat_comp} in Appendix~\ref{sec:A}.
Indeed, Theorem~\ref{thm:G_stat_comp} differs from
Theorem~\ref{thm:A_stat_comp} in that it does not establish an asymptotic average regret of zero but rather in the order of the residual gain $\hat{\varepsilon} 
$.
The intuition for the $O(\hat{\varepsilon})$ term is that, without the benefit of hindsight, agent $i$ will always regret not setting the price to the a priori unknown maximum  $d_i^\gamma$, instead of $d_i^{\gamma+1}$, when losing the auction.
However, in cases where $d^\gamma_i - d^{\gamma+1}_i$ correspond to the distance between two adjacent independent samples from a common distribution $\mathcal{D}_{-i}$, the residual gain $\hat{\varepsilon}$ diminishes as the number of samples, or equivalently $N-1$, grows. For a large number of agents $N$, the residual gain $\hat{\varepsilon}$ is hence expected to be modest.

Another important difference to Theorem~\ref{thm:A_stat_comp} is that 
its asymptotic guarantee 
requires the initial budget $k_{i,1}$ to grow sublinearly rather than linearly with respect to the time horizon $T$, c.f. Assumption \ref{ass:G_stat_comp_limit_T}. 
This effectively ensures that the target expenditure rate $\rho_i$ tends to zero 
and that strategy $K$ maximizes the correct expected dual objective $\Psi_i^0$. 
However, the initial budget $k_{i,1}$ cannot be kept constant.
We show in the proof of Theorem~\ref{thm:G_stat_comp} that the multiplier under strategy $K$ is bounded at all time steps $t \in [T]$ by $\mu_{i,t} \leq \mu_{i,1} + \epsilon k_{i,1}.$ 
Since, as in the standard setting with no budget gains, it is required that $\epsilon_i $ diminishes to zero asymptotically (c.f. Assumption~\ref{ass:A_stat_comp_limit_T} in Appendix~\ref{sec:A}), $\mu_{i,t}$ would not be able to converge to any value greater than the initial value $\mu_{i,1}$ with a constant $k_{i,1}$.
We refer to this phenomenon as the \emph{vanishing box problem}.

Finally, 
the proof of Theorem~\ref{thm:G_stat_comp} 
requires establishing the sublinear growth rate of $\Exp_{\bm{v_i}, \bm{d_i}} \left[T - \mathscr{T}_i\right]$ with respect to $T$, which is harder to obtain since the hitting time $\mathscr{T}_i$ defined in Equation~\eqref{equ:G-hit-time} is a stricter notion than the budget depletion time used 
for Theorem~\ref{thm:A_stat_comp}.
This difficulty is addressed by Assumption~\ref{ass:G_stat_comp_hitting_time_main}. By introducing a strictly positive minimum valuation $\underline{v_i}$, Assumption~\ref{ass:G_stat_comp_hitting_time-1} ensures that agent $i$ always wins when the multiplier is close to $\underline{\mu}$, whereas the strictly positive minimum competing bid $\underline{d_i}$ introduced in Assumption ~\ref{ass:G_stat_comp_hitting_time-2} ensures that the agent always loses when the multiplier is close to $\overline{\mu}$. In practice, Assumption~\ref{ass:G_stat_comp_hitting_time_main} implies that agents always have a need to
participate 
in the auction.

\subsection{Convergence under Simultaneous Learning} \label{sec:G_sim_lear}

In this section, we take the next step towards our main goal of establishing that strategy $K$ constitutes an approximate Nash equilibrium when adopted by all agents.
Namely, we establish that the learning dynamics \emph{converge in the simultaneous learning setting} in which all agents follow strategy $K$, denoted by joint strategy profile $\bm{K}$.
The exact notion of convergence considered is presented in the main result of the section, Theorem~\ref{thm:G_sim_lear_cv}.

Before stating this result, we first adapt our previous definitions to the multi-agent setting.
Let $\bm{\mu}_t \in \mathbb{R}_+^N$ be the \emph{multiplier profile} stacking the multipliers $\mu_{i,t}$ of all agents $i \in \mathcal{N}$.
We extend the \emph{expected dual objective}, \emph{expected gain}, \emph{expected expenditure} and the \emph{expected loss} respectively as 
\begin{equation*}
\small
\begin{aligned}
    \Psi_i^0(\bm{\mu}) &= \Exp_{\bm{v}} \left[v_i  -\mu_i  g_i - (\Delta v_i - \mu_i d_i^\gamma )^+\right],
    \quad
     G_i(\bm{\mu}) = \Exp_{\bm{v}} \left[g_i \right],\\
     Z_i(\bm{\mu}) &= \Exp_{\bm{v}} \left[d_i^\gamma \mathds{1}\{\Delta v_i  > \mu_i d_i^\gamma \} \right], 
     \quad
     L_i(\bm{\mu}) = Z_i(\bm{\mu}) - G_i(\bm{\mu}),
    \end{aligned}
\end{equation*}
Compared to~\eqref{equ:G-dual-objective}, the expectation is now with respect to the profile of valuations $\bm{v} \sim \bm{\mathcal{V}}.$
Indeed, both the competing bid $d_i^\gamma = \gamma\textsuperscript{th}\mbox{-}\max_{j:j\neq i} \{ \Delta v_j / \mu_j\}$ and the auction price $p^{\gamma+1} = \gamma+1\textsuperscript{th}\mbox{-}\max_{j} \{ \Delta v_j / \mu_j\}$ are functions of $\bm{v}$ and $\bm{\mu}.$
We will aim to show that $\bm{\mu}_t$ converges to a \emph{stationary multiplier profile}, which is a multiplier profile $\bm{\mu^{\star0}} \in \mathbb{R}^N_{>0}$ satisfying $L_i \big(\bm{\mu^{\star0}} \big) = 0$ for all agents $i \in \mathcal{N}$.
This multiplier profile is stationary in the sense that in expectation, update rule~\eqref{equ:G-mu} will yield $\mu^\star_{i,t+1} = \mu^\star_{i,t}$ for all agents $i$, since the expected losses $Z_i \big(\bm{\mu^{\star0}}\big)$ equal the expected gain $G_i \big(\bm{\mu^{\star0}} \big)$.

Notice that stationary multipliers $\bm{\mu^{\star0}}$ are numerous: if multiplier $\bm{\mu^{\star0}}$ is stationary, so is $\eta \bm{\mu^{\star0}}$ for all $\eta > 0$, as the expenditures $Z_i$ and gains $G_i$ are equally scaled by $1 / \eta$.
This property of $\bm{\mu^{\star0}}$ is novel to the karma setting, 
as with no budget gains a unique scale is fixed by the target expenditure rates $\rho_i$.
To fix the unique $\bm{\mu^{\star0}}$ that strategy profile $\bm{K}$ converges to, 
the projection is moved from the multiplier update~\eqref{equ:G-mu} to the bid~\eqref{equ:G-b} in strategy $K$, as compared to standard adaptive pacing (c.f. strategy $A$ in Appendix~\ref{sec:A}).
This modification, combined with a shared gradient step size $\epsilon$ for all agents, implies the following property
\begin{equation} \label{equ:G-mu-hyperplane}
\small
    \begin{aligned}
    \sum_{i \in \mathcal{N}} \mu_{i,t+1} = \sum_{i \in \mathcal{N}} \mu_{i,t} + \epsilon \sum_{i \in \mathcal{N}} \left(z_{i,t} - g_{i,t} \right) \stackrel{\textup{(a)}}{=} \sum_{i \in \mathcal{N}} \mu_{i,t}, \\
    \text{hence }
     \bm{\mu_t} \in \bm{H_{\mu_1}} = \left\{\bm{\mu} \in \mathbb{R}^N \left\lvert \: \sum_{i \in \mathcal{N}} \left(\mu_i - \mu_{i,1}\right) = 0 \right. \right\}.
     \end{aligned}
\end{equation}
Property~\eqref{equ:G-mu-hyperplane} anchors the scale of $\bm{\mu^{\star0}}$ that is feasible under strategy profile $\bm{K}$ to the initial multiplier profile $\bm{\mu_1}$, and implies that the \emph{average multiplier} $\mu_m = \sum_{i \in \mathcal{N}} \mu_{i,1} / N$ is preserved over time.
Notice that Property~\eqref{equ:G-mu-hyperplane}
would not hold if the projection was in the multiplier update~\eqref{equ:G-mu}: intuitively, a projection there would cause agent $i$ to `forget' part of the history of expenses and gains, and affect the convergence of the whole population as a consequence by shifting the hyperplane of feasible profiles.

In the standard setting with no budget gains, it is common to assume strong monotonicity of the expected expenditure $\bm{Z}$, c.f.
\cite{balseiro2019learning} which shows that it is implied by a \emph{diagonal strict concavity} condition~\cite{rosen1965existence}, and that it always holds in symmetric settings.
In our karma setting with budget gains, 
the natural extension is to assume strong monotonicity of the expected loss $\bm{L}$
which effectively replaces the expected expenditure $\bm{Z}$.
This adaptation is however not straightforward, as the multiple zeros of $\bm{L}$ on $\bm{U} := \prod_{i\in\mathcal{N}} \left(\underline{\mu}, \overline{\mu}\right)$ would immediately break the property.
Instead, we use Property~\eqref{equ:G-mu-hyperplane} and restrict our monotonicity requirement to multipliers lying in the hyperplane $\bm{H_{\mu_1}}$.
\begin{assumption}[Monotonicity]\label{ass:G_sim_lear_strong_mono}
    The expected loss $\bm{L}$ is $\lambda$-strongly monotone over $\bm{U}\cap \bm{H_{\mu_1}}$ with parameter $\lambda >0$, i.e., for all $\bm{\mu}, \bm{\mu'} \in \bm{U}\cap \bm{H_{\mu_1}}$ 
    , it holds that $(\bm{\mu} - \bm{\mu}')^\top (\bm{L}(\bm{\mu}) - \bm{L}(\bm{\mu'}) ) \leq - \lambda \Vert \bm{\mu} - \bm{\mu'}\Vert_2^2.$
\end{assumption}
As shown in Appendix~\ref{sec:A_sim_lear_unique_mu^*}, Assumption~\ref{ass:G_sim_lear_strong_mono} ensures that the stationary multiplier profile $\bm{\mu^{\star0}}$ is unique up to a multiplicative constant if it exists.
With these preliminaries, we are ready to state the main result of this section regarding the asymptotic convergence of strategy profile $\bm{K}$,
both with respect to the multiplier profile iterates $\bm{\mu_t}$ and the strategic competition costs $C_i^{\bm{K}}$ of all agents $i \in \mathcal{N}$ defined in Equation~\eqref{equ:C-strategic}.

\begin{theorem}[Convergence under Simultaneous Learning] \label{thm:G_sim_lear_cv}
There exist constants $C_1$ and $C_2\in \mathbb{R_+}$ such that the average expected distance to the stationary multiplier profile $\bm{\mu^{\star0}} \in \bm{H_{\mu_1}}$ and the strategic competition cost of strategy profile $\bm{K}$ for any agent $i \in \mathcal{N}$ satisfy respectively
\begin{equation*}
\small
\begin{aligned}
    \frac{1}{T} \sum_{t=1}^T \Exp_{\bm{v}} \left[\left\Vert \bm{\mu_t} - \bm{\mu^{\star0}} \right\Vert_2^2 \right] 
     \leq  C_1 N \left( \epsilon + \frac{1}{\epsilon T} 
    +  \frac{\Exp_{\bm{v}} \left[ T - \underline{\mathscr{T}}  \right]}{T} \right),\\
     \frac{1}{T} \mathcal{C}_i^{\bm{K}} - \Psi_i^0 \big(\bm{\mu^{\star0}}\big)
     \leq C_2 \left( N \left( \epsilon^{1/2} +   \frac{1}{\epsilon T}\right) + \frac{ \Exp_{\bm{v}} \left[T - \underline{\mathscr{T}} \right]}{T}\right).
     \end{aligned}
\end{equation*}
Moreover, for suitably chosen parameters of strategy profile $\bm{K}$,
the multiplier profile $\bm{\mu_t}$ converges in expectation to the stationary 
profile $\bm{\mu^{\star0}}$, and the average strategic competition cost $\mathcal{C}_i^{\bm{K}}$ converges to the expected dual objective $\Psi^0_i \big(\bm{\mu^{\star0}} \big)$ for all agents $i \in \mathcal{N}$, i.e., 
$\lim\limits_{T \to \infty} \: \frac{1}{T} \sum_{t=1}^T \Exp_{\bm{v}} \left[ \big\Vert \bm{\mu_t} - \bm{\mu^{\star0}} \big\Vert_2^2 \right] =  \lim\limits_{T \rightarrow \infty} \: \frac{1}{T} \: \mathcal{C}_i^{\bm{K}} - \Psi_i^0 \big(\bm{\mu^{\star0}} \big) = 0. $
\end{theorem}

The full statement of Theorem~\ref{thm:G_sim_lear_cv} with the required technical assumptions is included in Appendix~\ref{sec:ass_G_sim}.
The detailed proof 
is moreover included in Appendix~\ref{sec:G_sim_lear_cv_proof} and mostly follows from similar arguments as Theorem~\ref{thm:G_stat_comp}.
The main step from the first to the second bound 
is to show that the profiles of expected dual objectives $\bm{\Psi^0}$ and expected losses $\bm{L}$ are Lipschitz continuous in $\bm{\mu}$ on the entire compact set $\bm{U}$, which can be guaranteed by the absolute continuity of the valuations. 
As before, the main difficulty in comparison to the standard setting with no budget gains, c.f. Theorem~\ref{thm:A-sim-lear-converge} in Appendix~\ref{sec:A}, lies in ensuring that the expectation $\Exp_{\bm{v}} \left[ T - \underline{\mathscr{T}}  \right]$ 
grows sublinearly with respect to $T.$
This challenge is addressed analogously as in Section~\ref{sec:G_stat_comp}:
Assumption~\ref{ass:G_stat_comp_hitting_time-1}
deterministically guarantees that $\mathscr{T}_i^{\overline{\mu}}= T$ as any agent close to $\overline{\mu}$ will lose the auction and transition away from $\overline{\mu}$. 
A similar deterministic guarantee cannot be derived at the lower bound $\underline{\mu}$, however, due to the preservation of the average multiplier $\mu_m .$
For this reason, technical Assumption~\ref{ass:G-sim-learn-hitting-time-4}, c.f. Appendix~\ref{sec:ass_G_sim}, imposes a probabilistic condition on the lower bound $\underline{\mu}$. 
This assumption is discussed further in Section~\ref{sec:disc}.

\subsection{Approximate Nash Equilibrium in Parallel Auctions}\label{sec:G_eps_NE}

In this section, we finally combine the results of the previous two sections to achieve the main goal of establishing that the profile of adaptive karma pacing strategies $\bm{K}$ constitutes an approximateNash equilibrium under suitable conditions.

Notice that one cannot immediately conclude that strategy profile $\bm{K}$ constitutes an approximate Nash equilibrium despite the previously established asymptotic guarantee on the strategic competition costs. 
Namely, Theorem~\ref{thm:G_sim_lear_cv} ensures that the multiplier profile converges asymptotically to $\bm{\mu^{\star0}}$ under strategy profile $\bm{K}$.
Therefore, agent $i$'s distribution of competing bids becomes stationary, and we showed in the proof of Theorem~\ref{thm:G_stat_comp} that $\Psi_i^0 \big(\bm{\mu^{\star0}} \big)$ lower bounds any \emph{stationary competition cost}.
However, agent $i$ could potentially improve its cost by unilaterally deviating to a strategy $\beta_i \neq K$ that causes non-convergence of $\bm{\mu_t}$ and violation of the stationary competition assumption.
For this reason, following~\cite{balseiro2019learning}, we consider a natural extension of our setting in which there are multiple \emph{parallel auctions} causing the effect of any single agent on the multipliers of others to become negligible as the number of agents grows.

\smallskip

\textbf{Parallel Auctions.}
Let there be $M \geq 1$ auctions that are held in parallel at each time step $t \in [T]$.
The number $M$ could represent different priority roads, or the same road accessed at different times of the day, or a combination thereof.
Each agent $i \in \mathcal{N}$ participates in one auction $m_{i,t} \in [M]$ per time step, where $m_{i,t}$ is drawn independently across agents and time from a fixed distribution $\bm{\pi_i} = \left(\pi_{i,m}\right)_{m \in [M]}$; each $\pi_{i,m}$ denotes the probability for agent $i$ to participate in auction $m$. We adapt the definition of competing bids accordingly as $d_{i,t}^\gamma = \gamma\textsuperscript{th}\mbox{-}\max_{j:j\neq i} \left\{ \mathds{1}\{ m_{j,t} = m_{i,t}\}  b_{j,t} \right\}$.
Finally, we consider that the aggregate payment of all $M$ auctions gets redistributed uniformly, leading to karma gains $g_{i,t} = \frac{\gamma}{N} \sum_{m \in [M]} p_{m,t}^{\gamma+1}$, where $p_{m,t}^{\gamma+1}$ is the price of auction $m$ defined as \mbox{$p_{i,t}^{\gamma+1} = \gamma+1\textsuperscript{th}\mbox{-}\max_{i\in\mathcal{N}} \{ \mathds{1}\{ m_{i,t} = m\}  b_{i,t} \}$}.
This aggregate redistribution scheme is advantageous over redistributing the payment of each auction among its agents, since the aggregation restricts the influence of a single agent over the gains, and thereby the multipliers, of others.
The distributions $\left(\bm{\pi_i}\right)_{i \in \mathcal{N}}$ yield \emph{matching probabilities} $\bm{a_i} = \left(a_{i,j}\right)_{j\neq i}$, where $a_{i,j} = \mathbb{P}\{m_{j} = m_{i}\}$ denotes the probability that agent $j$ is matched in the same auction as agent $i.$
It is straightforward to show that the previous Theorems~\ref{thm:G_stat_comp} and~\ref{thm:G_sim_lear_cv} also hold in the extended parallel auction setting.

\begin{theorem} [Approximate Nash Equilibrium] \label{thm:G_eps_NE}~
 There exists a constant $C\in \mathbb{R_+}$ such that each agent $i \in \mathcal{N}$ can decrease its average strategic competition cost by deviating from strategy $K$ to any strategy $\beta_i \in \mathcal{B}$ by at most
\begin{equation*}
\small
\begin{aligned}\label{equ:G-Nash-bound}
     \frac{1}{T} \left(\mathcal{C}_i^{K} - \mathcal{C}_i^{\beta_i, \bm{K_{-i}}} \right)
     &\leq C \Bigg( \left( \Vert \bm{a_i}\Vert_2 + \frac{M\gamma}{N} \right) \bigg( \sqrt{N\epsilon} \left(1 + \frac{1}{\epsilon^{3/2} T}\right) + \Vert \bm{a_i}\Vert_2 \\
     &+ \frac{\gamma}{\sqrt{N}}  \bigg) 
    + \left( \frac{\gamma}{N} + \frac{\overline{k_1}}{T} \right)
    +   \left( \frac{\overline{k_1}}{T} + \frac{M\gamma}{N} \right) 
    \frac{\mathbb{E}_{\bm{v}, \bm{m}} \left[T - \underline{\mathscr{T}}   \right]}{T} \Bigg)
\end{aligned}
\end{equation*}
Moreover, for suitably chosen parameters,
strategy profile $\bm{K}$ constitutes an approximate Nash equilibrium, i.e., it holds for all agents $i \in \mathcal{N}$ that
$ \lim\limits_{T, N, M \to \infty} \frac{1}{T} \left( \mathcal{C}_i^{\bm{K}} - \inf_{\beta_i \in \mathcal{B}} \mathcal{C}_i^{\beta_i, \bm{K_{-i}}} \right) = 0.$
\end{theorem}

The full statement of Theorem~\ref{thm:G_eps_NE} with the required technical assumptions is included in Appendix~\ref{sec:ass_G_NE}.
The detailed proof of Theorem~\ref{thm:G_eps_NE},
included in Appendix~\ref{sec:G_eps_NE_proof},
involves lower-bounding the average expected cost under strategy $\beta_i$ in terms of the optimal expected dual objective $\Psi_i^0 \big(\bm{\mu^{\star0}} \big)$, and showing that asymptotically agent $i$ cannot affect the multiplier profile of the other agents which converges to $\bm{\mu^{\star0}}$.
Competition hence becomes stationary, for which strategy $K$ is optimal.
While the proof follows a similar structure as in the standard setting with no budget gains, notice however that the bound in Theorem~\ref{thm:G_eps_NE}
is substantially different from its counterpart with no budget gains, 
c.f. Theorem~\ref{thm:A_eps_NE} in Appendix~\ref{sec:A}, and requires adapting the technical assumptions in order to establish the asymptotic guarantee.

This achieves the main goal of our analysis.
For the class of karma mechanisms with redistribution of payments, we have devised the simple adaptive karma pacing strategy $K$ and provided conditions in which it constitutes an approximate Nash equilibrium.

\section{Discussion}\label{sec:disc}

In this paper, we devised a learning strategy, called \emph{adaptive karma pacing}, that learns to bid optimally in karma mechanisms in which payments are redistributed in every time step.
This simple strategy constitutes an $\varepsilon$-Nash equilibrium in large populations, and can hence be effectively employed to provide decision support,
which is an important step toward the practical implementation of these mechanisms.

\smallskip

\emph{Welfare implications.}
Adaptive karma pacing results in bids that are linear in the agents' private valuations.
In a symmetric setting with a common valuation distribution, the simultaneous adoption of adaptive karma pacing converges asymptotically to the same multiplier for all agents, and bids become \emph{truth-revealing}.
Therefore, resources will be allocated efficiently to the agents with the highest private valuation, and without the need to use money.
In a non-symmetric setting, adaptive karma pacing will be similarly efficient among identical agents,
and further interpersonal comparability assumptions~\cite{roberts1980interpersonal} are needed to assess its welfare properties across agents with different valuation distributions.

\smallskip

\emph{Comparison to adaptive pacing in monetary auctions.}

Owing to the redistribution of karma, the possibility to gain karma and manipulate that gain, as well as the preservation of karma per capita, lead to several conceptual and technical novelties.
These novelties include the inevitable regret term 
in Theorem~\ref{thm:G_stat_comp} stemming from the possibility of manipulating the gain upon losing;
the impossibility for agents to simultaneously track their target expenditure rates and deplete their budgets due to the preservation of karma;
the non-uniqueness of stationary multiplier profiles, which requires moving the multiplier projection from the multiplier update to the bid in Algorithm~\ref{alg:G};
and, as a consequence, the need to control a stricter notion of hitting time in the proofs, which in addition to the budget depletion time considered in monetary auctions includes the first time a multiplier is projected.

\smallskip

\emph{Discussion of assumptions.}
Our main results require a number of technical assumptions which, we argue, are not highly restrictive.
These assumptions can be categorized as follows.
The assumptions on \emph{valuation and competing bid distributions} (e.g. the absolute continuity of valuations
and Assumption~\ref{ass:A_stat_comp_imply_differentiability+concavity}) are mild continuity and differentiability assumptions that are common in the literature, including in the standard monetary setting~\cite{balseiro2019learning}.
Nonetheless, we performed numerical experiments with discrete and unbounded valuation distributions to test the robustness of our results, c.f. Figures~\ref{fig:non_cont_v_perfs_vs_H} and~\ref{fig:non_cont_v_cv_mu} in Appendix~\ref{sec:numerical}, which show examples in which Theorems~\ref{thm:G_stat_comp} and Theorem~\ref{thm:G_sim_lear_cv} hold under discrete uniform and geometric distributions.

The assumptions on \emph{input parameters} of the adaptive karma pacing strategy (e.g., 
Assumptions~\ref{ass:G-parameters} and \ref{ass:G-sim-learn-parameters}) can be satisfied by design and/or tuning.
Generally, setting the initial multiplier $\mu_{i,1}$ close to the center of a sufficiently low $\underline{\mu}$ and a sufficiently high $\overline{\mu}$, and using a sufficiently small gradient step size $\epsilon$, suffices to satisfy these assumptions.
These assumptions are provided for technical completeness and to provide insight on the structure of the problem.

On the other hand, assumptions requiring to \emph{vary parameters asymptotically} (e.g., Assumptions
~\ref{ass:G_stat_comp_limit_T}, ~\ref{ass:A_stat_comp_limit_T} and \ref{ass:G_eps_NE_init}) are less natural to interpret in practice since typically the time horizon $T$ and number of agents $N$ are fixed by the setting.
For this reason, we provided bounds in all our theorems that give finite time and population guarantees.

Finally, assumptions needed to \emph{control the hitting time} (Assumptions~\ref{ass:G_stat_comp_hitting_time_main}, \ref{ass:G_stat_comp_hitting_time}, \ref{ass:G-sim-learn-hitting-time} and \ref{ass:G-eps-NE-hitting-time}) arise
from our proof technique seeking to deterministically guarantee that the multipliers $\mu_{i,t}$ will never reach their bounds $\underline{\mu}$ and $\overline{\mu}.$
In fact, we performed numerical experiments verifying that the hitting time quickly approaches the end of the time horizon under many parameter combinations that extend beyond those assumed, c.f. Figure~\ref{fig:hit_time} in Appendix~\ref{sec:numerical}.
However, relaxing these assumptions requires a continuous-space Markov chain-based analysis, which provides a promising avenue for future work.

\smallskip

\emph{Future work.}
We conclude with a discussion of future research directions.
We hope that our paper inspires future interest in karma mechanisms with redistribution or other forms of karma gains.
In this paper, we studied a simple uniform redistribution scheme; however, there are rich possibilities to design karma gains that have not yet been fully explored, including redistributing karma payments partially rather than fully.
Note that our analysis does not directly extend to such settings, as the primary challenge lies in identifying appropriate assumptions to establish the convergence of the mean squared error, similar to the approach in Section~\ref{sec:G_helpful_lem}.
Moreover, in this 
class of karma mechanisms, it will be important to address the \emph{vanishing box problem}, c.f. Section~\ref{sec:G_stat_comp}, which prevents convergence from being achieved with a fixed initial budget.
Towards this direction, we performed a numerical experiment using a \emph{variable gradient step size} $\epsilon_t$, c.f. Figure~\ref{fig:fixed_budget} in Appendix~\ref{sec:numerical}, which provides preliminary evidence that asymptotic convergence can be achieved even with a fixed initial budget.
Therefore, extending the present analysis for variable gradient step sizes provides an exciting avenue for further investigation.
Finally, armed with the simple adaptive karma pacing strategy, we hope to see many practical implementations of karma mechanisms that address societal challenges involving scarce resource allocation.

\bibliographystyle{ACM-Reference-Format}
\bibliography{biblio_karma}

\appendix

\section{Related Work}\label{sec:related_work}

Our paper is situated within two main research fields: non-monetary mechanism design for repeated allocations; and learning in monetary auctions.
To the extent of our knowledge, learning in repeated non-monetary auctions has not been studied before.

\paragraph{Non-monetary mechanism design for repeated allocations.}

The celebrated Gibbard-Satterhwhite impossibility theorem~\cite{gibbard1973manipulation,satterthwaite1975strategy}, which states that it is generally impossible for mechanisms to be jointly strategyproof and non-dictatorial when no monetary transfers are allowed, has triggered a long line of works on \emph{non-monetary mechanism design}~\cite{nisan2007algorithmic} and \emph{stable matching}~\cite{roth1990two} that aim at overcoming this impossibility.
The classical setting in these areas regards static, one-shot allocations, e.g., college admissions~\cite{gale1962college}, school choice~\cite{abdulkadirouglu2003school}, organ donations~\cite{roth2004kidney}, or course allocations~\cite{budish2012multi}.
Many works in these areas employ \emph{pseudomarkets}~\cite{hylland1979efficient,pycia2023pseudomarkets,budish2011combinatorial}, in which users are issued individual budgets of artificial currency with which they acquire resources at the market clearing prices (those are, the prices at which all resources are allocated and all budgets are depleted).
Although it is generally difficult to determine the market clearing prices in large pseudomarkets~\cite{budish2011combinatorial}, these markets are non-dictatorial and \emph{approximately strategyproof in the large}: the potential gain from misreporting preferences diminishes as the number of users increases.

Relatively fewer works have investigated non-monetary \emph{repeated} allocations, however, despite being relevant to many important real-world problems, including food donations~\cite{prendergast2022allocation}, babysitting services~\cite{johnson2014analyzing}, computation resources~\cite{vuppalapati2023karma}, and public goods (e.g., transportation~\cite{elokda2024carma} or water~\cite{van2015equity}).
Some of the pioneering works in this direction rely on the abstract notion of \emph{promised utilities}~\cite{guo2020dynamic,balseiro2019multiagent}, which elicit truthfulness by letting users trade-off immediate utility for future promised utility, in a manner similar to trading-off utility for money. However, computing the correct utility to promise each user requires central knowledge of all users' valuation distributions, and the precise mechanism by which these future promises are implemented is not specified.
A more practical alternative is to employ artificial currency, similarly to pseudomarkets in the static setting, but with the distinguishing feature that users bid currency for resources in an online fashion, without perfect knowledge of their (potentially infinitely repeated) future resource valuations~\cite{guo2009competitive,gorokh2021monetary,gorokh2021remarkable,siddartha2023robust,elokda2023self}.

Most recent theoretical works on artificial currency for repeated allocations consider finite repetitions and an initial endowment of currency that cannot be replenished over the course of the allocations~\cite{gorokh2021monetary,gorokh2021remarkable,siddartha2023robust}.
In \cite{gorokh2021monetary}, a black-box method to convert truthful monetary mechanisms to approximately truthful artificial currency mechanisms is proposed, which however leads to complex mechanisms requiring exact central knowledge of all users' valuation distributions.
In contrast, \cite{gorokh2021remarkable,siddartha2023robust} adopt simple first-price auctions, and devise \emph{robust strategies} that guarantee their users a minimum level of utility irrespective of the bidding behaviors of others.
While robust strategies have the benefit of requiring little information or adaptation by the users, they could lead to overly conservative outcomes in non-adversarial settings.

On the other hand, when allocations are to be repeated indefinitely, many application-driven works have found it natural to keep the total amount of artificial currency in the system constant, either through peer-to-peer exchanges~\cite{vishnumurthy2003karma,friedman2006efficiency,johnson2014analyzing,elokda2023self}, or by redistributing the payments collected in each time step~\cite{prendergast2022allocation,elokda2024carma}.
This type of mechanisms has been employed successfully in practice, most famously in the context of food donations~\cite{prendergast2022allocation}, despite not yet being fully understood theoretically.
A difficulty in the theoretical analysis arises from the possibility to gain currency which leads to highly non-trivial budget dynamics.
For this purpose, \cite{elokda2023self} performs a mean field game-based analysis to guarantee the existence of a Nash equilibrium when the artificial currency takes the form of integer tokens, referred to as \emph{karma}.
An algorithm to compute the equilibrium is proposed that is however centralized in nature.

\paragraph{Learning in repeated monetary auctions}

Repeated monetary auctions with budget constraints are part of a vast body of literature that considers online decision-making, see \cite{shalev-shwartz_online_2011} for a comprehensive survey.
Some early works employ a mean-field game framework to study these auctions~\cite{iyer_mean_2011,balseiro2015repeated}, with which it is difficult to extract analytical insights.
Instead, \cite{balseiro2019learning} establishes a comprehensive framework for \emph{adaptive pacing strategies} based on online dual gradient ascent and regret analysis.
For second-price auctions with budget constraints, it is shown that adaptive pacing strategies achieve sublinear regret in stationary competition settings, maximal competitive ratio in adversarial settings, and these strategies converge to approximate Nash equilibria when simultaneously adopted and when there are many parallel auctions.
The techniques developed in~\cite{balseiro2019learning} form the basis for the present paper, in which we adopt and extend the analysis to artifical currency-based mechanisms with and without payment redistribution.

The results in~\cite{balseiro2019learning} have been extended in several directions which however all remain in the monetary realm.
Namely, \cite{gaitonde2022budget} shows that adaptive pacing strategies guarantee at least half of the optimal welfare in both first and second price auctions regardless of the convergence of their dynamics; \cite{balseiro2023best} extends the individual guarantees of dual ascent based methods to non-convex settings; and \cite{castiglioni2022online} proposes a `best of two worlds strategy' that generalizes the results in \cite{balseiro2019learning} to first price auctions.
An important delineation in these works regards the objective that individual bidders seek to maximize~\cite{balseiro2021landscape}: works considering \emph{utility maximization}~\cite{balseiro2019learning,castiglioni2022online,balseiro2023best} include monetary payment costs in the objective; whereas works considering \emph{value maximization}~\cite{gaitonde2022budget,lucier2024autobidders}, also known as \emph{liquid welfare}, do not include monetary payment costs in the objective.
However, even works considering value maximization assume that money has known value outside the auctions, since monetary payment costs still enter the optimization in constraints either on individual bids~\cite{gaitonde2022budget} or on the total expenditure~\cite{lucier2024autobidders} (also known as \emph{return of investment constraints}).
Our paper thus complements these works as we completely drop any explicit cost of losing currency in the optimization, and moreover consider that the budget increases throughout the auction campaign by redistributing payments.
\section{Artificial Currency Mechanisms with no Budget Gains}\label{sec:A}

This appendix serves the dual purpose of reviewing standard results on adaptive pacing in monetary auctions~\cite{balseiro2019learning}, as well as tailoring these results to the artificial currency setting in which the currency has no known value outside the auctions.
We follow the same notation and consider the same setting as laid out in Section~\ref{sec:model}, but consider that auction payments are not redistributed, i.e., $g_{i,t} = 0$ for all $i \in \mathcal{N}$, $t \in [T]$.
We omit the detailed proofs of the results in this appendix, since with the proposed modifications to standard adaptive pacing, the proofs follow similar arguments as in~\cite{balseiro2019learning}.

\subsection{Derivation of Adaptive Pacing}
\label{sec:A-derivation}

\textbf{Optimal Cost with the Benefit of Hindsight.}
To derive a candidate optimal bidding strategy, we follow the dual stochastic gradient ascent technique of~\cite{balseiro2019learning}.
For a fixed realization of valuations $\bm{v_i}$ and competing bids $\bm{d_i}$, agent $i$'s optimal cost with the benefit of hindsight is given by the following optimization problem
\begin{equation}\label{equ:A_stat_comp_def_hindsight}
\begin{aligned}
\small
\mathcal{C}_i^H (\bm{v_i},\bm{d_i}) = &\min_{\bm{x_i}\in \{0,1\}^T}  \sum_{t=1}^T v_{i,t}( 1 - x_{i,t} \Delta), \quad 
\mbox{s.t.} \sum_{t=1}^T x_{i,t} d_{i,t}^\gamma \leq \rho_i T.
\end{aligned}
\end{equation}
In Problem~\eqref{equ:A_stat_comp_def_hindsight}, we may use auction outcomes $x_{i,t}$ directly as decision variables rather than bids $b_{i,t}$, since given competing bids $d_{i,t}^\gamma$, $b_{i,t}$ affect the cost and the budget only through $x_{i,t} = \mathds{1} \{b_{i,t} > d_{i,t}^\gamma\}$.
 
The Lagrangian dual problem associated with Problem~\eqref{equ:A_stat_comp_def_hindsight} is
\begin{subequations}
\label{equ:A_stat_comp_hindsight-dual-all}
\small
\begin{align}
    \mathcal{C}_i^H (\bm{v_i},\bm{d_i}) 
    &\geq \delta_i^H \big(\bm{v_i},\bm{d_i^\gamma}\big)\\ 
    &:= \sup\limits_{\mu_i \geq 0} \min\limits_{\bm{x}_i \in \{0,1\}^T} \sum_{t=1}^T  x_{i,t} \left(\mu_i d_{i,t}^\gamma - \Delta v_{i,t}\right) + v_{i,t} -\mu_i \rho_i \label{equ:A_stat_comp_hindsight-dual} \\
    &= \sup_{\mu_i \geq 0} \sum_{t=1}^T  v_{i,t}  - \mu_i \rho_i  - \left(\Delta v_{i,t} - \mu_i d_{i,t}^\gamma \right)^+ \label{equ:A_stat_comp_hindsight-dual-2} \\
    &:=\sup\limits_{\mu_i \geq 0} \sum_{t=1}^T \delta_{i,t} \big(v_{i,t}, d_{i,t}^\gamma, \mu_i\big). 
\end{align}
\end{subequations}
Notice that for a fixed multiplier $\mu_i \geq 0$, the inner minimum in~\eqref{equ:A_stat_comp_hindsight-dual} is obtained by winning all auctions satisfying $\Delta v_{i,t} > \mu_i d_{i,t}^\gamma$.
This can be achieved by bidding $b_{i,t}=\Delta v_{i,t} / \mu_i$, and yields~\eqref{equ:A_stat_comp_hindsight-dual-2}.

\smallskip

\textbf{Adaptive Pacing.}
We perform a stochastic gradient ascent scheme in order to approximately solve the dual Problem~\eqref{equ:A_stat_comp_hindsight-dual-all} using online observations.
Namely, the agent considers a candidate optimal multiplier $\mu_{i,t}$ and places its bid accordingly with $b_{i,t}= \Delta v_{i,t} / \mu_{i,t}$.
It then updates $\mu_{i,t+1}$ using the subgradient given by
$
    \frac{\partial \delta_{i,t}}{ \partial \mu_{i,t}} \big(v_{i,t}, d_{i,t}^\gamma, \mu_{i,t} \big) = d_{i,t}^\gamma \mathds{1}\{b_{i,t} > d_{i,t}^\gamma \} - \rho_i 
    = z_{i,t} - \rho_i 
    $.
This yields the \emph{adaptive pacing strategy}, which is denoted by $A$ and summarized in Algorithm~\ref{alg:A}.

With respect to standard adaptive pacing in monetary auctions~\cite{balseiro2019learning}, an important difference in strategy $A$ is that the denominator in the bid~\eqref{equ:A-bid} is $\mu_{i,t}$ instead of $\mu_{i,t} + 1$.
This 
is a consequence of 
the fact that the valuation in artificial currency is not known a-priori; and could lead to a rapid depletion of the budget if $\mu_{i,t}$ becomes small during the learning process even for a short transient period.
For this reason, it is necessary to introduce the lower bound $\underline{\mu}$ in Algorithm~\ref{alg:A}.

\begin{algorithm}[ht]
\caption{\textbf{Adaptive Pacing $A$}}
\label{alg:A}

\KwInput{ Time horizon $T$, target expenditure rate $\rho_i >0$, multiplier bounds $\overline{\mu} > \underline{\mu} > 0$, gradient step size $\epsilon > 0$.}

\KwInitialize{ Initial multiplier $\mu_{i,1} \in [\underline{\mu}, \overline{\mu} ]$, initial budget $k_{i,1} = \rho_i T$.}

\For{ t=1,\dots, T:}{
\begin{enumerate}
    \item Observe the realized valuation $v_{i,t}$ and place bid 
    $ b_{i,t} = \min \left\{ \dfrac{\Delta v_{i,t}}{ \mu_{i,t}}, k_{i,t} \right\};$ \hfill\tagx[$A$-$b$]{equ:A-bid}
    \item Observe the expenditure $z_{i,t}$ and update the multiplier $\mu_{i,t+1} = P_{[\underline{\mu}, \overline{\mu} ]} \left( \mu_{i,t} + \epsilon ( z_{i,t} 
    - \rho_i) \right),$ \hfill\tagx[$A$-$\mu$]{equ:A-mu}
    
    as well as the karma budget
    $k_{i,t+1} = k_{i,t} - z_{i,t} 
    .$
\end{enumerate}}

\end{algorithm}

\subsection{Asymptotic Optimality under Stationary Competition}\label{sec:A_stat_comp}

In this section, we establish that strategy $A$ is asymptotically optimal in a stationary competition setting, where a single agent $i$ bids against competing bids $\bm{d_i}= \big(d_{i,t}^{\gamma}\big)_{t \in [T]}$  drawn independently across time from a fixed distribution $\mathcal{D}_i$.
To state the main result of this section, we must first define the \emph{expected dual objective} and \emph{expected expenditure} when agent $i$ follows strategy $A$.
For a fixed multiplier $\mu_i > 0$, these two quantities are respectively given by 
\begin{equation}\label{equ:A-stat-comp-dual}
\small
\begin{aligned}
    \Psi_i(\mu_i) &= \Exp_{v_i, \: d_i} \left[v_i  -\mu_i \rho_i - \left(\Delta v_i - \mu_i d_i^\gamma \right)^+\right],
    \\
    Z_i(\mu_i) &= \Exp_{v_i, \: d_i} \left[d_i^\gamma \: \mathds{1} \left\{\Delta v_i  > \mu_i d_i^\gamma \right\} \right],
    \end{aligned}
\end{equation}
where the expectation is with respect to the stationary distributions $\mathcal{V}_i$ and $\mathcal{D}_i^\gamma$.
Notice that in Equation~\eqref{equ:A-stat-comp-dual},
it is assumed that the budget constraint in the bid~\eqref{equ:A-bid} does not become active.
This is without loss of generality as we prove that the \emph{budget depletion time} $\mathscr{T}^k_i$, defined in Equation~\eqref{equ:G-hit-time}, asymptotically approaches the time horizon $T$, and the budget constraint indeed never becomes active until the end of the horizon.
The maximum of the expected dual objective is denoted by $\Psi_i \big(\mu^\star_i\big)$, where $\mu^\star_i > 0$ is a multiplier satisfying $Z_i \big(\mu_i^\star\big)=\rho_i$ and causing the expected expenditure to equal the target expenditure rate.

Finally, the \emph{average expected regret} of strategy $A$ is given by $ \mathcal{R}_i^A = \frac{1}{T} \: \Exp_{\bm{v_i}, \: \bm{d_i}} \left[ \mathcal{C}_i^A (\bm{v_i},\bm{d_i}) - \mathcal{C}_i^H (\bm{v_i},\bm{d_i}) \right].$
Notice that achieving $\mathcal{R}_i^A=0$ is a sufficient condition for the optimality of strategy $A$ under stationary competition, since the expected optimal cost with the benefit of hindsight $\Exp_{\bm{v_i}, \: \bm{d_i}} \left[\mathcal{C}_i^H (\bm{v_i},\bm{d_i})\right]$ provides a lower bound on the stationary competition cost~\eqref{equ:C-stat-comp} of any strategy.
The following standard assumption is required in order to ensure that $\mathcal{R}_i^A$ decays asymptotically to zero.

\begin{assumption}[$\epsilon(T)$ under Stationary Competition]\label{ass:A_stat_comp_limit_T}
The gradient step size $\epsilon$ is a function of the time horizon $T$ satisfying $\epsilon(T) \xrightarrow[T\to \infty]{} 0$ and $T \epsilon(T) \xrightarrow[T\to \infty]{} \infty$.
\end{assumption}

\begin{theorem} [Asymptotic Optimality under Stationary Competition] \label{thm:A_stat_comp}
    There exist a constant $C \in \mathbb{R}_+$ such that the average expected regret of agent $i \in \mathcal{N}$ for following strategy $A$ satisfies
    \small
\begin{equation*}
\small
      \frac{1}{T} \: \Exp_{\bm{v_i}, \: \bm{d_i}} \left[ \mathcal{C}_i^A (\bm{v_i},\bm{d_i}) - \mathcal{C}_i^H (\bm{v_i},\bm{d_i}) \right] 
    \leq C \left( \epsilon 
    + \frac{1 + \epsilon }{\epsilon T} \right).
\end{equation*}
\normalsize
    Moreover, 
    for the asymptotic framework described by Assumption~\ref{ass:A_stat_comp_limit_T},
    strategy $A$ achieves sub-linear regret and 
converges to the optimal expected cost with the benefit of hindsight, i.e.,
    \begin{equation*}
    \small
    \lim_{T \to \infty} \: \frac{1}{T} \: \Exp_{\bm{v_i},\bm{d_i}} \left[ \mathcal{C}_i^A (\bm{v_i},\bm{d_i}) - \mathcal{C}_i^H (\bm{v_i},\bm{d_i}) \right] = 
    0.
    \end{equation*}
    Furthermore, the convergence rate is at least $O(T^{-1/2})$ with the choice of $\epsilon \propto T^{-1/2}$.
\end{theorem}

\subsection{Convergence under Simultaneous Learning}\label{sec:A_sim_lear}

In this section, we establish that the learning dynamics \emph{converge in the simultaneous learning setting} in which all agents follow strategy $A$, denoted by joint strategy profile $\bm{A}$.
Before stating the main result of the section, we first adapt our previous definitions to the multi-agent setting.
Let $\bm{\mu}_t \in \mathbb{R}_+^N$ be the \emph{multiplier profile} stacking the multipliers $\mu_{i,t}$ of all agents $i \in \mathcal{N}$.
We extend the \emph{expected dual objective} and the \emph{expected expenditure} respectively as 
\begin{equation*}
\small
\begin{aligned}
    \Psi_i(\bm{\mu}) &= \Exp_{\bm{v}} \left[v_i  -\mu_i \rho_i - \left(\Delta v_i - \mu_i d_i^\gamma \right)^+\right],
    \\
    Z_i(\bm{\mu}) &= \Exp_{\bm{v}} \left[d_i^\gamma \: \mathds{1}\left\{\Delta v_i  > \mu_i d_i^\gamma \right\} \right].
    \end{aligned}
\end{equation*}
Compared to~\eqref{equ:A-stat-comp-dual}
, the expectation is now with respect to the profile of valuations $\bm{v} \sim \bm{\mathcal{V}}$, as given $\bm{v}$ and $\bm{\mu}$ the competing bid $d_i^\gamma = \gamma\textsuperscript{th}\mbox{-}\max_{j:j\neq i} \{ \Delta v_j / \mu_j\}$ can be uniquely determined.

We will aim to show that $\bm{\mu}_t$ converges to a \emph{stationary multiplier profile}, which is a multiplier profile $\bm{\mu^\star} \in \mathbb{R}^N_{>0}
$ satisfying $Z_i \big(\bm{\mu^\star} \big) = \rho_i$, for all agents $i \in \mathcal{N}$.
This multiplier profile is stationary in the sense that in expectation, update rule~\eqref{equ:A-mu} will yield $\mu^\star_{i,t+1} = \mu^\star_{i,t}$ for all agents $i$, since the expected losses $Z_i \big(\bm{\mu^\star} \big)$ equal the target rates $\rho_i$.
To prove the convergence of the multiplier profile, it is standard to adopt the following assumption on the expected expenditure $\bm{Z}$~\cite{balseiro2019learning}.
\begin{assumption}[Monotonicity]\label{ass:A_sim_lear_strong_mono}
    The expected expenditure $\bm{Z}$ is $\lambda$-strongly monotone over $\bm{U}= \prod_{i \in \mathcal{N}} \left(\underline{\mu}, \overline{\mu}\right)$ with parameter $\lambda > 0$, i.e., for all $\bm{\mu}, \bm{\mu'} \in \bm{U}$, it holds that $(\bm{\mu} - \bm{\mu}')^\top (\bm{Z}(\bm{\mu}) - \bm{Z}(\bm{\mu'}) ) \leq - \lambda \Vert \bm{\mu} - \bm{\mu'}\Vert_2^2$.
\end{assumption}

\begin{theorem}[Convergence under Simultaneous Learning]\label{thm:A-sim-lear-converge}
There exist constants $C_1$ and $C_2\in \mathbb{R_+}$ such that the average expected distance to the stationary multiplier profile $\bm{\mu^\star}$ and the strategic competition cost of strategy profile $\bm{A}$ for any agent $i \in \mathcal{N}$ satisfy respectively
\begin{equation*}
\small
\begin{aligned}
    \frac{1}{T} \sum_{t=1}^T \Exp_{\bm{v}} \left[\left\Vert \bm{\mu_t} - \bm{\mu^\star} \right\Vert_2^2 \right] \leq C_1 N\left( \epsilon +  \dfrac{ 1 + \epsilon}{\epsilon T}  \right),\\
     \frac{1}{T} \: \mathcal{C}_i^{\bm{A}} - \Psi_i \big(\bm{\mu^\star}\big)
    \leq C_2 \left( \frac{1}{T} + N 
    \left( \epsilon^{1/2} +  \frac{ 1}{\epsilon T} \right)\right).
    \end{aligned}
\end{equation*}
Moreover,
for the asymptotic framework described by Assumption~\ref{ass:A_stat_comp_limit_T},
the multiplier profile $\bm{\mu_t}$ converges in expectation to the stationary 
profile $\bm{\mu^\star}$, and the average strategic competition cost $\mathcal{C}_i^{\bm{A}}$ converges to the optimal expected dual objective $\Psi_i \big(\bm{\mu^\star} \big)$ for all agents $i \in \mathcal{N}$, i.e., 
$\lim_{T \to \infty} \: \frac{1}{T} \sum_{t=1}^T \Exp_{\bm{v}} \left[\Vert \bm{\mu_t} - \bm{\mu^\star} \Vert_2^2 \right] = 0 \text{ and } \lim_{T \rightarrow \infty} \: \frac{1}{T} \: \mathcal{C}_i^{\bm{A}} - \Psi_i \big(\bm{\mu^\star} \big) = 0. $ 
\end{theorem}

\subsection{Approximate Nash Equilibrium in Parallel Auctions}\label{sec:A_eps_NE}

In this section, we combine the results of the previous two sections to establish that the profile of adaptive pacing strategies $\bm{A}$ constitutes an $\varepsilon$-Nash equilibrium in the same parallel auction setting introduced in Section~\ref{sec:G_eps_NE}, for which the following standard assumption is needed~\cite{balseiro2019learning}.

\begin{assumption}[Matching Probabilities]\label{ass:A_eps_NE}
The matching probabilities $\left(\bm{a_i}\right)_{i \in \mathcal{N}}$ satisfy:
\begin{enumerate}[label=\ref{ass:A_eps_NE}.\arabic*]
    \item  The fraction of potential auction winners $\frac{M\gamma}{N}$ is constant for all $N \in \mathbb{N}$
    ;\label{ass:frac_winner_const}

    \item There exists a constant $\kappa>0$ such that $\sqrt{N} \max_{i\in\mathcal{N}} \Vert\bm{a_i}\Vert_2 \leq \kappa$ for all $N \in \mathbb{N}$; \label{ass:A_eps_NE-1}
    
    \item It holds that $\max\limits_{i\in\mathcal{N}} \Vert\bm{a_i}\Vert_2 \xrightarrow[T,N,M \to \infty]{} 0$.\label{ass:A_eps_NE-2}
\end{enumerate}
\end{assumption}
Notice that Assumptions~\ref{ass:A_eps_NE-1} and~\ref{ass:A_eps_NE-2} are implied by Assumption~ \ref{ass:frac_winner_const} in many cases, e.g. when agents are assigned to auctions uniformly at random.
In this case, $a_{i,j} = 1 / M^2$ for all $i\neq j$, and $\sqrt{N} \: \Vert\bm{a_i}\Vert_2 = N / M $ is constant.

\begin{theorem}[$\varepsilon$-Nash Equilibrium] \label{thm:A_eps_NE}
    There exists a constant $C\in \mathbb{R_+}$ such that each agent $i \in \mathcal{N}$ can decrease its average strategic competition cost by deviating from strategy $A$ to any strategy $\beta_i \in \mathcal{B}$ by at most
\begin{equation*}
\small
\begin{aligned}
     \frac{1}{T} \left(\mathcal{C}_i^{A} - \mathcal{C}_i^{\beta, A_{-i}} \right)
    &\leq C \Bigg( 
    \Vert \bm{a_i}\Vert_2 
    \bigg( \sqrt{N} \bigg( \epsilon + \frac{1}{ \epsilon T} \bigg)
    + 
    \Vert \bm{a_i}\Vert_2 
    \bigg) 
    +\frac{1 + \epsilon}{\epsilon T} 
   \Bigg).
\end{aligned}
\end{equation*}
Moreover,
for the asymptotic framework described by Assumptions~\ref{ass:A_stat_comp_limit_T} and~\ref{ass:A_eps_NE},
strategy profile $\bm{A}$ constitutes an $\varepsilon$-Nash equilibrium, i.e., it holds for all agents $i \in \mathcal{N}$ that
\begin{equation*}
\small
\lim_{T, N, M \to \infty} \frac{1}{T} \left( \mathcal{C}_i^{\bm{A}} - \inf_{\beta_i \in \mathcal{B}} \mathcal{C}_i^{\beta, \bm{A_{-i}}} \right) = 0.
\end{equation*}
\end{theorem}

\section{Numerical experiments}\label{sec:numerical}

This appendix presents numerical results that complement the main theoretical results presented in the body of the paper.
Figures \ref{fig:verif_thms} and~\ref{fig:test_hypoth} were produced by running $100$ simulations for each parameter combination, with the mean and the estimated $95\%$ confidence interval displayed.
They use a log-log scale with the time horizon $T$ in the x-axis and the quantity of interest in the y-axis.
Therefore, a decreasing trend implies that the quantity of interest is converging asymptotically to zero, and the slope of decrease gives the rate of convergence as a power of $T$.

Figure~\ref{fig:comp_strat} on the other hand is the result of a single simulation and shows the evolution of critical parameters throughout an episode of length $T=5000.$

\begin{figure*}[!tbh]
     \centering
     \begin{subfigure}{0.48\textwidth}
         \centering
         \includegraphics[width=\textwidth]{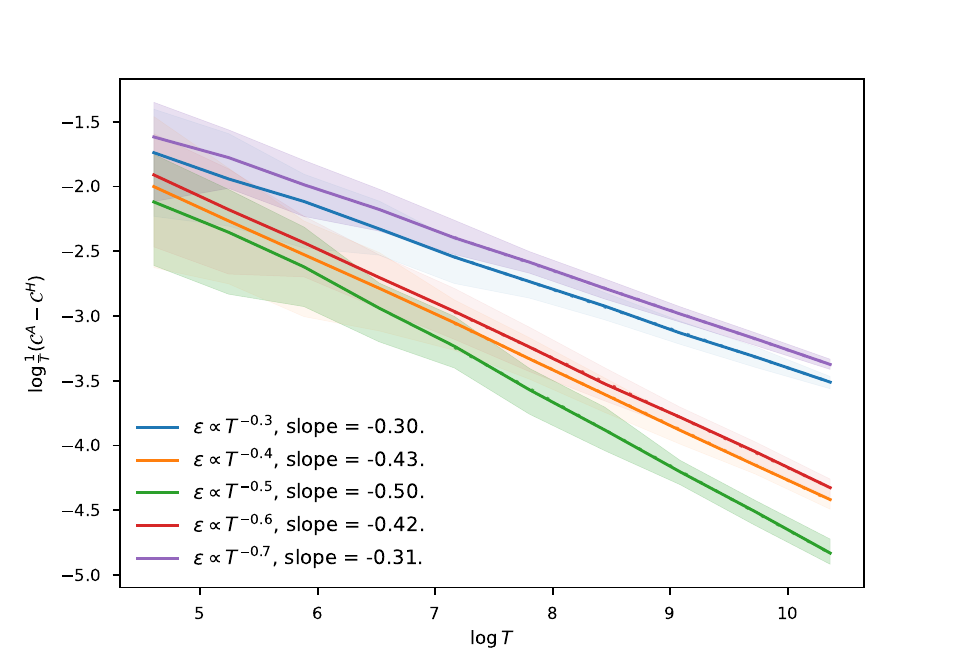}
         \caption{Theorem \ref{thm:A_stat_comp}: Stationary competition, no budget gains.}
         \label{fig:P_vs_H_no_redistrib}
     \end{subfigure}
     \hfill
     \begin{subfigure}[b]{0.48\textwidth}
         \centering
         \includegraphics[width=\textwidth]{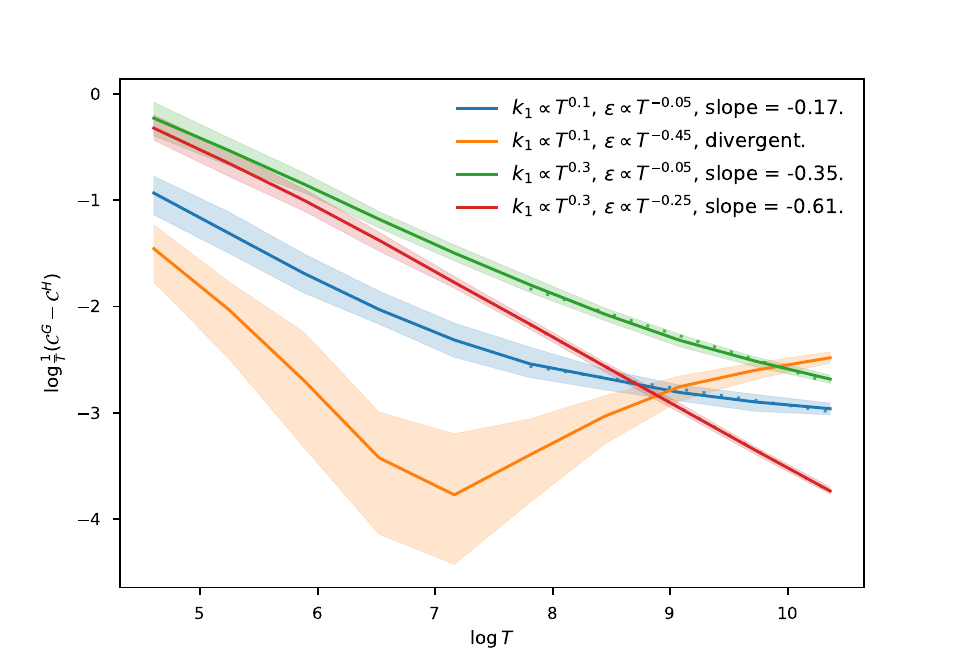}
         \caption{Theorem \ref{thm:G_stat_comp}: Stationary competition, with budget gains.}
         \label{fig:P_vs_H_with_redistrib}
     \end{subfigure}
     \hfill
     \begin{subfigure}[b]{0.48\textwidth}
         \centering
         \includegraphics[width=\textwidth]{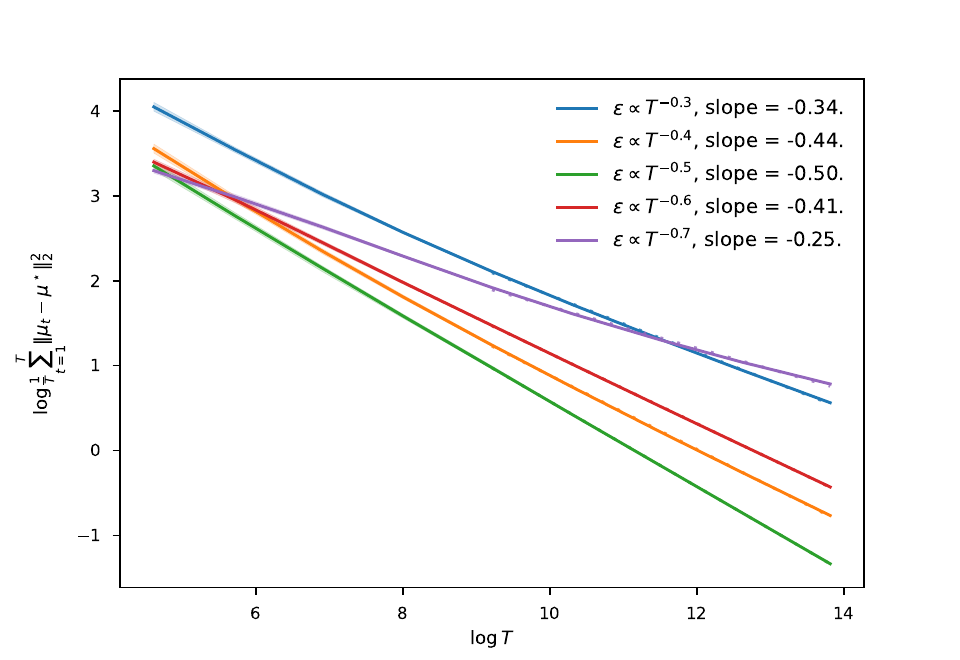}
         \caption{Theorem \ref{thm:A-sim-lear-converge}: Simultaneous learning, no budget gains.}
         \label{fig:cv_mu_no_redistrib}
     \end{subfigure}
     \hfill
     \begin{subfigure}[b]{0.48\textwidth}
         \centering
         \includegraphics[width=\textwidth]{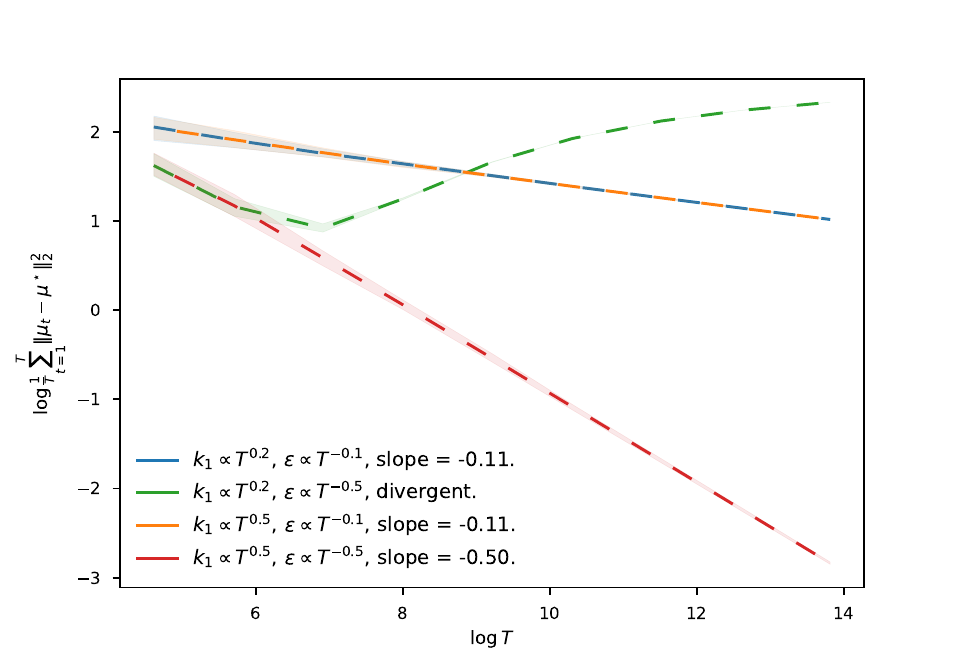}
         \caption{Theorem \ref{thm:G_sim_lear_cv}: Simultaneous learning, with budget gains.}
         \label{fig:cv_mu_with_redistrib}
     \end{subfigure}
        \caption{Numerical validation of the main theorems.
        Figures~\ref{fig:P_vs_H_no_redistrib} and~\ref{fig:P_vs_H_with_redistrib} show the convergence of costs to the minimum with the benefit of hindsight, without and with budget gains due to payment redistribution.
        Figures~\ref{fig:cv_mu_no_redistrib} and~\ref{fig:cv_mu_with_redistrib} show the convergence of multipliers under simultaneous learning, without and with budget gains.
        }
        \Description{This figure gathers 4 different subfigures, which numerically validate Theorems 4.1, 4.2, B.1 and B.2. 
        In subfigure (a), we plot in a log-log space the average regret with respect to hindsight as a function of the time horizon T, for different values of the step size. The slope in the log-log plot indeed corresponds to the values predicted by Theorem B.1, with epsilon equal to one over the square root of T as the optimal value.
        In subfigure (b), we also plot the average regret with strategy G in the stationary competition setting for different values of initial budget and epsilon. While the convergence rate does not align with the theoretical values and indicates some slack in our analysis, the qualitative predictions of Theorem 4.1 are verified. We also experimentally observe the vanishing box problem.
        In subfigures (c) and (d), we plot the distance to the stable multiplier profile for strategy $A$ and $G$ respectively. In subfigure (c), we again verify the prediction of Theorem B.2 quantitatively, whereas subfigure (d) again verifies the qualitative predictions of Theorem 4.2 (indicating some slack in the analysis), and exemplifies the vanishing box problem.}
        \label{fig:verif_thms}
\end{figure*}

\begin{figure*}[!tbh]
     \centering
     \begin{subfigure}[t]{0.48\textwidth}
         \centering
         \includegraphics[width=\textwidth]{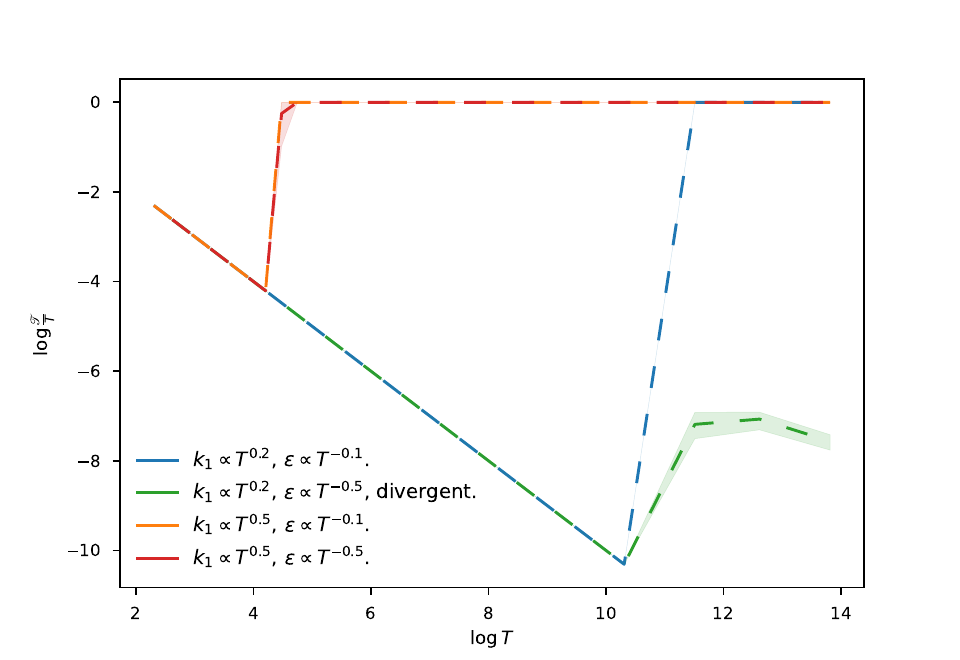}
         \caption{Validation of Assumption~\ref{ass:G-sim-learn-hitting-time}. The hitting time $\underline{\mathscr{T}}$ approaches $T$ asymptotically.}
         \label{fig:hit_time}
     \end{subfigure}
     \hfill
     \begin{subfigure}[t]{0.48\textwidth}
         \centering
        \includegraphics[width=\textwidth]{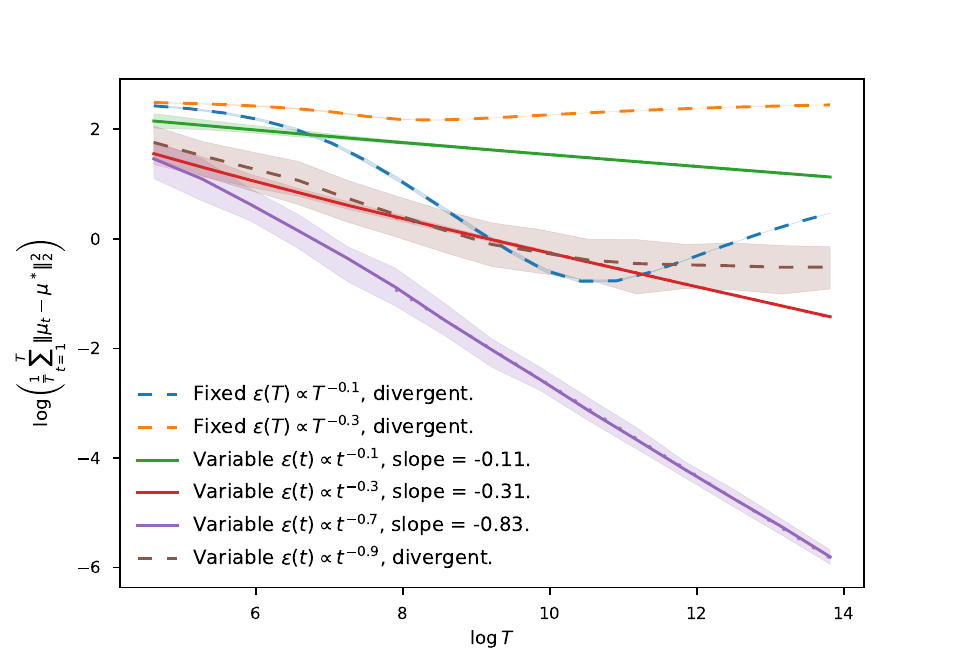}
        \caption{For a fixed initial budget, fixed step sizes $\epsilon$ diverge and variable step sizes $\epsilon_t$ converge under simultaneous learning.}
        \label{fig:fixed_budget}
     \end{subfigure}
     
     \begin{subfigure}[b]{0.48\textwidth}
         \centering
         \includegraphics[width=\textwidth]{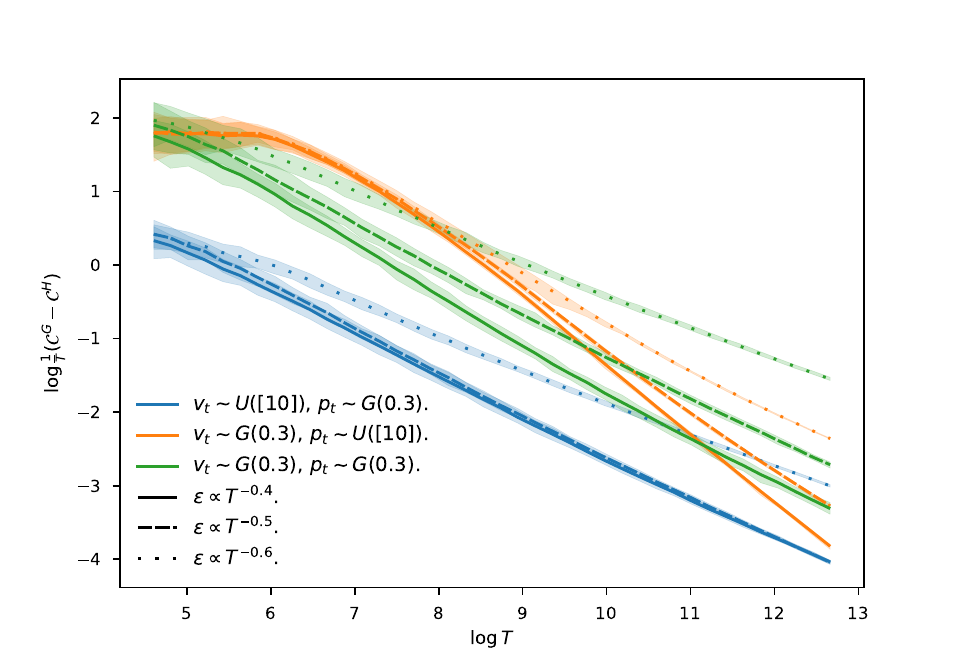}
         \caption{Stationary competition with discrete distributions of valuations and competing bids.}
         \label{fig:non_cont_v_perfs_vs_H}
     \end{subfigure}
     \hfill
     \begin{subfigure}[b]{0.48\textwidth}
         \centering
         \includegraphics[width=\textwidth]{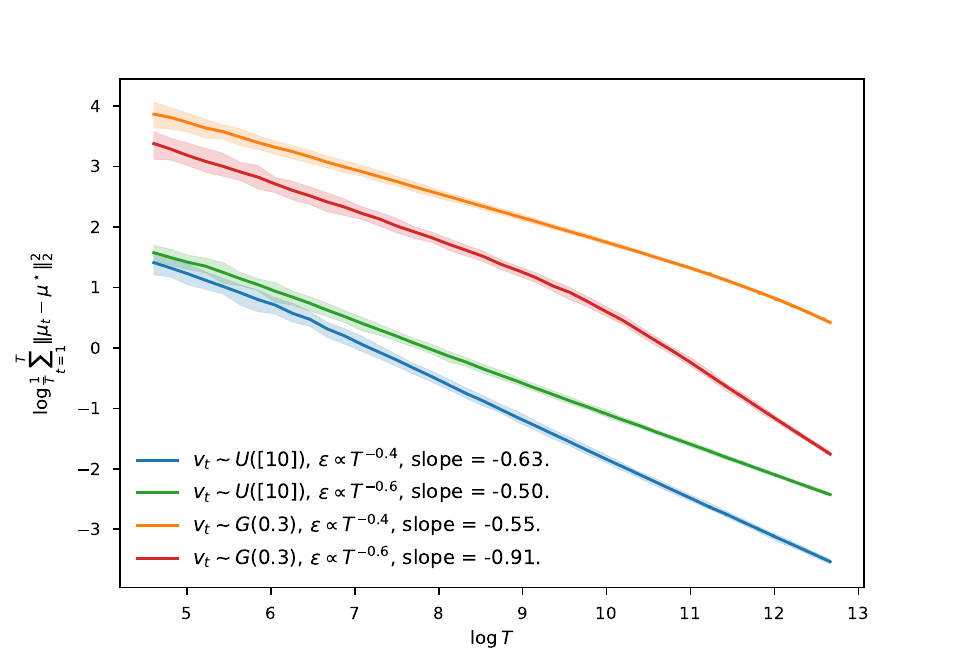}
         \caption{Simultaneous learning with discrete valuation distributions.}
         \label{fig:non_cont_v_cv_mu}
     \end{subfigure}
        \caption{Numerical validation of 
        the absolute continuity of valuations
        and Assumption~\ref{ass:G-sim-learn-hitting-time}, and potential solution for the vanishing box problem (Figure~\ref{fig:fixed_budget}).
        Figure~\ref{fig:hit_time} shows that Assumption~\ref{ass:G-sim-learn-hitting-time-4} is satisfied in practice.
        Figures~\ref{fig:non_cont_v_perfs_vs_H} and~\ref{fig:non_cont_v_cv_mu}  respectively show that the results of Theorems~\ref{thm:G_stat_comp} and~\ref{thm:G_sim_lear_cv} hold also when the distribution of valuations is not continuous.}
        \Description{This figure gathers 4 different subfigures.
        Subfigure (a) plots the logarithm of the hitting time divided by the time horizon T as a function of the logarithm of the time horizon, for different values of the initial budget and gradient time-step. For all values (except one that showcases the bounded box problem), the hitting time is deterministically equal to the time horizon T for T large enough.
        Subfigure (b) shows how variable stepsizes fix the bounded box problem: even for a fixed initial budget, one observes convergence of the average distance to the stable multiplier as T grows large. This is only the case though for an initial budget large enough: the smaller the stepsizes, the larger the initial budget is required to achieve convergence.
        Subfigures (c) and (d) are equivalent to figures 2.b and 2.c respectively, but they use non-continuous distributions for the valuations and the competing bids instead of continuous distributions as required per Theorems 4.1 and 4.2. In all cases, we observe qualitatively convergence (of the regret in (c) and of the average distance to the stable multiplier in (d)), which shows that our results may be extended beyond continuous distributions.
        }
        \label{fig:test_hypoth}
\end{figure*}

\begin{figure*}[!htb]
     \centering
     \begin{subfigure}[t]{0.48\textwidth}
         \centering
         \includegraphics[width=\textwidth]{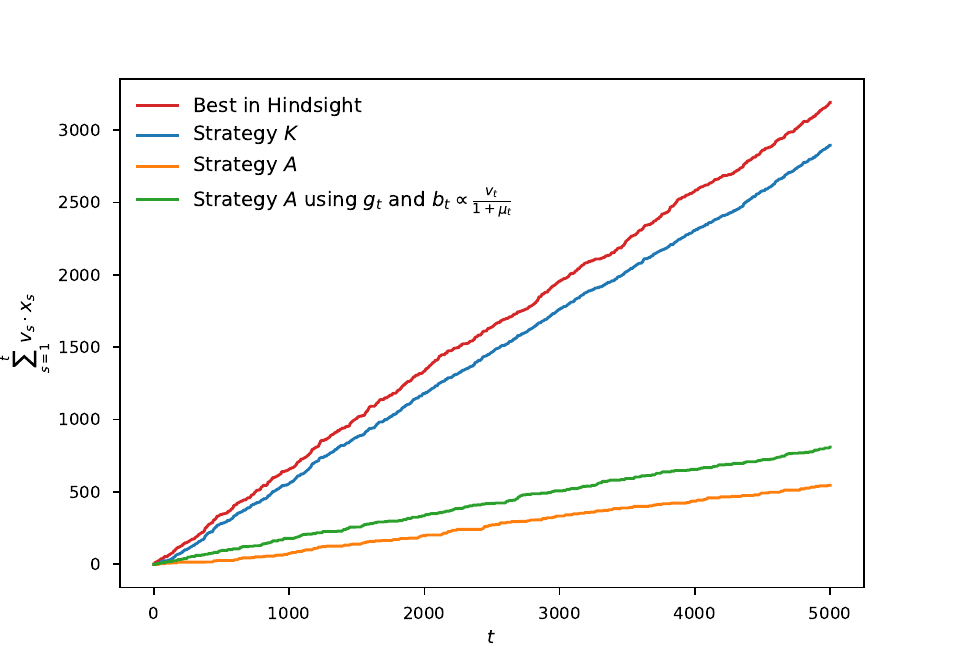}
         \caption{Cumulative value of time saved by winning auctions.}
         \label{fig:comp_strat_VoT}
     \end{subfigure}
     \hfill
     \begin{subfigure}[t]{0.48\textwidth}
         \centering
        \includegraphics[width=\textwidth]{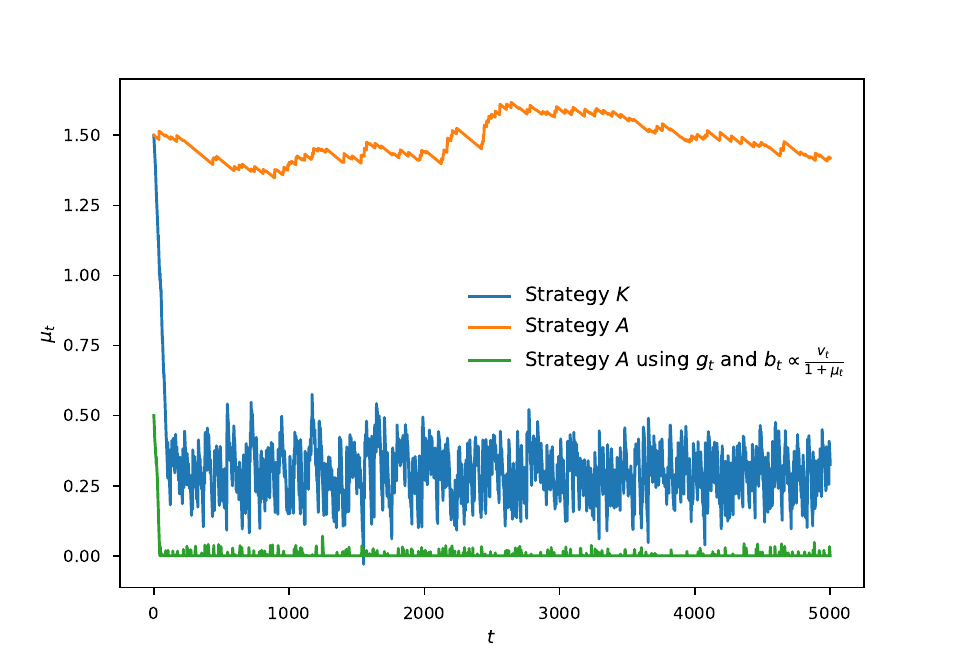}
        \caption{Evolution of the strategies' multiplier throughout an episode.}
        \label{fig:comp_strat_mu}
     \end{subfigure}
     
        \caption{Comparison of different strategies over an individual episode. 
        By ``Best in Hindsight'' we refer to the solution of Equation~\ref{equ:G_lowerbound_cost_H} which allows for temporary violations of the budget throughout the episode.
        By ``Strategy $K$'' and ``Strategy $A$'' we refer to Algorithms~\ref{alg:G} and~\ref{alg:A} respectively.
        By ``Strategy $A$ using $g_t$ and $b_t \propto \frac{v_t}{1+\mu_t}$'' we refer to the extension of the adaptive strategy proposed in \cite{balseiro_learning_2019} for classical monetary settings to handle budget increases. It places bids using $b_{i,t} = \min\{ \Delta v_{i,t} / (1 + \mu_{i,t}) , k_{i,t} \}$ and uses the multiplier update $\mu_{i,t+1} = P_{[0, \overline{\mu} ]} \left( \mu_{i,t} + \epsilon ( z_{i,t}  - g_{i,t} - \rho_i) \right).$
        Subfigure $(a)$ shows that both variations of strategy $A$ perform suboptimally compared to $K$, and subfigure $(b)$ explains their performances. On one hand, Strategy $A$ simply does not take gains into account and converges to a stationary multiplier that only depletes the initial budget. On the other hand, the alteration of Strategy $A$ aims to deplete the initial budget and the gains at the end of the episode, but it cannot achieve it because the denominator $1 + \mu_{i,t}$ in the bid formulation stays larger than one.
        }
        \Description{This figure gathers 2 different subfigures.
        Subfigure (a) plots the cumulative value of time saved by winning auctions as a function of the period $t$ ( and not the horizon capital $T$). All strategies follow linear trends, but strategy $K$ saves a value of time almost three times larger than that of strategy $A$ and strategy $A$ with gain.
        Subfigure (b) plots the evolution of the multiplier $\mu_t$ for the different strategies. While the multiplier of strategy $K$ quickly converges to a value around 0.25, the vanilla strategy $A$ converges to a much higher value, around 1.5, explaining why it performs so poorly. The modification of strategy $A$ however sees its multiplier converge to 0 and plateau there because of the projection to the interval $[0, \overline{\mu}].$ Note moreover that one is added to the value of the multiplier before placing bids, which means that the effective value of the multiplier is effectively stuck at 1.
        }
        \label{fig:comp_strat}
\end{figure*}

Unless specified otherwise, the experiments share the following parameters: valuations are drawn uniformly at random from $[0,1]$; multipliers are bounded to the interval $[0.1, 1000]$; and the travel time saving of the priority road is $\Delta = 5$. 
In the stationary competition setting, the competing bids $d^\gamma_{i,t}$ are also drawn uniformly at random from $[0,1]$, and we consider that the agent cannot set the price, i.e., $d^{\gamma+1}_{i,t} \approx d^\gamma_{i,t}$.
In the simultaneous learning setting, we consider $N=50$ agents and $\gamma = 5$ winners per period. 
The experiment-specific parameters are detailed below.
In the following, $C_\epsilon$ and $C_k$ denote the multiplicative constants such that $\epsilon(T) = C_\epsilon T^x$ and $k_1(T) = C_k T^y$, for the indicated powers $x$ and $y$ in the figures.
\begin{itemize}
    \item \emph{Figure \ref{fig:P_vs_H_no_redistrib}:} the initial multiplier is  $\mu_{i,1} = 10$; the target expenditure rate is $\rho = 0.2$; 
    and $C_\epsilon = 40$;
    
    \item \emph{Figure \ref{fig:P_vs_H_with_redistrib}:} the initial multiplier is $\mu_{i,1} = 5$; the karma gain is $g_{i,t} = d_{i,t}^\gamma / 10$; and $C_\epsilon = 20$, $C_k=3$;

    \item \emph{Figure \ref{fig:cv_mu_no_redistrib}:} the target expenditure rate is $\rho_i = 0.1$;
    the initial multipliers are $\mu_{i,1} = 5$ and $\mu_{i,1} = 6$ for the first and second halves of the population, respectively; and $C_\epsilon = 10$.
    The stationary multiplier profile $\bm{\mu^\star}$ is approximated by the mean profile that simulations with largest $T$ converged to;

    \item \emph{Figure \ref{fig:cv_mu_with_redistrib}:} the initial multipliers are $\mu_{i,1} = 5$ and $\mu_{i,1} = 6$ for the first and second halves of the population, respectively; the stationary multiplier is $\mu^\star_i = 5.5$ by symmetry; $C_\epsilon = 1$ and $C_k =3$;

    \item \emph{Figure \ref{fig:hit_time}:} parameters are the same as Figure \ref{fig:cv_mu_with_redistrib} with the exception of $C_k = 6$;

    \item \emph{Figure \ref{fig:fixed_budget}:} the fixed initial budget is $k_{i,1} =100$; 
    the initial multipliers are $\mu_{i,1} = 5$ and $\mu_{i,1} = 6$ for the first and second halves of the population, respectively; and $C_\epsilon =1$; 

    \item \emph{Figure \ref{fig:non_cont_v_perfs_vs_H}:} the fixed initial budget is $k_{i,1}=100$;
    the initial multiplier is $\mu_{i,1} = 55 > \mu^\star_i$; the karma gain is $g_{i,t} = d_{i,t}^\gamma / 10$; and $C_\epsilon = 20$;
    Valuations and competing bids are drawn either from the discrete uniform distribution $\mathcal{U}(\{1,\dots,10\})$ or the geometric distribution $\mathcal{G}(p)$ with parameter $p=0.3$;

    \item \emph{Figure \ref{fig:non_cont_v_cv_mu}:} the initial multipliers are $\mu_{i,1} = 5$ and $\mu_{i,1} = 6$ for the first and second halves of the population, respectively; the initial budget is $k_{i,1} = 100 \times T^{0.61}$;
    and $C_\epsilon = 1.$
    Valuations are drawn either from the discrete uniform distribution $\mathcal{U}(\{1,\dots, 10\})$ or the geometric distribution $\mathcal{G}(p)$ with parameter $p=0.3$.

    \item \emph{Figure \ref{fig:comp_strat}:} the initial multipliers are $\mu_1 =1.5$ for strategies $K$ and $A$, and $\mu_1 = 0.5$ for the alteration of $A.$ The initial budget is fixed at $k_{i,1} = 200$ and the step-size at $\epsilon_i = 0.01.$ The karma gain is $g_{i,t} = d_{i,t}^\gamma / 20.$ 
\end{itemize}


\section{Full Statements of the Theorems}\label{sec:G_add_ass}

In the following sections~\ref{sec:ass_G_stat}--\ref{sec:ass_G_NE}, we gather all technical assumptions required for the full statement of Theorems~\ref{thm:G_stat_comp}--\ref{thm:G_eps_NE}, respectively.

\subsection{Full Statement of Theorem~\ref{thm:G_stat_comp}}\label{sec:ass_G_stat}

To fully state Theorem~\ref{thm:G_stat_comp}, we require a few technical assumptions which are listed and briefly discussed below.

\begin{assumption}[Differentiability and Concavity]\label{ass:G_stat_comp_L_mono}
    The following conditions hold:
    \begin{enumerate}[label=\ref{ass:G_stat_comp_L_mono}.\arabic*]
    \item \label{ass:G-stat-comp-dual-diff} The expected dual objective $\Psi_i^K(\mu_i)$ is differentiable on $\mathbb{R}_+$ with derivative $\Psi_i^{K'}: \mu_i \mapsto L_i(\mu_i) -  \mu_i G_i^{'}(\mu_i)$;
    
    \item \label{ass:G-differentiable} The expected expenditure $Z_i(\mu_i)$ and the expected gain $G_i(\mu_i)$ are twice differentiable with derivatives bounded by $\overline{Z^{'}}> 0$ and $\overline{G^{'}}> 0$ in the first order, respectively, and $\overline{Z^{''}}> 0$ and $\overline{G^{''}}> 0$ in the second order, respectively;

    \item \label{ass:G-stat-comp-L-decreasing} The expected karma loss $L_i(\mu_i) = Z_i(\mu_i) - G_i(\mu_i)$ is strongly decreasing with parameter $\lambda > 0$ on $\big[\underline{\mu}, \overline{\mu}\big]$.
    \end{enumerate}
\end{assumption}

Assumption~\ref{ass:G_stat_comp_L_mono} is common in the optimization literature~\cite{balseiro2019learning,balseiro2019multiagent} and implies strong concavity of $\Psi_i^0$, which ensures 
the existence of a unique maximizer $\mu_i^{\star0}$ and 
the convergence of gradient-based techniques.
It is guaranteed to hold under mild conditions on the distributions $\mathcal{V}_i$ and $\mathcal{D}^\gamma_i$ which we include in Appendix~\ref{sec:A_stat_comp_proof_imply_differentiability+concavity}.

\begin{assumption}[Parameter design under stationary competition] \label{ass:G-parameters}
    The input parameters of strategy $K$ satisfy:
    \begin{enumerate}[label=\ref{ass:G-parameters}.\arabic*]
        \item \label{ass:G_stat_comp_mu^*} The multiplier bounds are chosen such that
         $\underline{\mu} < \mu_i^{\star0} < \overline{\mu}$;

         \item \label{ass:G_stat_comp_distr_comp_bids} The distribution of competing bids $\mathcal{D}_i$ has support in $\Big[0, \frac{\Delta}{\underline{\mu}}\Big]^2$ and satisfies
        $0 < \rho_i < \big(1 - \frac{\gamma}{N}\big) \: \mathbb{E}_{ \bm{d_i}} \left[ d_i^{\gamma} \right] $;

        \item \label{ass:G_stat_comp_eps} The gradient step size satisfies $\epsilon \leq \min \bigg\{  \frac{1}{2\lambda}, \frac{1}{\overline{Z^{'}} + \overline{G^{'}}} \bigg\}.$
    \end{enumerate}
\end{assumption} 

Assumption~\ref{ass:G-parameters} regards the choice of the \emph{input parameters of strategy $K$} and can be satisfied by design, i.e., by setting $\underline{\mu}$ and $\epsilon$ sufficiently low and $\overline{\mu}$ sufficiently high.


\begin{assumption}[Control of Hitting Time] \label{ass:G_stat_comp_hitting_time}
    The following conditions hold:
    \begin{enumerate}[label=\ref{ass:G_stat_comp_hitting_time}.\arabic*]
      
        
        
        \item The algorithm parameters satisfy $\underline{\mu} < \frac{\Delta \underline{v_i}}{\overline{d_i}^2} $ and $\overline{\mu} > \frac{\Delta}{\underline{d_i}}$, as well as  $\epsilon < \min \left\{  \frac{\Delta \underline{v_i}}{\overline{d_i}^2} - \underline{\mu}, \frac{1}{\overline{d_i}} \Big( \overline{\mu} -\frac{\Delta}{\underline{d_i}}   \Big) \right\}$;\label{ass:G_stat_comp_hitting_time-3}
        
        \item The initial budget satisfies $k_{i,1} > \frac{ \overline{\mu} - \mu_{i,1}}{\epsilon} + \frac{\Delta}{\underline{\mu}}$.\label{ass:G_stat_comp_hitting_time-5}
    \end{enumerate}
\end{assumption}

Assumption~\ref{ass:G_stat_comp_hitting_time} compliments Assumption~\ref{ass:G_stat_comp_hitting_time_main} in controlling the hitting time $\mathscr{T}_i$ deterministically.
Specifically, Assumptions~\ref{ass:G_stat_comp_hitting_time_main} and~\ref{ass:G_stat_comp_hitting_time-3} guarantee that $\mathscr{T}_i^{\underline{\mu}}= \mathscr{T}_i^{\overline{\mu}} = T$, while Assumption~\ref{ass:G_stat_comp_hitting_time-5} implies $\mathscr{T}_i^k \geq \mathscr{T}_i^{\overline{\mu}}.$ Note moreover that Assumptions~\ref{ass:G_stat_comp_limit_T} and \ref{ass:G_stat_comp_hitting_time-5} together imply that $k_{i,1}(T) \xrightarrow[T\to \infty]{} \infty$.


\begin{reptheorem}{thm:G_stat_comp}
Under Assumptions~\ref{ass:G_stat_comp_L_mono} and \ref{ass:G-parameters},
there exists a constant $C \in \mathbb{R}_+$ such that the average expected regret of an agent $i \in \mathcal{N}$ for following strategy $K$ in the stationary competition setting satisfies
\begin{equation*}
\begin{aligned}
      &\frac{1}{T} \Exp_{\bm{v_i}, \bm{d_i}} \left[ \mathcal{C}_i^{K} (\bm{v_i}, \bm{d_i}) - \mathcal{C}_i^H (\bm{v_i}, \bm{d_i})  \right] \\
        &\leq   C \left(\epsilon  +  \frac{ 1 }{\epsilon T} + \frac{ k_{i,1}  }{ T}
        + \frac{\mathbb{E}_{\bm{v_i}, \bm{d_i}} \left[T - \mathscr{T}_i  \right]}{T} + \hat{\varepsilon}\right)
    \end{aligned}
\end{equation*}
Moreover, for the asymptotic framework described by Assumptions~\ref{ass:G_stat_comp_limit_T}, \ref{ass:G_stat_comp_hitting_time_main}, \ref{ass:A_stat_comp_limit_T} and~\ref{ass:G_stat_comp_hitting_time},
strategy $K$ asymptotically converges to an $O(\hat{\varepsilon})$-neighborhood of the optimal expected cost with the benefit of hindsight, i.e., 
$$\lim_{T \to \infty} \: \frac{1}{T} \: \Exp_{\bm{v_i},\bm{d_i}} \left[ \mathcal{C}_i^K (\bm{v_i},\bm{d_i}) - \mathcal{C}_i^H (\bm{v_i},\bm{d_i}) \right] = O(\hat{\varepsilon}).$$
\end{reptheorem}

\subsection{Full Statement of Theorem~\ref{thm:G_sim_lear_cv}}\label{sec:ass_G_sim}

To fully state Theorem~\ref{thm:G_sim_lear_cv},
we require to adapt a few of the previous technical assumptions to the multi-agent setting.

\begin{assumption}[Parameter design under simultaneous learning] \label{ass:G-sim-learn-parameters}
    The following conditions hold:
    \begin{enumerate}[label=\ref{ass:G-sim-learn-parameters}.\arabic*]
        \item \label{ass:G_sim_lear_mu_star} The multiplier bounds and the initial multipliers are chosen such that
         $\underline{\mu} < \mu_i^{\star0} < \overline{\mu}$ for all agents $i\in\mathcal{N}$; 

        \item \label{ass:G_sim_lear_shared_eps} The shared gradient step size of the agents satisfies $\epsilon \leq \frac{1}{2\lambda}$.   
    \end{enumerate}
\end{assumption}

\begin{assumption}[Control of Hitting Time] \label{ass:G-sim-learn-hitting-time}
The following conditions hold for all agents $i \in \mathcal{N}$:
\begin{enumerate}[label=\ref{ass:G-sim-learn-hitting-time}.\arabic*]
    
    \item \label{ass:G-sim-learn-hitting-time-2} The multiplier bounds satisfy $0 < \underline{\mu} < \frac{\underline{v} }{2} \mu_m $ and \\$\overline{\mu} \geq \mu_m  \left(  1 +  \frac{2} {\underline{v}} \left(1 - \frac{\gamma}{N} \right)^{-1} - \frac{ \underline{v} } {2 } \right)$;
    
    \item \label{ass:G-sim-learn-hitting-time-3} The shared gradient step size satisfies $$0 < \epsilon <  \frac{ \mu_m \underline{\mu}}{\Delta} \min\left\{ \left(1 - \frac{\underline{v}}{2} \right)\frac{N}{\gamma} ,\frac{\underline{v}}{2} \left( 1 + (\underline{v} + 1) \frac{\gamma}{N}\right)^{-1} , \frac{1}{\underline{v}} \Big( 1 - \frac{\gamma^2}{N^2} \Big)^{-1} \right\};$$

    \item \label{ass:G-sim-learn-hitting-time-4}  The underlying learning dynamics are such that $$\frac{1}{T} \mathbb{E}_{\bm{v}} \left[  \mathscr{T}_i^{\underline{\mu}}\right] \xrightarrow[T \to \infty, ~\epsilon(T) \to 0]{} 1.$$
    
\end{enumerate}
\end{assumption}

Assumption~\ref{ass:G-sim-learn-parameters} 
is the multi-agent counterpart of Assumption~\ref{ass:G-parameters}, while Assumption~\ref{ass:G-sim-learn-hitting-time} partially replaces and compliments Assumptions~\ref{ass:G_stat_comp_hitting_time_main} and~\ref{ass:G_stat_comp_hitting_time} to control the hitting time $\mathscr{T}_i$ deterministically in the simultaneous adoption setting.


\begin{reptheorem}{thm:G_sim_lear_cv}
Under Assumptions~\ref{ass:G_sim_lear_strong_mono} and~\ref{ass:G-sim-learn-parameters},
there exist constants $C_1$ and $C_2\in \mathbb{R_+}$ such that the average expected distance to the stationary multiplier profile $\bm{\mu^{\star0}} \in \bm{H_{\mu_1}}$ and the strategic competition cost of strategy profile $\bm{K}$ for any agent $i \in \mathcal{N}$ satisfy respectively
\begin{equation*}
\begin{aligned}
    \frac{1}{T} \sum_{t=1}^T \Exp_{\bm{v}} \left[\left\Vert \bm{\mu_t} - \bm{\mu^{\star0}} \right\Vert_2^2 \right] 
     \leq  C_1 N \left( \epsilon + \frac{1}{\epsilon T} 
    +  \frac{\Exp_{\bm{v}} \left[ T - \underline{\mathscr{T}}  \right]}{T} \right),\\
     \frac{1}{T} \mathcal{C}_i^{\bm{K}} - \Psi_i^0 \big(\bm{\mu^{\star0}}\big)
     \leq C_2 \left( N \left( \epsilon^{1/2} +   \frac{1}{\epsilon T}\right) + \frac{ \Exp_{\bm{v}} \left[T - \underline{\mathscr{T}} \right]}{T}\right).
     \end{aligned}
\end{equation*}
Moreover, for 
the asymptotic framework described in Assumptions~\ref{ass:G_stat_comp_hitting_time-1}, \ref{ass:A_stat_comp_limit_T}, 
\ref{ass:G_stat_comp_hitting_time-5}, and~\ref{ass:G-sim-learn-hitting-time},
the multiplier profile $\bm{\mu_t}$ converges in expectation to the stationary profile $\bm{\mu^{\star0}}$, and the average strategic competition cost $\mathcal{C}_i^{\bm{K}}$ converges to the expected dual objective $\Psi^0_i \big(\bm{\mu^{\star0}} \big)$ for all agents $i \in \mathcal{N}$, i.e., 
$\lim\limits_{T \to \infty} \: \frac{1}{T} \sum_{t=1}^T \Exp_{\bm{v}} \left[ \big\Vert \bm{\mu_t} - \bm{\mu^{\star0}} \big\Vert_2^2 \right] = 0 \text{ and } \lim\limits_{T \rightarrow \infty} \: \frac{1}{T} \: \mathcal{C}_i^{\bm{K}} - \Psi_i^0 \big(\bm{\mu^{\star0}} \big) = 0. $
\end{reptheorem}

\subsection{Full Statement of Theorem~\ref{thm:G_eps_NE}.}\label{sec:ass_G_NE}

To fully state Theorem~\ref{thm:G_eps_NE}, we require the following additional technical assumptions.

\begin{assumption}[$\epsilon(T,N)$ and $\bm{k_1}(T,N)$ under Parallel Auctions]\label{ass:G_eps_NE_init}
    The shared gradient step size $\epsilon$ and initial budget $\bm{k_1}$ are functions of the time horizon $T$ and population size $N$ satisfying $N \epsilon \xrightarrow[T,N \to \infty]{} 0$, $\epsilon^{3/2} T\xrightarrow[T,N \to \infty]{} \infty$ and $\overline{k_1} / T \xrightarrow[T,N \to \infty]{} 0$.
\end{assumption}

\begin{assumption}[Control of Hitting Time]\label{ass:G-eps-NE-hitting-time}
The underlying learning dynamics are such that $$\frac{1}{T} \mathbb{E}_{\bm{v}} \left[ \mathscr{T}_i^{\underline{\mu}} + \mathscr{T}_i^{\overline{\mu}}\right] \xrightarrow[T \to \infty, ~\epsilon(T) \to 0
]{} 1.$$
\end{assumption}

Assumption~\ref{ass:G_eps_NE_init} strengthens the standard Assumption~\ref{ass:A_stat_comp_limit_T} by a) requiring that the gradient step size and initial budget vary with $N$ in addition to $T$ and b) requiring $\epsilon = \omega(T^{-2/3})$ instead of $\epsilon =\omega(T^{-1})$. 
Moreover, Assumption~\ref{ass:G-eps-NE-hitting-time} strengthens previous Assumption~\ref{ass:G-sim-learn-hitting-time-4},
which is required because minimum valuations are of little help in the parallel auction setting where agents with the highest and smallest multipliers may never interact.

\begin{reptheorem}{thm:G_eps_NE}
Under Assumptions~\ref{ass:G_sim_lear_strong_mono} and~\ref{ass:G-sim-learn-parameters},
there exists a constant $C\in \mathbb{R_+}$ such that each agent $i \in \mathcal{N}$ can decrease its average strategic competition cost by deviating from strategy $K$ to any strategy $\beta_i \in \mathcal{B}$ by at most
\begin{equation*}
\begin{aligned}
     \frac{1}{T} \left(\mathcal{C}_i^{K} - \mathcal{C}_i^{\beta_i, \bm{K_{-i}}} \right)
     &\leq C \Bigg( \left( \Vert \bm{a_i}\Vert_2 + \frac{M\gamma}{N} \right) \bigg( \sqrt{N\epsilon} \left(1 + \frac{1}{\epsilon^{3/2} T}\right) + \Vert \bm{a_i}\Vert_2 \\
     &+ \frac{\gamma}{\sqrt{N}}  \bigg) 
    + \left( \frac{\gamma}{N} + \frac{\overline{k_1}}{T} \right)
    +   \left( \frac{\overline{k_1}}{T} + \frac{M\gamma}{N} \right) 
    \frac{\mathbb{E}_{\bm{v}, \bm{m}} \left[T - \underline{\mathscr{T}}   \right]}{T} \Bigg)
\end{aligned}
\end{equation*}
Moreover, for the asymptotic framework described by Assumptions~\ref{ass:G_stat_comp_limit_T}, \ref{ass:frac_winner_const}, \ref{ass:A_eps_NE-2}, \ref{ass:G_stat_comp_hitting_time-5}, \ref{ass:G_eps_NE_init} and~\ref{ass:G-eps-NE-hitting-time},
strategy profile $\bm{K}$ constitutes an $\varepsilon$-Nash equilibrium, i.e., it holds for all agents $i \in \mathcal{N}$ that
$$ \lim_{T, N, M \to \infty} \frac{1}{T} \Big( \mathcal{C}_i^{\bm{K}} - \inf_{\beta_i \in \mathcal{B}} \mathcal{C}_i^{\beta_i, \bm{K_{-i}}} \Big) = 0.$$
\end{reptheorem}

\section{Proofs of the main Theorems}
This section contains the 
proofs of the Theorems in Section~\ref{sec:G}. It is organized as follows: first, 
we prove in Section~\ref{sec:G_helpful_lem} a bound on the expected mean squared error $s_t.$  
We then include the proofs of Theorems~\ref{thm:G_stat_comp}, \ref{thm:G_sim_lear_cv} and~\ref{thm:G_eps_NE}  in Sections~\ref{sec:proof_G_stat_comp}, \ref{sec:G_sim_lear_cv_proof} and~\ref{sec:G_eps_NE_proof} respectively. 
 The proofs of the secondary results discussed in Section~\ref{sec:G} are deferred to Section~\ref{sec:suppl_res}.

\subsection{Controlling the expected mean squared error}\label{sec:G_helpful_lem}
We recall that, by definition of the hitting time $\mathscr{T}_i$ in Equation~\eqref{equ:G-hit-time}, it holds for all $1 \leq t\leq \mathscr{T}_i$ that the projected multiplier used to formulate bids in \eqref{equ:G-b} verifies $P_{[\underline{\mu}, \overline{\mu}_i ]} (\mu_{i,t}) = \mu_{i,t}.$
We further define $\underline{\mathscr{T}} = \min_{i \in \mathcal{N}} \mathscr{T}_i$ as the earliest hitting time among the whole population.
For any fixed $t \in [T]$, we consider the expected mean squared error 
$s_t = \mathbb{E}_{\bm{v}} \left[ \big\Vert \bm{\mu_t} - \bm{\mu^{\star 0 }} \big\Vert_2^2  \vert t\leq \underline{\mathscr{T}} \right] =  \sum_{i=1}^N s_{i,t}$, where $s_{i,t} = \mathbb{E}_{\bm{v}} \left[ \big(\mu_{i,t} - \mu_i^{\star 0 }\big)^2 \vert t\leq \underline{\mathscr{T}} \right]$ is the individual expected error of agent $i.$ 
The following Lemma bounds the expected mean squared error in different scenarios.

\begin{lemma} [Upper bound on the Expected Mean Squared Error] ~ \label{lem:G_bound_MSE}
When all agents follow $K$ strategies
, the expected mean squared error $s_t$ satisfies:
\begin{equation*}
    s_t \leq  N \overline{\mu}^2 (1 - 2\lambda \epsilon)^{t-1} + \frac{ N \Delta^2 \epsilon }{2 \lambda \underline{\mu}^2 }  .
\end{equation*}
Taking $N =1$ gives a bound for an agent's individual error in the stationary competition setting.

\end{lemma}

\begin{proof}
For any agent $i \in \mathcal{N}$ and time step $t < \underline{\mathscr{T}} $, the multiplier update rule gives:
 \begin{align*}
      \left( \mu_{i,t+1} - \mu_i^{\star 0} \right) ^2 
      &= \left( \mu_{i,t} + \epsilon (z_{i,t} -  g_{i,t}) -\mu_i^{\star 0} \right)^2 ,\\
      &=  \left(\mu_{i,t} - \mu_i^{\star 0} \right)^2 +2 \epsilon \left(\mu_{i,t} -\mu_i^{\star 0}\right)(z_{i,t} -  g_{i,t} ) \\
      &+ \epsilon^2  ( z_{i,t} -  g_{i,t} )^2.
 \end{align*}
Taking the expectation on previous valuations of all agents (resp. previous valuations and competing prices) gives
\begin{equation}\label{equ:G_MSE_1}
    \begin{aligned}
        s_{i,t+1}
        &\leq  s_{i,t}
        + 2  \epsilon \mathbb{E}_{\bm{v}} \left[  \left(\mu_{i,t} - \mu_i^{\star 0} \right) \left(z_{i,t} -  g_{i,t}  \right) \vert t\leq \underline{\mathscr{T}}\right] \\
        &+ \epsilon^2 \mathbb{E}_{\bm{v}} \left[  ( z_{i,t} -  g_{i,t} )^2 \vert t\leq \mathscr{T}  \right],\\
        &\stackrel{(a)}{=} s_{i,t}
        + 2 \epsilon \mathbb{E}_{\bm{v}} \left[  \left(\mu_{i,t} - \mu_i^{\star 0}\right) \mathbb{E}_{\bm{v}} [z_{i,t} -  g_{i,t} \lvert \bm{\mu_t},  t\leq \mathscr{T}  ]  \vert t\leq \underline{\mathscr{T}} \right] \\
        &+ \epsilon^2 \mathbb{E}_{\bm{v}} \left[  ( z_{i,t} -  g_{i,t} )^2 \vert t\leq \mathscr{T} \right],\\
        &\stackrel{(b)}{=} s_{i,t}
        + 2  \epsilon \mathbb{E}_{\bm{v}} \left[  \left(\mu_{i,t} - \mu_i^{\star 0}\right)  L_i(\bm{\mu_t})   \vert t\leq \underline{\mathscr{T}} \right] \\
        &+ \epsilon^2 \mathbb{E}_{\bm{v}} \left[  ( z_{i,t} - g_{i,t}  )^2\vert t\leq \mathscr{T}  \right].\\
    \end{aligned}  
\end{equation}
Equality $(a)$ holds by linearity of expectation; equality $(b)$ by definition of $L_i(\bm{\mu_t})$ (resp. $L_i(\mu_{i,t})$).

First, we bound the sum over the whole population of the middle term in Equation~\eqref{equ:G_MSE_1}.
\begin{equation}\label{equ:G_MSE_2}
    \begin{aligned}
        &\mathbb{E}_{\bm{v}} \left[\sum_{i=1}^N     \left(\mu_{i,t} - \mu_i^{\star 0} \right) L_i(\bm{\mu_t})  \bigg\vert t\leq \underline{\mathscr{T}}\right] \\
        &=  \mathbb{E}_{\bm{v}} \left[  \sum_{i=1}^N \left(\mu_{i,t} -\mu_i^{\star 0}\right) \left( L_i\big(\bm{\mu_t}\big) - L_i\big(\bm{\mu^{\star 0}}\big) + L_i\big(\bm{\mu^{\star 0}}\big)  \right) \bigg\vert t\leq \underline{\mathscr{T}} \right] ,\\
        &\stackrel{(a)}{\leq} - \mathbb{E}_{\bm{v}} \left[\lambda  \left\Vert \bm{\mu_t} - \bm{\mu^{\star0}} \right\Vert_2^2  \vert t\leq \underline{\mathscr{T}} \right] ,\\
        &= -\lambda s_t.
    \end{aligned}
\end{equation}
Inequality $(a)$ uses the definition of $\bm{\mu^{\star 0}}$ (resp. $\mu_i^{\star 0}$), as well as the $\lambda$-strong monotonicity of the expected expenditure function $\bm{L}$ from Assumption~\ref{ass:G_sim_lear_strong_mono} (resp. Assumption~\ref{ass:G_stat_comp_L_mono}).

On the other hand, as $z_{i,t}$ and $ g_{i,t} $ are positive and bounded by $\Delta / \underline{\mu}$ under Assumption~\ref{ass:G-parameters}, we similarly bound the right-most term of Equation~\eqref{equ:G_MSE_1}:
\begin{equation}\label{equ:G_MSE_3}
    \sum_{i=1}^N  \mathbb{E}_{\bm{v}} \left[  ( z_{i,t}  -  g_{i,t} )^2 \vert t\leq \mathscr{T} \right] 
    \leq  \frac{N \Delta^2} {\underline{\mu}^2}.
\end{equation}

Combining Equations~\eqref{equ:G_MSE_1}, \eqref{equ:G_MSE_2} and~\eqref{equ:G_MSE_3} gives the following recursion
\begin{equation*}
    s_{t+1} 
    \leq (1 - 2\lambda \epsilon ) s_t +  \frac{ N\Delta^2 \epsilon^2}{\underline{\mu}^2}.
\end{equation*}

Since $s_1 \leq N \overline{\mu}^2$ and $ 2\lambda \epsilon \leq 1$ under Assumption~\ref{ass:G-parameters}, \cite[Lemma C.4] {balseiro2019learning} finally yields the desired result, i.e.,
\begin{equation*}
    s_t \leq  N \overline{\mu}^2  (1-2\lambda \epsilon )^{t-1} 
    + \frac{ N\Delta^2 \epsilon}{2\lambda \underline{\mu}^2} .
\end{equation*}

\end{proof}

However, if a single agent $i$ in $\mathcal{N}$ places bids with a strategy $\beta \in \mathcal{B}$ different from $K$ in a parallel auction setting, we consider the individual expected mean squared errors of all other agents.
For the purpose of the analysis, we allow this strategy to use \emph{complete information}, i.e., it can access the budgets $\bm{k_{s}}$, auctions $\bm{m_{s}}$, valuations $\bm{v_{s}}$, bids $\bm{b_{s}}$, and expenditures $\bm{z_{s}}$ of all agents at all previous time steps $s < t$, and in addition, it has knowledge of the valuation distributions $\bm{\mathcal{V}}$.
We denote the set of strategies that use complete information as $\mathcal{B}^{CI} \supset \mathcal{B}$ and suppose that $\beta_i \in \mathcal{B}^{CI}$.
To keep notations consistent, we take the convention $\mu_{i,t} = \mu_i^\star$ for all $t\in \mathbb{N}$, ensuring that $s_{i,t} = 0.$

\begin{lemma} [Expected Mean Squared Error with Unilateral Deviation] ~ \label{lem:G_bound_MSE_eps_NE}
If all agents follow strategy $K$ but agent $i$ follows strategy $\beta \in \mathcal{B}^{CI}$, the expected mean squared error of the population verifies:
   \begin{equation*}
    s_t 
    \leq  N \overline{\mu}^2  (1 - \lambda \epsilon )^{t-1} 
    + \frac{N \Delta^2 \epsilon}{\lambda \underline{\mu}^2} 
    + \frac{\Delta^2 }{\lambda ^2 \underline{\mu}^2} \left(\Vert \bm{a_i}\Vert_2 
    + \frac{ \gamma}{ \sqrt{N} } \right)^2.
\end{equation*}

\end{lemma}
\begin{proof}
We now turn to \emph{the case where all agents but $i$ follow strategy $K.$} Note that $\mu_{i,t}$ is a priory not defined since agent $i$ does not follow $K$; for convenience, we take the convention $\mu_{i,t} = \mu_i^{\star 0 }$ for all $t \in \mathbb{N}.$ 
    For an agent $j \neq i$, Equation~\eqref{equ:G_MSE_1} $(a)$ still holds and gives
    \begin{equation}\label{equ:G_MSE_epsNE_1}
    \begin{aligned}
        s_{j,t+1}
        &\stackrel{(a)}{\leq} s_{j,t}
        + 2  \epsilon \mathbb{E}_{\bm{v}, \bm{m}} \left[  \left(\mu_{j,t} - \mu_j^{\star 0}\right)  \mathbb{E}_{\bm{v}, \bm{m}} \left[z_{j,t}^\beta -  g_{j,t}^\beta \Big\lvert \bm{\mu_t} \right]   \Big\vert t\leq \underline{\mathscr{T}}\right] \\
        &+ \epsilon^2 \mathbb{E}_{\bm{v}, \bm{m}} \left[  \left( z_{j,t}^\beta   -  g_{j,t}^\beta \right)^2 \Big\vert t\leq \mathscr{T} \right] \\
        &\stackrel{(b)}{\leq} s_{j,t}
        + 2 \epsilon \mathbb{E}_{\bm{v}, \bm{m}} \left[  \left(\mu_{j,t} - \mu_j^{\star 0}\right) L_j(\bm{\mu_t})  \vert t\leq \underline{\mathscr{T}} \right] \\
        &+2 \epsilon \mathbb{E}_{\bm{v}, \bm{m}} \bigg[   \left \lvert \mathbb{E}_{\bm{v}, \bm{m}} \left[z_{j,t}^\beta  - g_{j,t}^\beta \big\lvert \bm{\mu_t} \right] - L_j(\bm{\mu_t}) \right\lvert \cdot \left\lvert \mu_{j,t} - \mu_j^{\star 0} \right\lvert  \Big\vert t\leq \underline{\mathscr{T}} \bigg] \\
        &+ \epsilon^2 \mathbb{E}_{\bm{v}, \bm{m}} \left[  \left( z_{j,t}^\beta   -  g_{j,t}^\beta \right)^2 \Big\vert t\leq \mathscr{T} \right]  .
    \end{aligned}  
\end{equation}
Inequality $(a)$ directly follows from inequality~\eqref{equ:G_MSE_1} $(a)$; inequality $(b)$ follows from adding and subtracting $L_j(\bm{\mu_t})$ in the second term of $(a)$, and noting that $xy \leq \lvert x \lvert \cdot \lvert y \lvert$ for $x,y \in \mathbb{R}.$
   
As Equations~\eqref{equ:G_MSE_2} and~\eqref{equ:G_MSE_3} hold when summing over all agents $j\neq i$, it only remains to show what happens for the term in $\mathbb{E}_{\bm{v}, \bm{m}} \left[z_{j,t}^\beta - g_{j,t}^\beta \Big\lvert \bm{\mu_t} \right].$ 
Let $d_{j,t}^{\gamma,K} = \gamma\textsuperscript{th}\mbox{-}\max_{k:k\neq j} \left\{ \mathds{1}\{ m_{k,t} = m_{j,t}\}  \frac{\Delta v_{k,t}}{\mu_{k,t}} \right\} $ denote the competing price in the imaginary setting where agent $i$ follows strategy $K$, as opposed to the real competing bid $d_{j,t}^{\gamma, \beta}.$
\begin{equation*}\label{equ:G_MSE_epsNE_2}
    \begin{aligned}
        &\left\vert \mathbb{E}_{\bm{v}, \bm{m}} \left[z_{j,t}^{\beta} - g_{i,t}^{\beta} \big\vert \bm{\mu_t} \right] - L_j(\bm{\mu_t}) \right\vert \\
        &\stackrel{(a)}{=} \Bigg \vert \mathbb {E}_{\bm{v}, \bm{m}} \bigg[ \mathds{1}\left\{m_{j,t} = m_{i,t} \right\} \bigg(
        d_{j,t}^{\gamma,\beta} \mathds{1}\left\{ \frac{\Delta v_{j,t}}{\mu_{j,t}} \geq d_{j,t}^{\gamma,\beta} \right\} \\
        &- d_{j,t}^{\gamma,K} \mathds{1}\left\{ \frac{\Delta v_{j,t}}{\mu_{j,t}} \geq d_{j,t}^{\gamma,K} \right\} \bigg) +   g_{j,t}^{K} - g_{j,t}^{\beta}  \bigg] \Bigg \vert ,\\
        &\stackrel{(b)}{\leq} \frac{\Delta}{\underline{\mu}} a_{i,j} +  \Big\vert \frac{\gamma}{N} \left(p_{m_j,t}^{K} -p_{m_j,t}^{\beta} \right) \Big\vert, \\
        &\stackrel{(c)}{\leq} \frac{\Delta}{\underline{\mu}} a_{i,j} +  \frac{\gamma\Delta }{N \underline{\mu} } .
    \end{aligned}
\end{equation*}
Equality $(a)$ uses that $d_j^{\gamma,\beta} = d_j^{\gamma,K}$ when $i$ plays in a different auction from $j$; inequality $(b)$ holds by rewriting  $g_t = \frac{\gamma}{N} \sum_{m=1}^M p_{m,t}$, and noting that  $0 \leq d_j^\gamma \mathds{1}\left\{ \frac{\Delta v_{j,t}}{\mu_{j,t}} \geq d_{j,t}^{\gamma} \right\} \leq \frac{\Delta}{\underline{\mu}}.$   Inequality $(c)$  uses that the price $p_m$ of any auction is bounded $0 \leq p_m \leq \frac{\Delta}{\underline{\mu}}.$

Summing the previous equation
over all $j\neq i$ gives
\begin{equation}\label{equ:G_MSE_epsNE_3}
\begin{aligned}
     \mathbb{E}_{\bm{v}, \bm{m}} \Bigg[  &\sum_{j\neq i}  \big\vert \mu_{j,t} - \mu_j^{\star 0} \big \vert \cdot \left\vert \mathbb{E}_{\bm{v}, \bm{m}} \left[ z_{j,t}^{\beta} - g_{i,t}^{\beta} \big\vert \bm{\mu_t} \right] - L_j(\bm{\mu_t}) \right\vert \bigg\vert t\leq \underline{\mathscr{T}} \Bigg] \\
    &\leq  \frac{\Delta}{\underline{\mu}}  \mathbb{E}_{\bm{v}, \bm{m}} \left[ \sum_{j\neq i}  a_{i,j} \big\vert \mu_{j,t} - \mu_j^{\star 0} \big\vert \bigg\vert t\leq \underline{\mathscr{T}}\right] \\
    &+ \frac{\gamma\Delta }{N \underline{\mu} } \mathbb{E}_{\bm{v}, \bm{m}} \Big[ \big\Vert \bm{\mu_t} -  \bm{\mu^{\star0}} \big\Vert_1 \vert t\leq \underline{\mathscr{T}}\Big] ,\\
    &\stackrel{(a)}{\leq} \left( \frac{\Delta}{\underline{\mu}} \Vert \bm{a_i}\Vert_2 +  \frac{ \gamma \Delta }{\sqrt{N}\underline{\mu}}  \right)
    \mathbb{E}_{\bm{v}, \bm{m}} \Big[ \big\Vert \bm{\mu_t} -  \bm{\mu^{\star0}} \big\Vert_2 \vert t\leq \underline{\mathscr{T}}\Big] ,\\
    &\stackrel{(b)}{\leq} \left( \frac{\Delta}{\underline{\mu}} \Vert \bm{a_i}\Vert_2 + \frac{\gamma\Delta }{\sqrt{N} \underline{\mu} } \right)  s_t^{1/2},
\end{aligned}
\end{equation}
where inequality $(a)$ follows from Cauchy-Schwartz and the relation between $\Vert \cdot \Vert_1$ and $\Vert \cdot \Vert_2$,  and $(b)$ from Jensen's inequality.

Combining Equations~\eqref{equ:G_MSE_2}, \eqref{equ:G_MSE_3}, \eqref{equ:G_MSE_epsNE_1} and~\eqref{equ:G_MSE_epsNE_3} gives the following recursion, i.e.,
\begin{equation*}
    s_{t+1} 
    \leq (1 - 2\lambda \epsilon ) s_t 
    + 2 \epsilon \left( \frac{\Delta}{\underline{\mu}} \Vert \bm{a_i}\Vert_2 + \frac{\gamma\Delta }{\sqrt{N} \underline{\mu} } \right)   s_t^{1/2}
    +   \frac{N\Delta^2 \epsilon^2}{\underline{\mu}^2} .
\end{equation*}

Since $s_1 \leq N \overline{\mu}^2 $ and $ 2\lambda \epsilon \leq 1$ under Assumption~\ref{ass:G-sim-learn-parameters}, \cite[Lemma C.5] {balseiro2019learning} finally yields the desired result:
\begin{equation*}
    s_t 
    \leq  s_1  (1 - \lambda \epsilon )^{t-1} 
    + \frac{ N \Delta^2 \epsilon }{\lambda \underline{\mu}^2} 
    + \frac{\Delta^2 }{\lambda^2 \underline{\mu}^2} \left(\Vert \bm{a_i}\Vert_2 
    + \frac{ \gamma}{ \sqrt{N} } \right)^2.
\end{equation*}

\end{proof}

\subsection{Proof of Theorem~\ref{thm:G_stat_comp}}\label{sec:proof_G_stat_comp}

In this section, we first prove Theorem \ref{thm:G_stat_comp} by considering the perspective of a single agent $i\in \mathcal{N}$. We drop the subscript $i$ for simplicity of notation.

\subsubsection{Lower Bound on Performances in Hindsight.}

Let $\Psi^H$ denote the expected Lagrangian dual function associated with the lower bound on the total cost in hindsight $\underline{\mathcal{C}}^{H}$, defined for all $\mu\geq 0$ as
\begin{equation*}
    \begin{aligned}
        \Psi^H(\mu) 
        &= \mathbb{E}_{v, d} \Big[\delta_{t}^H \big(v_{t}, d_{t}^\gamma, \mu \big) \Big] \\
        &=  \mathbb{E}_{v, d} \Big[v  - \mu \left(\rho +  \frac{\gamma}{N} d^\gamma \right) - (\Delta v - \mu d^\gamma )^+\Big].
    \end{aligned}
\end{equation*}
Let $\mu^{\star H}$ denote its maximizer (the unicity is implied by Assumption~\ref{ass:G_stat_comp_L_mono}).
We lower-bound the lowest cost in hindsight using $\Psi^H$ in the following manner:
\begin{equation}\label{equ:G_stat_comp_1}
 \begin{aligned}
     &\mathbb{E}_{\bm{v}, \bm{d}} \left[\mathcal{C}^{H} (\bm{v}, \bm{d})  \right] 
     \stackrel{(a)}{\geq} \mathbb{E}_{\bm{v}, \bm{d}} \left[\underline{\mathcal{C}}^{H} (\bm{v}, \bm{d})  \right] ,\\
     &\stackrel{(b)}{\geq} \mathbb{E}_{\bm{v}, \bm{d}} \left[ \sup_{\mu\geq 0} \sum_{t=1}^T v_{t} -\mu \left(\rho +  \frac{\gamma}{N} d_t^\gamma \right) -\left(\Delta v_{t} - \mu d_{t}^\gamma \right)^+ \right], \\
     &\stackrel{(c)}{\geq}  \sup_{\mu\geq 0} \mathbb{E}_{\bm{v}, \bm{d}} \left[\sum_{t=1}^T v_{t}  -\mu \left(\rho +  \frac{\gamma}{N} d_t^\gamma \right) - \left(\Delta v_{t} - \mu d_{t}^\gamma\right)^+ \right], \\
     &\stackrel{(d)}{=}  \sup_{\mu \geq 0} \sum_{t=1}^T \Psi^H (\mu),\\
    & := T \Psi^H \big(\mu^{\star H}\big).
 \end{aligned}
 \end{equation}
    Inequality $(a)$ follows from Equation~\eqref{equ:G_lowerbound_cost_H} and inequality $(b)$ from Equation~\eqref{equ:G_stat_comp_hindsight-dual}. 
    Inequality $(c)$ holds since \\$\mathbb{E}_X [\sup_y f_y(X)] \geq \sup_y \mathbb{E}_X[ f_y(X)]$. Indeed, $\mathbb{E}_X[\sup_y f_y(X)] \geq \mathbb{E}_X[f_y(X)]$ holds for all $y$, hence taking $\sup_y$ on the right hand side gives the desired result.
    Finally, inequality $(d)$  uses the linearity of expectation.

  \subsubsection{Upper Bound on the Expected Total Cost and Cost-per-Auction.}\label{sec:proof_G_stat_comp_tot_cost}

  We upper-bound the total cost of strategy $K$ by considering that the worst always happens after budget depletion. 
    \begin{equation}\label{equ:G_stat_comp_2}
      \begin{aligned}
          \mathbb{E}_{\bm{v},\bm{d}}\left[ \mathcal{C}^K(\bm{v},\bm{d}) \right] 
        &\stackrel{(a)}{\leq} \mathbb{E}_{\bm{v},\bm{d}} \left[ \sum_{t=1}^{ \mathscr{T}} c_t^K \right]  +  \mathbb{E}_{\bm{v},\bm{d}} \left[T - \mathscr{T}  \right]
      \end{aligned}
  \end{equation}
  Inequality $(a)$ is obtained by assuming that the agent always gets maximum utility $v_t =1$ after time $\mathscr{T}$, but never manages to access the fast road.

For the rest of the proof, we only consider $t\leq \mathscr{T}$, where bids $b_{i,t}$ in Equation~\eqref{equ:G-b} are exactly equal to $\frac{\Delta v_{i,t}}{\mu_{i,t}}.$ 
We first rewrite the cost-per-period in the following manner:
\begin{equation}
\label{equ:G_stat_comp_2.5}
\begin{aligned}
    c_t^K 
    &= v_t \left(1 - \Delta \mathds{1}\{ \Delta v_t > \mu_t d_t^\gamma\}\right) ,\\
    &= v_t  - \mu_t g_t - (\Delta v_t - \mu_t d_t^\gamma)^+ +\mu_t \left(g_t  - d_t^\gamma \mathds{1}\{\Delta v_t > \mu_t d_t^\gamma \}\right).
\end{aligned}
\end{equation}
Noting that the first three terms are reminiscent of the dual objective, while the two last ones correspond to the period's karma losses $z_t - g_t$, we further define  $C(\mu) = \mathbb{E}_{v, d}[c^K \lvert \mu]= \Psi^0(\mu) - \mu L(\mu)$ the expected cost-per-period. 
Note that Assumption~\ref{ass:G_stat_comp_L_mono} ensures it is twice differentiable with derivatives $C^{'}(\mu) = - \mu \left( L^{'}(\mu) + G^{'}(\mu) \right) = -\mu Z^{'}(\mu)$ and $C^{''}(\mu) = -\left(Z^{'}(\mu)+ \mu Z^{''}(\mu) \right).$
Note moreover that $C(\mu^{\star 0}) = \Psi^0(\mu^{\star 0})$ by definition of $\mu^{\star 0}.$
Taking expectation 
and making a Taylor expansion in $\mu^{\star0}$ gives the following for some $\zeta$ between $\mu_t$ and $\mu^{\star0}$:
\begin{equation}\label{equ:G_stat_comp_4}
    \begin{aligned}
        \mathbb{E}_{\bm{v}, \bm{d}} \left[ c_t^K \big\vert t\leq \mathscr{T}\right] 
        &= \mathbb{E}_{\bm{v}, \bm{d}}\left[ C(\mu_t) \vert t\leq \mathscr{T} \right],\\
        &= C\big(\mu^{\star 0}\big) + \mathbb{E}_{\bm{v}, \bm{d}}\left[\mu_t - \mu^{\star 0}\vert t\leq \mathscr{T}\right] C^{'}\left(\mu^{\star 0}\right) \\
        &+ \mathbb{E}_{\bm{v}, \bm{d}}\left[\left(\mu_t -\mu^{\star0}\right)^2 \vert t\leq \mathscr{T}\right] \frac{C^{''}(\zeta)}{2},\\
        &\leq \Psi^0\big(\mu^{\star0}\big) +  \left\lvert \mathbb{E}_{\bm{v}, \bm{d}} \left[\mu_t - \mu^{\star0}\vert t\leq \mathscr{T}\right] \right\lvert \mu^{\star0} \overline{Z^{'}} \\
        &+ \mathbb{E}_{\bm{v}, \bm{d}} \left[\left(\mu_t - \mu^{\star0}\right)^2\vert t\leq \mathscr{T}\right] \frac{\overline{\mu}\overline{Z^{''}} + \overline{Z^{'}}}{2},\\
       &= \Psi^0\big(\mu^{\star0}\big) +  r_t \overline{Z^{'}}+ s_t \frac{\overline{\mu}\overline{Z^{''}} + \overline{Z^{'}}}{2},
    \end{aligned}
\end{equation}
where $s_t$ stands for the mean squared error $s_t = \mathbb{E}_{\bm{v}, \bm{d}} \left[(\mu_t-\mu^{\star0})^2 \vert t\leq \mathscr{T} \right] $ and $r_t$ for the absolute mean error $r_t = \mu^{\star0} \left\lvert\mathbb{E}_{\bm{v}, \bm{d}}[\mu_t -\mu^{\star 0} \vert t\leq \mathscr{T} ] \right \lvert $, defined for all $t \in [T].$
 
\subsubsection{Upper Bound on the Absolute Mean Error.} 
Recall that \\$P_{[\underline{\mu}, \overline{\mu}_i ]} (\mu_{t+1}) = \mu_{t+1}$ for all $t < \mathscr{T}$ by definition of the hitting time $\mathscr{T}$, and the multiplier is directly used to shed bids in Equation~\eqref{equ:G-b}.
Taking expectation on the update rule $ \mu_{t+1}= \mu_t + \epsilon (z_t - g_t)$ hence gives the expected loss function, i.e.,
\begin{equation*}
    \begin{aligned}
        \mathbb{E}_{\bm{v}, \bm{d}} \left[ \mu_{t+1} - \mu_{t} \vert t+1 \leq \mathscr{T} \right] = \epsilon L(\mu_t).
    \end{aligned}
\end{equation*}

Reordering terms, subtracting and multiplying by $\mu^{\star 0}$, and making a Taylor expansion in $\mu^{\star0}$ then gives the following:
\begin{equation*}
    \begin{aligned}
        \mu^{\star 0} \mathbb{E}_{\bm{v}, \bm{d}} \left[ \mu_{t+1} - \mu^{\star 0} \vert t+1 \leq \mathscr{T} \right] 
        &= \mu^{\star 0} \mathbb{E}_{\bm{v}, \bm{d}} \left[\mu_t - \mu^{\star 0} \vert t \leq \mathscr{T} \right] \\
        &+ \epsilon \mu^{\star 0} \Big( L^{'}\big(\mu^{\star 0}\big) \mathbb{E}_{\bm{v}, \bm{d}} \left[\mu_t - \mu^{\star 0} \vert t \leq \mathscr{T} \right] \\
        &+ \frac{1}{2} L^{''}(\zeta) \mathbb{E}_{\bm{v}, \bm{d}} \left[ (\mu_t - \mu^{\star 0})^2 \vert t \leq \mathscr{T} \right]  \Big)  .
    \end{aligned}
\end{equation*}
The equality follows from Taylor's Theorem for some $\zeta$ between $\mu_t$ and $\mu^\star$ because $L$ is twice differentiable under Assumption~\ref{ass:G_stat_comp_L_mono}.

Taking absolute values then gives the following 
recursion, i.e.,
\begin{align*}
    r_{t+1} 
       &\leq \left\lvert 1 + \epsilon L^{'}\big(\mu^{\star 0}\big) \right\lvert r_t +  \epsilon \frac{\mu^{\star 0}}{2}     \left(\overline{Z^{''}} + \overline{G^{''}}\right) s_t  , \\
        &\stackrel{(a)}{=} \left( 1 + \epsilon L^{'}\big(\mu^{\star 0}\big) \right) r_t +  \epsilon \frac{\mu^{\star 0}}{2}     \left(\overline{Z^{''}} + \overline{G^{''}}\right) s_t  , \\
       &\stackrel{(b)}{\leq} (1 - \epsilon \lambda ) r_t +  \epsilon   \frac{ \overline{\mu}}{2}  \left(\overline{Z^{''}} + \overline{G^{''}}\right)   s_t  .
\end{align*}
Note that we have  $- \left(\overline{Z}^{'} + \overline{G}^{'} \right) \leq L^{'}(\mu) \leq -\lambda$ for all $\mu \geq 0$ under Assumption~\ref{ass:G_stat_comp_L_mono}. 
Hence Equality $(a)$ holds because $1 + \epsilon L^{'}\big(\mu^{\star 0}\big) \geq 1- \epsilon \left(\overline{Z}^{'} + \overline{G}^{'}\right) \geq 0 $ under Assumption~\ref{ass:G_stat_comp_eps}. Inequality $(b)$ on the other hand uses that $1 + \epsilon L^{'}(\mu^{\star 0}) \leq 1 - \epsilon \lambda$ and the bound on $\mu^{\star 0}$ from Assumption~\ref{ass:G_stat_comp_mu^*}.

Since $\lambda \epsilon \leq 1$  and $r_1 \leq \overline{\mu}^2$ under Assumption~\ref{ass:G-parameters}, \cite[Lemma C.4]{balseiro2019learning} yields
\begin{equation*}
    r_t 
    \leq \overline{\mu}^2 \left(1 - \epsilon \lambda \right)^{t-1} + \epsilon \frac{\overline{\mu}}{2} \left(\overline{Z^{''}} + \overline{G^{''}}\right)  \sum_{j=1}^{t-1} (1 - \epsilon\lambda)^{t-1-j} s_j ,
\end{equation*}

Noting that the partial sums in $(1-\epsilon \lambda)^t$ are smaller than their limit $1/(\epsilon \lambda)$, we finally get the following upper bound on the total absolute error, i.e.,
\begin{equation} \label{equ:G_stat_comp_6}
    \begin{aligned}
         \sum_{t=1}^{\mathscr{T} } r_t 
        &\leq \overline{\mu}^2 \sum_{t=1}^{\mathscr{T} } (1-\epsilon \lambda)^{t-1} + \epsilon \frac{\overline{\mu}}{2} \left(\overline{Z^{''}} + \overline{G^{''}}\right)\sum_{t=1}^{\mathscr{T} } \sum_{j=1}^{t-1} (1-\epsilon\lambda)^{t-1-j}  s_j  ,\\
        &\leq \frac{\overline{\mu}^2}{\epsilon\lambda} + \epsilon \frac{\overline{\mu}}{2} \left(\overline{Z^{''}} + \overline{G^{''}}\right) \sum_{j=1}^{\mathscr{T} -1}  s_j  \sum_{t=0}^{\mathscr{T} -1-j}  (1-\epsilon\lambda)^{t} ,\\
        &\leq \frac{\overline{\mu}^2}{\epsilon\lambda} 
        + \frac{\overline{\mu}}{2\lambda} \left(\overline{Z^{''}} + \overline{G^{''}}\right) \sum_{j=1}^{\mathscr{T}} s_j.
    \end{aligned}
\end{equation}

The same trick combined with Lemma~\ref{lem:G_bound_MSE} also gives
\begin{equation}\label{equ:G_stat_comp_6.5}
    \begin{aligned}
    \sum_{j=1}^{\mathscr{T}} s_j 
    &\leq   \mathscr{T} \frac{ \Delta^2 \epsilon}{2 \lambda \underline{\mu}^2  } + \overline{\mu}^2  \sum_{j=1}^{\mathscr{T}} (1 - 2\lambda \epsilon)^{t-1},\\
    &\leq   \mathscr{T} \frac{ \Delta^2 \epsilon}{2 \lambda \underline{\mu}^2  } + \frac{\overline{\mu}^2 }{2 \lambda \epsilon} .
    \end{aligned}
\end{equation}

\subsubsection{Bound on the Difference of Expected Dual Objectives.}

We now need to bound the difference $ \left\vert\Psi^0 \big(\mu^{\star0}\big) - \Psi^H\big(\mu^{\star H}\big)\right\vert.$
We first relate the properties known for $\Psi^0$ to $\Psi^H.$

By linearity of the differentiability operator, we first get that $\Psi^H$ is differentiable under Assumption~\ref{ass:G_stat_comp_L_mono} since $d^\gamma$ is bounded under Assumption~\ref{ass:G_stat_comp_distr_comp_bids}. 
Its derivative is defined for all $\mu\geq 0$ by $\Psi^{'H}(\mu) = Z(\mu) - \frac{  \gamma}{N}  \mathbb{E}_{ \bm{d}}\left[d^{\gamma}\right] - \rho = L^H(\mu) - \rho$, where we recognize the analogous expected karma loss function $L^H.$ 

Since  $-G^{'}$ gives a positive contribution to the derivative of $ L$ for all $\mu \geq 0$ when the distribution of valuation $\bm{\mathcal{V}}$ is absolutely continuous
(c.f. Equation~\eqref{equ:G_stat_compt_diff_conc_0}), Assumption~\ref{ass:G_stat_comp_L_mono} also implies that $Z$ is strictly decreasing with parameter $\lambda$, hence $L^H$ is both continuous and $\lambda$-strictly decreasing. 
Finally, note that Assumption~\ref{ass:G_stat_comp_distr_comp_bids} implies $L^H(0) = \left(1 - \frac{ \gamma}{N}\right) \mathbb{E}_{\bm{v}, \bm{d}}[ d^{\gamma}] >\rho$; moreover, we have $L^H \big(\mu^{\star0} \big) = L \big(\mu^{\star0}\big) - \frac{ \gamma}{N} \mathbb{E}_{\bm{v}, \bm{d}}\left[d^{\gamma} - p^{\gamma+1}(\mu^{\star0})\right] \leq 0 < \rho $ since $d^{\gamma} \geq p^{\gamma+1}\big(\mu^{\star0}\big)$ and $\rho$ is strictly positive by Assumption~\ref{ass:G_stat_comp_distr_comp_bids}.

The intermediate value theorem hence implies that $\mu^{\star H}$ the maximizer of $\Psi^H$ is such that $L^H\big(\mu^{\star H}\big) = \rho$ lives in the open interval $ ]0,\mu^{\star0}[.$
We can then bound the distance between $\mu^{\star0}$ and $\mu^{\star H}$ using the strong monotonicity of $L^H$, i.e.,
\begin{equation}\label{equ:G_stat_comp_7}
\begin{aligned}
    0 \leq \mu^{\star0} - \mu^{\star H} 
    &\leq \frac{1}{\lambda} \left( L^H \big(\mu^{\star H } \big) - L^H \big(\mu^{\star 0}\big) \right) ,\\
    & = \frac{1}{\lambda} \left(\rho - \frac{ \gamma}{\lambda N} \mathbb{E}_{ \bm{d}} \left[d^{\gamma} - p_t^{\gamma+1}\big(\mu^{\star0}\big) \right] \right) , \\
    &\leq \frac{1}{\lambda} (\rho + \hat{\varepsilon}).
\end{aligned}
\end{equation}
The last inequality holds since $d^{\gamma +1} \leq p_t^{\gamma+1} \big(\mu^{\star0} \big) \leq d^{\gamma}.$

We now turn to bound the difference of dual functions. First, note that: 
\begin{equation}\label{equ:G_stat_comp_10}
    \begin{aligned}
         &\mathbb{E}_{\bm{v}, \bm{d}} \left[ \left(\Delta v - \mu^{\star H} d^\gamma \right)^+ - \left(\Delta v - \mu^{\star0} d^\gamma \right)^+\right] \\
         &= \mathbb{E}_{\bm{v}, \bm{d}} \left[ \left(\mu^{\star 0}  - \mu^{\star H}\right) d^\gamma \mathds{1} \left\{ \frac{\Delta v}{d^\gamma} \geq \mu^{\star0} \geq \mu^{\star H} \right\} \right],\\
         &+ \mathbb{E}_{\bm{v}, \bm{d}} \left[ \left(\Delta v - \mu^{\star H} d^\gamma \right)^+\mathds{1} \left\{ \mu^{\star0}  > \frac{ \Delta v}{d^\gamma} \geq \mu^{\star H}  \right\} \right] ,\\
         &\leq \frac{\Delta }{\underline{\mu}} \left(\mu^{\star0} - \mu^{\star H} \right) + \Delta  \mathbb{P} \left[ \mu^{\star 0}  > \frac{ \Delta v}{d^\gamma} \geq \mu^{\star H} \right].
    \end{aligned}
\end{equation}
We express this probability explicitly by conditioning on the value of $d^\gamma$. Let $H:d^\gamma \mapsto H(d^\gamma)$ denote the cumulative probability function of the competing bid $d^\gamma.$
\begin{equation}\label{equ:G_stat_comp_11}
\begin{aligned}
    \mathbb{P} \left[\mu^{\star 0}  > \frac{ \Delta v}{d^\gamma} \geq \mu^{\star H} \right]
    &= \int_0^{\Delta /\underline{\mu}} \int_{y \mu^{\star H} / \Delta}^{y \mu^{\star 0} / \Delta} \nu(x) dH(y) dx  \\
    &\stackrel{(a)}{\leq } \frac{\overline{\nu} }{\Delta} \left(\mu^{\star 0} - \mu^{\star H} \right) \int_0^{\Delta /\underline{\mu}} y dH(y)  \\
    &\leq  \frac{\overline{\nu} }{\underline{\mu}} \left(\mu^{\star 0} - \mu^{\star H} \right)
\end{aligned}
\end{equation} 
Inequality $(a)$ uses the change of variable $z=\frac{x}{y}$ as well as the absolute continuity of the distribution of valuations $\bm{\mathcal{V}}$.

Together, Equations~\eqref{equ:G_stat_comp_7}, \eqref{equ:G_stat_comp_10} and~\eqref{equ:G_stat_comp_11} finally give the following bound on the difference of expected dual objective functions, i.e.,
\begin{equation} \label{equ:G_stat_comp_12}
    \begin{aligned}
        \left\vert \Psi^0\big(\mu^{\star0}\big) - \Psi^H \big(\mu^{\star H}\big) \right \vert
        &= \bigg\vert \mathbb{E}_{\bm{v}, \bm{d}} \bigg[ \left(\Delta v - \mu^{\star H} d^\gamma \right)^+ - \left(\Delta v - \mu^{\star 0} d^\gamma \right)^+ \\
        &+   \frac{\gamma}{N} \left( \mu^{\star H} d^\gamma - \mu^{\star0} p^{\gamma+1} \big(\mu^{\star0}\big) \right)   + \mu^{\star H} \rho \bigg]\bigg \vert ,\\
        &\stackrel{(a)}{\leq } \frac{\Delta}{\underline{\mu}} (1+ \overline{\nu}) \big(\mu^{\star 0} - \mu^{\star H} \big) 
        + \frac{\gamma}{N} \Big( \big( \mu^{\star0}- \mu^{\star H} \big)\\
        &\times \mathbb{E}_{\bm{d}} \left[ d^\gamma  \right] + \mu^{\star 0} \mathbb{E}_{\bm{v}, \bm{d}} \left[ d^\gamma - p^{\gamma+1} \big(\mu^{\star0}\big) \right]   \Big) + \mu^{\star H} \rho ,\\
        &\stackrel{(b)}{\leq } \left( \overline{\mu} + \frac{\Delta}{\lambda \underline{\mu}} \left(1+ \overline{\nu} + \frac{ \gamma }{ N }\right)  \right) (\rho + \hat{\varepsilon}).
    \end{aligned}
\end{equation}
Inequality $(a)$ follows from adding and removing $\mu^{\star 0} d^\gamma$, using the linearity of expectation, the triangle inequality and Equations~\eqref{equ:G_stat_comp_10} and~\eqref{equ:G_stat_comp_11};
inequality $(b)$ uses Assumptions~\ref{ass:G_stat_comp_mu^*} and \ref{ass:G_stat_comp_distr_comp_bids}, as well as Equation~\eqref{equ:G_stat_comp_7}.

\subsubsection{Intermediate Conclusion.}

Together, Equations~\eqref{equ:G_stat_comp_1}, \eqref{equ:G_stat_comp_2}, \eqref{equ:G_stat_comp_4}, \eqref{equ:G_stat_comp_6}, 
\eqref{equ:G_stat_comp_6.5}, and 
~\eqref{equ:G_stat_comp_12} finally lead to the following, i.e.,
\begin{equation*}
    \begin{aligned}
        &\mathbb{E}_{\bm{v}, \bm{d}} \left[ \mathcal{C}^{K} (\bm{v}, \bm{d}) - \mathcal{C}^H (\bm{v}, \bm{d})  \right] \\
        &\leq \mathbb{E}_{\bm{v}, \bm{d}} \left[ \sum_{t=1}^{ \mathscr{T}} c_t^K \right]  
        +   \mathbb{E}_{\bm{v}, \bm{d}} \left[T - \mathscr{T}  \right] - T \Psi^H(\mu^{\star H}) ,\\
        & \leq
         T \left( \Psi^0 \big(\mu^{\star 0}\big) - \Psi^H \big(\mu^{\star H}\right) \big)
        +  \mathbb{E}_{\bm{v}, \bm{d}} \left[T - \mathscr{T}  \right] 
        + \frac{ \overline{\mu}^2 \overline{Z^{'}}}{\epsilon\lambda}  \\
        &+ \left(\frac{\overline{\mu} \overline{Z^{'}}}{2\lambda} \left(\overline{Z^{''}} + \overline{G^{''}}\right) + \frac{\overline{\mu}\overline{Z^{''}} + \overline{Z^{'}}}{2} \right) \mathbb{E}_{\bm{v}, \bm{d}} \left[\sum_{t=1}^{ \mathscr{T}}  s_t \right],\\ 
        & \leq
        T \left( \overline{\mu} + \frac{\Delta}{\lambda \underline{\mu}} \left(1+ \overline{\nu} + \frac{ \gamma }{ N }\right)  \right) (\rho + \hat{\varepsilon} )
        +  \mathbb{E}_{\bm{v}, \bm{d}} \left[T - \mathscr{T}  \right] 
        + \frac{ \overline{\mu}^2 \overline{Z^{'}}}{\epsilon\lambda} \\
        &+ \left(\frac{\overline{\mu} \overline{Z^{'}}}{2\lambda} \left(\overline{Z^{''}} + \overline{G^{''}}\right) + \frac{\overline{\mu}\overline{Z^{''}} + \overline{Z^{'}}}{2} \right) \left(T \frac{ \Delta^2 \epsilon}{2 \lambda \underline{\mu}^2  } + \frac{\overline{\mu}^2 }{2 \lambda \epsilon} \right).
    \end{aligned}
\end{equation*}

Rewriting $\rho = k_1/T$, we hence proved the existence of a constant $C$ in $\mathbb{R}$ such that
\begin{equation*}
        \dfrac{1}{T} \mathbb{E}_{\bm{v}, \bm{d}} \left[ \mathcal{C}^{K} (\bm{v}, \bm{d}) - \mathcal{C}^H (\bm{v}, \bm{d})  \right] 
        \leq  C \left(\epsilon  +  \frac{ 1 + k_1 \epsilon  }{\epsilon T} 
        + \frac{\mathbb{E}_{\bm{v}, \bm{d}} \left[T - \mathscr{T}  \right]}{T} + \hat{\varepsilon} \right).
\end{equation*}

\subsubsection{Control of Hitting Time.}
Note again that the hitting time $\mathscr{T}$ defined in Equation~\eqref{equ:G-hit-time} is a stronger notion than the karma depletion time considered in the artificial currency setting, c.f. Section~\ref{sec:A}.
We finally prove in the following that $\mathscr{T} =T$ under Assumption~\ref{ass:G_stat_comp_hitting_time}.

\smallskip
We first show that $\mathscr{T}^{\overline{\mu}} = T$, i.e., that $\mu_t \leq \overline{\mu}$ for all $1\leq t \leq T$. 
By contradiction, suppose there exists a first time $2 \leq t \leq T$ where $\mu_t > \overline{\mu}$. By the pigeonhole principle, $\mu_{t-1}$ must belong in $\left [ \overline{\mu} - \epsilon \overline{d}, \overline{\mu}\right]$ and the agent must have won the auction. The latter is however not possible because we have: $$b_{t-1} = \frac{\Delta v_{t-1}}{\mu_{t-1}} \leq \frac{\Delta }{\overline{\mu} - \epsilon \overline{d}} < \underline{d},$$ where the last inequality follows from $\epsilon <  \frac{1}{\overline{d}}\left( \overline{\mu} -\frac{\Delta}{\underline{d}} \right)$
in Assumption~\ref{ass:G_stat_comp_hitting_time}.

\smallskip
Similarly, we prove that $\mathscr{T}^{\underline{\mu}} = T$, i.e. $\mu_t \geq \underline{\mu}$ for all $1\leq t \leq T$.
Indeed, suppose $\mu_t$ belongs in $\left[ \underline{\mu}, \underline{\mu} + \epsilon \overline{d} \right]$. We then have:
$$b_t = \frac{\Delta v_{t}}{\mu_{t}} \geq \frac{\Delta \underline{v}}{\underline{\mu} + \epsilon \overline{d}} > \overline{d},$$ where the last inequality follows from $ \epsilon <   \frac{\Delta \underline{v}}{\overline{d}^2} - \underline{\mu}$ in Assumption~\ref{ass:G_stat_comp_hitting_time}. Hence, the auction is lost, and $\mu_{t+1} > \mu_t.$ The pigeonhole principle then implies the desired result.

\smallskip
We finally prove that $\mathscr{T}^{k} = T$ and that the budget is never depleted. Indeed, suppose $t$ is such that $k_t \leq \Delta / \underline{\mu}.$ Then we have: 
$$\mu_t  =  \mu_1 + \epsilon( k_1 - k_t) > \overline{\mu},$$ using $k_{i,1} > \frac{ \overline{\mu} - \mu_{i,1}}{\epsilon} + \frac{\Delta}{\underline{\mu}} $ in Assumption~\ref{ass:G_stat_comp_hitting_time}. This is impossible since we previously proved that $\mu_t \leq \overline{\mu}$ for all $1\leq t \leq T$. 

Given the definition of $\mathscr{T}_i$ in Equation~\eqref{equ:G-hit-time}, these three properties together imply that $\mathscr{T}_i = T.$

\subsection{Proof of Theorem~\ref{thm:G_sim_lear_cv}}\label{sec:G_sim_lear_cv_proof}

We start the proof of Theorem~\ref{thm:G_sim_lear_cv} by showing that the multiplier profile $\bm{\mu_t}$  converges on average to the stationary multiplier $\bm{\mu^{\star0}}.$ 

\subsubsection{Convergence on Average of the Multiplier Profile.}

Summing Lemma~\ref{lem:G_bound_MSE} gives:
\begin{equation*}
    \begin{aligned}
        &\frac{1}{T} \sum_{t=1}^T \mathbb{E}_{\bm{v}} \left[ \left\Vert \bm{\mu_t} - \bm{\mu^{\star0}} \right\Vert_2^2 \right]\\
        &= \frac{1}{T} \mathbb{E}_{\bm{v}} \left[ \sum_{t=1}^T  \left\Vert \bm{\mu_t} - \bm{\mu^{\star0}} \right\Vert_2^2 \mathds{1}\{t\leq \underline{\mathscr{T}} \} + \left\Vert \bm{\mu_t} - \bm{\mu^{\star0}} \right\Vert_2^2 \mathds{1}\{t > \underline{\mathscr{T}} \} \right],\\
        &\stackrel{(a)}{\leq} \frac{1}{T} \mathbb{E}_{\bm{v}} \left[ \sum_{t=1}^{\underline{\mathscr{T}}} s_t \right] + \frac{N \overline{\mu}^2}{T} \mathbb{E}_{\bm{v}} \left[ T - \underline{\mathscr{T}}  \right], \\ 
        &\stackrel{(b)}{\leq} \frac{1}{T}   \sum_{t=1}^{\underline{\mathscr{T}}} \left[ N \overline{\mu}^2   (1 - 2\lambda \epsilon )^{t-1}    + \frac{N \Delta^2 \epsilon}{2\lambda \underline{\mu}^2}  \right]
        + \frac{N \overline{\mu}^2}{T} \mathbb{E}_{\bm{v}} \left[ T - \underline{\mathscr{T}}  \right], \\ 
        &\stackrel{(c)}{\leq} \frac{N \overline{\mu}^2 }{2\lambda \epsilon T} 
        + \frac{N \Delta^2 \epsilon}{2\lambda \underline{\mu}^2}
        + \frac{N \overline{\mu}^2}{T} \mathbb{E}_{\bm{v}} \left[ T - \underline{\mathscr{T}}  \right] .
    \end{aligned}
\end{equation*}
    Inequality $(a)$ holds by assuming that the error is maximal after time $\underline{\mathscr{T}}$;
    inequality $(b)$ uses Lemma~\ref{lem:G_bound_MSE}, and  inequality $(c)$ holds by bounding the partial geometric series in $1- 2\lambda \epsilon$ by its limit $1/(2\lambda \epsilon).$
    
    Hence there exist some constant $C_1$ in $\mathbb{R_+}$ such that
    \begin{equation*}
    \frac{1}{T} \sum_{t=1}^T \mathbb{E}_{\bm{v}} \left[\big\Vert \bm{\mu_t} - \bm{\mu^{\star0}} \big\Vert_2^2 \right] 
    \leq  C_1 N \left( \epsilon + \frac{1}{\epsilon T} 
    +  \frac{\mathbb{E}_{\bm{v}} \left[ T - \underline{\mathscr{T}}  \right]}{T} \right).
    \end{equation*}

In order to prove the second bound of Theorem~\ref{thm:G_sim_lear_cv}, we need to establish the local Lipschitz continuity of the expected loss function $L_i$ and the expected dual function $\Psi_i^0.$ We start by showing the continuity of the expected gain function $G.$

\subsubsection{Lipschitz Continuity of the Expected Gain}\label{sec:G_sim_lear_lipschitz_G}

We consider the more general setting with $M$ parallel auctions. Choosing the probability of agent $i$ and $j$ to play in the same auction $a_{i,j} =1$ equal to one for all pairs gives the setting considered in Section~\ref{sec:G_sim_lear}.
We consider the competing price $d_i^\gamma = \gamma\textsuperscript{th}\mbox{-}\max_{j:j\neq i} \big\{ \mathds{1}\{ m_j = m_i\}  \Delta v_j / \mu_j \big\}$, as well as the gain $g_i = \frac{\gamma}{N} \sum_{m=1}^M p_m^{\gamma+1}$, where $p_m^{\gamma+1} = {\gamma+1}\textsuperscript{th} $ $\mbox{-}\max_i \big\{ \mathds{1}\{ m_i = m\}  \Delta v_i / \mu_i \big\}$ denotes the price of the $m \textsuperscript{th}$ auction.
Note that the gains $g_i $ are the same for everyone, hence we drop the subscript $i.$ We first study the expected gain function $G.$

For any $m \in [M]$ and realized vectors  $\bm{v}$ and $\bm{m}$, the function  $\bm{\mu} \mapsto p_{m,t}^{\gamma+1}$ is differentiable in $\mu_i$ outside of sets of measure zero 
since the distribution of valuations $\bm{\mathcal{V}}$ is absolutely continuous.
Note that the derivative is non-null only when agent $i$ is the price setter in auction $m$, and that it is bounded by $\Delta / \underline{\mu}^2$ which is integrable over $[\underline{\mu}, \overline{\mu}].$ Leibniz's integral rule hence implies that $K$ is differentiable with respect to $\mu_i$ and verifies:

\begin{equation*}
\begin{aligned}
    \left\vert \frac{\partial G}{\partial \mu_i} (\bm{\mu})  \right\vert
    &= \left\vert - \mathbb{E}_{\bm{v}, \bm{m}} \left[\frac{\gamma}{N}  \sum_{m=1}^M \frac{\Delta v_i}{\mu_i^2} \mathds{1} \left\{m_i = m \right \} \mathds{1} \left\{ d_i^{\gamma} > \frac{\Delta v_i}{ \mu_i} > d_i^{\gamma+1} \right\} \right] \right\vert ,\\
    &=  \mathbb{E}_{\bm{v}, \bm{m}} \left[\frac{\gamma}{N}   \frac{\Delta v_i}{\mu_i^2} \mathds{1} \left\{ d_i^{\gamma} > \frac{\Delta v_i}{ \mu_i} > d_i^{\gamma+1} \right\} \right] .
\end{aligned}
\end{equation*}
Note that $\mathds{1} \left\{ d_i^{\gamma} > \frac{\Delta v_i}{ \mu_i} > d_i^{\gamma+1} \right\}$ is exactly equal to one when agent $i$ sets the price for its auction. Since it can only happen for $M$ agents simultaneously, we can bound the norm of the Jacobian of $K$ as follows, i.e.,
\begin{equation*}
\begin{aligned}
    \Vert J_G (\bm{\mu}) \Vert_2
    &\leq \Vert J_G (\bm{\mu}) \Vert_1 ,\\
    &=  \sum_{i\in\mathcal{N}} \mathbb{E}_{\bm{v}, \bm{m}} \left[    \frac{\gamma \Delta v_i}{N \mu_i^2}  \mathds{1} \left\{  d_i^{\gamma} > \frac{\Delta v_i}{ \mu_i} > d_i^{\gamma+1} \right\} \right] ,\\
    &\stackrel{(a)}{\leq}  \frac{\gamma \Delta}{N \underline{\mu}^2}  \mathbb{E}_{\bm{v}, \bm{m}} \left[  \sum_{i\in\mathcal{N}} \mathds{1} \left\{  d_i^{\gamma} > \frac{\Delta v_i}{ \mu_i} > d_i^{\gamma+1} \right\} \right] , \\
    &\stackrel{(b)}{\leq} \frac{M \gamma \Delta}{N \underline{\mu}^2}.
\end{aligned}
\end{equation*} 
Inequality $(a)$ uses the linearity of expectation as well as the fact that $\bm{\mu} \in \bm{U}$; inequality $(b)$ uses that at most $M$ agents can be price setters simultaneously.

This finally allows to establish that $K$ is locally Lipschitz continuous:
\begin{equation}\label{equ:G_sim_lear_lipschitz_cont_G}
    \begin{aligned}
        \frac{\vert G(\bm{\mu}) -  G(\bm{\mu'}) \vert}{ \Vert \bm{\mu} - \bm{\mu'} \Vert_2} 
        &\leq \sup_{\bm{\mu}\in\bm{U}} \frac{ \Vert J_G (\bm{\mu}) \bm{\mu} \Vert_2 } {\Vert \bm{\mu} \Vert _2}, \\
        &\stackrel{(a)}{\leq}  \sup_{\bm{\mu}\in\bm{U}} \Vert J_G (\bm{\mu}) \Vert_2 ,\\
        &\leq \frac{M \gamma \Delta}{N \underline{\mu}^2} .
    \end{aligned}
\end{equation}
Inequality $(a)$ uses Cauchy-Schwartz inequality.

\subsubsection{Lipschitz Continuity of the Expected Dual Objective}\label{sec:G_sim_lear_lipschitz_psi}

Again, we consider the more general setting with $M$ parallel auctions. Choosing the probability of agent $i$ and $j$ to play in the same auction $a_{i,j} =1$ equal to one for all pairs gives the setting considered in Section~\ref{sec:G_sim_lear}.
For an agent $i \in \mathcal{N}$, we rewrite the dual function $\Psi_i^0$ using the expected gain function $G_i$ and the dual function $\Psi_i$ from the setting without gains, i.e., $\Psi_i(\bm{\mu}) = \Exp_{\bm{v}} \left[v_i  -\mu_i \rho_i - (\Delta v_i - \mu_i d_i^\gamma )^+\right].$ 
We indeed have $\Psi_i^0(\bm{\mu}) = \Psi_i(\bm{\mu}) + \mu_i ( \rho_i - G(\bm{\mu}))$, hence the Lipschitz continuity of $\Psi_i^0$ on $\bm{U}$ can be deduced from that of $\Psi_i$ and $K$. We thus focus on bounding the derivatives of $\Psi_i$.

For any realized vectors  $\bm{v} = (v_j)_{j=1}^N$ and $\bm{m} = (m_j)_{j=1}^N$, the function  $\bm{\mu} \mapsto( \Delta v_i - \mu_i d_i^\gamma )^+$ is differentiable in $\mu_i$ with derivative bounded by $\Delta / \underline{\mu}$, except in the set $\{ (\bm{v},\bm{m}) : \Delta v_i = \mu_i d_i^\gamma \}$ of measure zero 
since the distribution of valuation $\bm{\mathcal{V}}$ is absolutely continuous.
Leibniz's integral rule hence implies:
\begin{equation*}
    \frac{\partial \Psi_i}{\partial \mu_i} (\bm{\mu}) =  \mathbb{E}_{\bm{v}, \bm{m}} \left[ d_i^\gamma \mathds{1}\{\Delta v_i \geq \mu_i d_i^\gamma \} \right] - \rho_i  .
\end{equation*}
Hence Assumption~\ref{ass:G_sim_lear_mu_star} ensures the derivative is bounded, i.e.,
\begin{equation}\label{G:lipschitz_cont_psi_1}
    \left\lvert \frac{\partial \Psi_i}{\partial \mu_i} (\bm{\mu})\right \lvert 
    \leq \frac{\Delta}{\underline{\mu}}.
\end{equation}

Furthermore, consider $j\neq i$ and let \\
$d_{i,j}^\gamma = \gamma\textsuperscript{th}\mbox{-}\max_{\ell \neq i,j } \left\{ \mathds{1} \{ m_\ell = m_i\} \Delta v_\ell / \mu_\ell\right\}$ be the competing price for $i$ before $j$ places its own bid.
We then rewrite:
\begin{equation*}
    \left( \Delta v_i - \mu_i d_i^\gamma \right)^+ =
    \begin{cases}
        \left( \Delta v_i - \mu_i \frac{\Delta v_j} {\mu_j} \right)^+  &\text{ if } m_j = m_i \text{ and }  d_{i,j}^{\gamma-1} \geq \frac{\Delta v_j}{ \mu_j} \geq d_{i,j}^\gamma, \\
        \left( \Delta v_i - \mu_i d_{i,j}^\gamma \right)^+  &\text{ otherwise.}
    \end{cases}
\end{equation*} 
The function $ \bm{\mu} \mapsto (\Delta v_i - \mu_i d_i^\gamma )^+$ is hence differentiable in $\mu_j$ with derivative bounded by $\Delta \overline{\mu}/ \underline{\mu}^2$ outside of the sets \\
$\left\{ (\bm{v},\bm{m}) : d_{i,j}^{\gamma-1} > \dfrac{\Delta v_i}{\mu_i} = \dfrac{\Delta v_j}{\mu_j} > d_{i,j}^\gamma,~ m_i = m_j \right\} $,\\
$\left\{ (\bm{v},\bm{m}) : \dfrac{\Delta v_i}{\mu_i} \geq \dfrac{\Delta v_j}{\mu_j} = d_{i,j}^{\gamma-1},~ m_i = m_j  \right\}$ 
and \\$\left\{ (\bm{v},\bm{m}) : \dfrac{\Delta v_i}{\mu_i} \geq \dfrac{\Delta v_j}{\mu_j} = d_{i,j}^\gamma,~ m_i = m_j  \right\}$, which are all of measure zero by absolute continuity of the distribution of valuation $\bm{\mathcal{V}}.$ Leibniz's integral rule then implies that $\Psi_i$ is differentiable with respect to $\mu_j$, i.e., 
 \begin{equation*}
         \dfrac{\partial \Psi_i}{\partial \mu_j} (\bm{\mu})
         = \mathbb{E}_{\bm{v}, \bm{m}} \left[ \dfrac{\Delta v_j \mu_i}{\mu_j^2} 
         \mathds{1}\left\{ \frac{v_i}{\mu_i} \geq \frac{v_j}{\mu_j},~ d_{i,j}^{\gamma-1} \geq  \dfrac{\Delta v_j}{\mu_j} \geq d_{i,j}^\gamma,~ m_i = m_j \right\}   \right].
 \end{equation*}
 Using that $v_i/\mu_i \geq v_j/\mu_j$ in the indicator function, we can bound the derivative as follows, i.e., 
 \begin{equation}\label{G:lipschitz_cont_psi_2}
         \left \lvert \dfrac{\partial \Psi_i}{\partial \mu_j} (\bm{\mu})  \right \lvert
         \leq  \mathbb{E}_{\bm{v}, \bm{m}} \left[ \dfrac{\Delta v_i}{\mu_j} \mathds{1}\left\{ m_i = m_j \right\}   \right] \\
         \leq \frac{\Delta }{\underline{\mu}} a_{i,j}.
 \end{equation}

 Together, Equations~\eqref{equ:G_sim_lear_lipschitz_cont_G}, \eqref{G:lipschitz_cont_psi_1}and~\eqref{G:lipschitz_cont_psi_2}  finally imply the Lipschitz continuity of $\Psi_i^0$ on $\bm{U}.$ 
\begin{equation}\label{equ:G_lipschitz_psi}
\begin{aligned}
    \left\lvert \Psi_i^0(\bm{\mu}) - \Psi_i^0(\bm{\mu'}) \right\lvert 
    &\leq \left\lvert \Psi_i(\bm{\mu}) - \Psi_i(\bm{\mu'}) \right\lvert +  \left\lvert \mu_i G(\bm{\mu}) - \mu_i' G(\bm{\mu'}) \right \vert \\
    &+ \rho_i \lvert \mu_i - \mu_i'\lvert ,\\
    &\stackrel{(a)}{\leq} \frac{\Delta}{\underline{\mu}} \left( \lvert \mu_i - \mu_i'\lvert + \sum_{j\neq i} a_{i,j} \lvert \mu_j - \mu_j'\lvert \right)
    + \mu_i \left\lvert G(\bm{\mu}) -  G(\bm{\mu'}) \right \vert \\
    &+ \left\vert G(\bm{\mu'}) \right\vert \cdot \left\vert \mu_i - \mu'_i \right\vert  + \rho_i \lvert \mu_i - \mu_i'\lvert , \\
    &\stackrel{(b)}{\leq}  \frac{\Delta}{\underline{\mu}} \left( \Vert \bm{a_i} \Vert_2 +   \frac{M\gamma  \overline{\mu} }{N \underline{\mu}} \right)    \Vert \bm{\mu} - \bm{\mu'} \Vert_2\\
     &+ \frac{\Delta}{\underline{\mu}} \left( 1 + \rho_i +  \frac{M \gamma  }{N }  \right)\left\vert \mu_i - \mu_i' \right\vert .
\end{aligned}  
\end{equation}
Inequality $(a)$ follows from Equations~\eqref{G:lipschitz_cont_psi_1} and~\eqref{G:lipschitz_cont_psi_2}, as well as by adding and substracting $\mu_i G(\bm{\mu'})$ and using the triangle inequality. Inequality $(b)$ uses Equation~\eqref{equ:G_sim_lear_lipschitz_cont_G}, the bound $G(\bm{\mu'}) \leq \frac{M\gamma \Delta}{N\underline{\mu}}$ for $\bm{\mu'} \in \bm{U}$, as well as Cauchy-Schwartz inequality.

In the case of a single auction, we hence define the Lipschitz constant for $\Psi_i^0$ such that $\left\lvert \Psi_i^0(\bm{\mu}) - \Psi_i^0(\bm{\mu'}) \right\lvert \leq \sqrt{N} \mathcal{L}_\Psi  \Vert \bm{\mu} - \bm{\mu'} \Vert_2$, i.e., $$\mathcal{L}_\Psi =  \frac{\Delta}{\underline{\mu}} \left( 1 + \frac{1}{\sqrt{N}} \left(\frac{ \gamma \overline{\mu}}{N \underline{\mu}} + \rho_i \right)\right).$$

\subsubsection{Bound on the Derivative of the Cumulative Distribution Function of Competing Prices.}

In order to prove the Lipschitz continuity of $L_i$, we first bound the derivative of the cumulative distribution function of competing prices $H_i(x, \bm{\mu_{-i}})  = \mathbb{P} \left[ d_i^\gamma \leq x \right].$ Since \cite[Lemma C.2]{balseiro2019learning} proposes such a bound when $\gamma = 1$, we only consider the case where $\gamma >1.$

Let $\mathscr{V}_j$ denote the cumulative distribution function of valuations associated with the density of valuations $\nu_j$ of agent $j \in \mathcal{N}$, and let $\overline{\mathscr{V}}_j$ denote the function $y \mapsto 1-\mathscr{V}_j(y).$
We further define $\mathcal{M}_i \in 2^{\mathcal{N} \setminus \{i\}}$ as the set of agents different from $i$ that play in the same auction as $i.$
For a fixed $\bm{\mu_{-i}} \in \mathbb{R}_+^{N-1}$, we write the cumulative distribution of competing prices by conditioning on each agent $j \neq i$ being the price setter for agent $i$, an event that we denote  by $\{j \to i\}.$
\begin{equation}\label{equ:G_lipschitz_H}
    \begin{aligned}
        &H_i(x, \bm{\mu_{-i}}) 
        = \mathbb{P} \left[ d_i^\gamma \leq x \right]\\
        &\stackrel{(a)}{=} \sum_{\mathcal{M}_i } \mathbb{P} \left[ \mathcal{M}_i \right] \mathbb{P} \left[ d_i^\gamma \leq x \lvert \mathcal{M}_i \right] \\
        &\stackrel{(b)}{=} \sum_{\substack{\mathcal{M}_i \\ \lvert \mathcal{M}_i \lvert < \gamma}} \mathbb{P} \left[ \mathcal{M}_i \right] 
        + \sum_{\substack{\mathcal{M}_i \\ \lvert \mathcal{M}_i \lvert \geq \gamma}} \sum_{p \in \mathcal{M}_i} \mathbb{P} \left[ \mathcal{M}_i \right]\\
        &\times \mathbb{P} \left[  \left\{ \frac{\Delta v_p} {\mu_p} \leq x \right\} \cap \{p \to i\}\Big\lvert \mathcal{M}_i \right]
        \\
        &\stackrel{(c)}{=} \sum_{\substack{\mathcal{M}_i \\ \lvert \mathcal{M}_i \lvert < \gamma}} \mathbb{P} \left[ \mathcal{M}_i \right]
        + \sum_{\substack{\mathcal{M}_i \\ \lvert \mathcal{M}_i \lvert \geq \gamma}} \sum_{p \in \mathcal{M}_i} \sum_{\substack{W \subset \mathcal{M}_i\setminus \{p\} \\ \lvert W \lvert = \gamma-1 }} 
        \mathbb{P} \left[ \mathcal{M}_i \right] \\
        &\times \mathbb{P} \left[ \left\{\frac{\Delta v_p}{ \mu_p} \leq x \right\} 
        \bigcap_{w \in W} \left\{ \frac{v_w }{ \mu_w} \geq \frac{v_p}{\mu_p} \right\} 
        \bigcap_{\substack{\ell \in \mathcal{M}_i \setminus W\\ \ell \neq p}} \left\{ \frac{v_\ell}{\mu_\ell} \leq \frac{v_p}{\mu_p} \right\} \Bigg\lvert \mathcal{M}_i \right] \\
        &\stackrel{(d)}{=} \sum_{\substack{\mathcal{M}_i \\ \lvert \mathcal{M}_i \lvert < \gamma}} \prod_{k\in \mathcal{M}_i} a_{i,k}
        + \sum_{\substack{\mathcal{M}_i \\ \lvert \mathcal{M}_i \lvert \geq \gamma}} \sum_{p \in \mathcal{M}_i} \sum_{\substack{W \subset \mathcal{M}_i\setminus \{p\} \\ \lvert W \lvert = \gamma-1 }} 
        \prod_{k\in \mathcal{M}_i} a_{i,k} \\
        & \times \int_0^{x \mu_p / \Delta} \nu_p(y) 
        \prod_{w \in W} \overline{\mathscr{V}}_w  \left( \frac{y \mu_w}{\mu_p} \right) 
        \prod_{\substack{\ell \in \mathcal{M}_i \setminus W\\ \ell \neq p}} \mathscr{V}_\ell  \left( \frac{y \mu_\ell}{\mu_p} \right)  dy\\
    \end{aligned}
\end{equation}
Equality $(a)$, $(b)$ and $(c)$ each follow from the Law of total probability, by partitioning on realizations of $M_i$, then price setters $p$, and finally sets of auction winners $W.$ Equality $(d)$ holds since valuations $\bm{v} = (v_j)_{j=1}^N$ and auctions $\bm{m} = (m_i)_{i=1}^N$ are drawn independently across agents.

Now consider an agent $j \neq i.$ For each term of the sums, note that the integrand is differentiable with respect to $\mu_j$ almost everywhere with derivative bounded by $ y \overline{\nu} \overline{\mu} / \underline {\mu}^2 $, which is integrable on $[0, x \overline{\mu} / \Delta]$ for all $x \geq 0$.  Leibniz's integration rule hence implies that $H_i(x, \bm{\mu_{-i}})$ is differentiable in $\mu_j$, i.e.,

\begin{equation}\label{equ:G_lipschitz_deriv_H}
    \frac{\partial  H_i(x, \bm{\mu_{-i}}) }{\partial \mu_j} =  \sum_{\substack{\mathcal{M}_i \\ \lvert \mathcal{M}_i \lvert \geq \gamma \\ j \in \mathcal{M}_i}} 
    A(\mathcal{M}_i) + B(\mathcal{M}_i) + C(\mathcal{M}_i),
\end{equation}
where we regroup in $A(\mathcal{M}_i)$ the terms where $j \in W$, in $B(\mathcal{M}_i)$ the terms verifying $j\notin W \cup \{ p\}$, and in $C(\mathcal{M}_i)$ the terms where $p=j$.
The expressions for $A(\mathcal{M}_i)$ and $B(\mathcal{M}_i)$ are the following, i.e.,
\begin{equation*}
    \begin{aligned}
        A(\mathcal{M}_i) 
        &= -  \sum_{\substack{p \in \mathcal{M}_i \\ p \neq j}} \sum_{\substack{W \subset \mathcal{M}_i\setminus \{p\} \\ \lvert W \lvert = \gamma-1  \\ j\in W}} \prod_{\substack{k\in \mathcal{M}_i}} a_{i,k}
        \int_0^{x \mu_p / \Delta} 
          \frac{y}{\mu_p} \nu_j \left( \frac{y \mu_j}{\mu_p} \right) \nu_p(y)\\
        &\times\prod_{\substack{w \in W \\ w \neq j }} \overline{\mathscr{V}}_w  \left( \frac{y \mu_w}{\mu_p} \right)
        \prod_{\substack{\ell \in \mathcal{M}_i \setminus W\\ \ell \neq p}} \mathscr{V}_\ell  \left( \frac{y \mu_\ell}{\mu_p} \right) dy,\\
        B(\mathcal{M}_i) 
        &=  \sum_{\substack{p \in \mathcal{M}_i \\ p \neq j}} 
        \sum_{\substack{W \subset \mathcal{M}_i\setminus \{p\} \\ \lvert W \lvert = \gamma-1 \\ j \notin W}}
        \prod_{\substack{k\in \mathcal{M}_i}} a_{i,k}
        \int_0^{x \mu_p / \Delta} 
        \frac{y}{\mu_p} \nu_j \left( \frac{y \mu_j}{\mu_p} \right) \nu_p(y)\\
        &\times\prod_{w \in W} \overline{\mathscr{V}}_w  \left( \frac{y \mu_w}{\mu_p} \right)
        \prod_{\substack{\ell \in \mathcal{M}_i \setminus W\\ \ell \neq j, p}} \mathscr{V}_\ell  \left( \frac{y \mu_\ell}{\mu_p} \right) dy .
    \end{aligned}
\end{equation*}

Moreover, we use the product rule in the expression of $C(\mathcal{M}_i)$ and further regroup in $C_1(\mathcal{M}_i)$ the terms of the product rule where the running index $k$ belongs to $W$, and in $C_2(\mathcal{M}_i)$ those where $k$ belongs to $\mathcal{M}_i \setminus W \cup \{j\}.$

\begin{equation*}
    \begin{aligned}
        C(\mathcal{M}_i) 
        &= \big( C_1(\mathcal{M}_i) + C_2(\mathcal{M}_i) \big) \prod_{\substack{k\in \mathcal{M}_i}} a_{i,k}, \\
        C_1(\mathcal{M}_i) 
        &=  \sum_{\substack{k \in\mathcal{M}_i \\ k \neq j}} \sum_{\substack{W \subset \mathcal{M}_i\setminus \{j\} \\ \lvert W \lvert = \gamma-1 \\ k \in W}}
        \int_0^{x \mu_j / \Delta} 
        \frac{y \mu_k}{\mu_j^2} \nu_k \left( \frac{y \mu_k}{\mu_j} \right) \nu_j(y) \\
        &\times\prod_{\substack{w \in W \\ w\neq k}} \overline{\mathscr{V}}_w  \left( \frac{y \mu_w}{\mu_j} \right)
        \prod_{\substack{\ell \in \mathcal{M}_i \setminus W\\ \ell \neq j}} \mathscr{V}_\ell  \left( \frac{y \mu_\ell}{\mu_j} \right) dy,\\
        C_2(\mathcal{M}_i) 
        &=- \sum_{\substack{k \in\mathcal{M}_i \\ k \neq j}} \sum_{\substack{W \subset \mathcal{M}_i\setminus \{j\} \\ \lvert W \lvert = \gamma-1 \\ k \notin W}}
        \int_0^{x \mu_j / \Delta} 
        \frac{y \mu_k}{\mu_j^2} \nu_k \left( \frac{y \mu_k}{\mu_j} \right) \nu_j(y) \\
        &\times \prod_{\substack{w \in W }} \overline{\mathscr{V}}_w  \left( \frac{y \mu_w}{\mu_j} \right)
        \prod_{\substack{\ell \in \mathcal{M}_i \setminus W\\ \ell \neq j,k}} \mathscr{V}_\ell  \left( \frac{y \mu_\ell}{\mu_j} \right) dy .\\
    \end{aligned}
\end{equation*}


We now proceed to bound the sum of the terms in $A(\mathcal{M}_i)$ in Equation~\eqref{equ:G_lipschitz_deriv_H}. First, we write the cumulative distribution of $d_{i,j}^{\gamma-1} = \gamma -1 \textsuperscript{th}\mbox{-}\max_{\ell \neq i,j } \left\{ \mathds{1} \{ m_\ell = m_i\} \Delta v_\ell / \mu_\ell \right\}$ in a similar manner as in Equation~\eqref{equ:G_lipschitz_H}. 
\begin{equation}\label{eq:G_lipschitz_expression_d_i,j_gamma-1}
    \begin{aligned}
        \mathbb{P} \left[ d_{i,j}^{\gamma-1} \leq x \right]
        &= \sum_{\substack{\mathcal{M}_i \\ \lvert \mathcal{M}_i \lvert < \gamma -1  \\ j \notin M_i}} \prod_{k\in \mathcal{M}_i} a_{i,k}
        + \sum_{\substack{\mathcal{M}_i \\ \lvert \mathcal{M}_i \lvert \geq \gamma-1   \\ j \notin M_i}} \sum_{p \in \mathcal{M}_i} \sum_{\substack{W \subset \mathcal{M}_i\setminus \{p\} \\ \lvert W \lvert = \gamma - 2 }} \\ 
        &\prod_{k\in \mathcal{M}_i} a_{i,k}  \int_0^{x \overline{\mu}_p / \Delta} \nu_p(y) 
        \prod_{w \in W} \overline{\mathscr{V}}_w  \left( \frac{y \mu_w}{\mu_p} \right) \\ 
        &\times \prod_{\substack{\ell \in \mathcal{M}_i \setminus W\\ \ell \neq p}} \mathscr{V}_\ell  \left( \frac{y \mu_\ell}{\mu_p} \right)  dy\\
    \end{aligned}
\end{equation}

Now consider the term $S$ in the expression of $A$ defined by the realization $(\mathcal{M}_i, p,W)$, and denote $\mathcal{M}_i' = \mathcal{M}_i \setminus \{j \}$ as well as $W' = W \setminus \{j\}.$ Note that $(\mathcal{M}_i',p,W')$ defines a term $S'$ in Equation~\eqref{eq:G_lipschitz_expression_d_i,j_gamma-1}. 
Looking at the expression of $S$ and $S'$, we verify that multiplying the integrand of $S$ by a factor $- a_{i,j} \frac{y }{\mu_p} \nu_j\left( \frac{y \mu_j }{\mu_p}\right)$ gives the integrand of $S'$.
Since $y \leq x \mu_p / \Delta$, this factor is bounded by $a_{i,j} x \overline{\nu} /\Delta $ by absolute continuity of the distribution of valuations $\bm{\mathcal{V}}$,
in turn bounded by $a_{i,j} \overline{\nu} / \underline{\mu}$ since $x \leq \Delta / \underline{\mu}.$ 
Since this bound holds for all terms $S$, it holds for the sum. Noting that $\mathbb{P} \left[ d_{i,j}^{\gamma-1} \leq x \right] \leq 1$ gives the following bound on $A$,
\begin{equation}\label{equ:G_lipschitz_bound_A}
    \Bigg\lvert \sum_{\substack{\mathcal{M}_i \\ \lvert \mathcal{M}_i \lvert \geq \gamma \\ j \in \mathcal{M}_i}} 
    A(\mathcal{M}_i)\Bigg\lvert
    \leq a_{i,j} \frac{\overline{\nu}}{\underline{\mu}} \mathbb{P} \left[ d_{i,j}^{\gamma-1} \leq x \right] 
    \leq a_{i,j} \frac{\overline{\nu}}{\underline{\mu}}.
\end{equation}

Furthermore, note that rewriting $k=p$ and making the change of variable $z = y \mu_p / \mu_j$ gives the following expression for $C_1(\mathcal{M}_i)$, i.e.,
\begin{equation*}
    \begin{aligned}
    C_1(\mathcal{M}_i) 
    &=  \sum_{\substack{p \in\mathcal{M}_i \\ p \neq j}} \sum_{\substack{W \subset \mathcal{M}_i\setminus \{j\} \\ \lvert W \lvert = \gamma-1 \\ p \in W}}
    \int_0^{x \mu_p / \Delta} 
    \frac{z }{\mu_p} \nu_p (z) \nu_j \left( \frac{z \mu_j}{\mu_p} \right) \\
    &\times \prod_{\substack{w \in W \\ w\neq p}} \overline{\mathscr{V}}_w  \left( \frac{z \mu_w}{\mu_p} \right)
    \prod_{\substack{\ell \in \mathcal{M}_i \setminus W\\ \ell \neq j}} \mathscr{V}_\ell  \left( \frac{z \mu_\ell}{\mu_p} \right) dz.
    \end{aligned}
\end{equation*}

Now consider the term $S$ in the expression of $C_1$ defined by the realization $(\mathcal{M}_i, p, W)$, and denote $\mathcal{M}_i' = \mathcal{M}_i' \setminus \{j \}$ as well as $W' = W \setminus \{p\}.$ Note that $(\mathcal{M}_i,p,W')$ also defines a term $S'$ in Equation~\eqref{eq:G_lipschitz_expression_d_i,j_gamma-1}. The rest of the argument is similar as for $A$ and we obtain:
\begin{equation}\label{equ:G_lipschitz_bound_C_1}
    \Bigg\lvert \sum_{\substack{\mathcal{M}_i \\ \lvert \mathcal{M}_i \lvert \geq \gamma \\ j \in \mathcal{M}_i}} 
    C_1(\mathcal{M}_i)  \prod_{\substack{k\in \mathcal{M}_i}} a_{i,k} \Bigg\lvert
    \leq a_{i,j} \frac{\overline{\nu}}{\underline{\mu}}.
\end{equation}

We now tackle the terms in $B(\mathcal{M}_i)$ in Equation~\eqref{equ:G_lipschitz_deriv_H}. First, note that $B(\mathcal{M}_i) = 0$ whenever $\lvert \mathcal{M}_i \lvert = \gamma$, since there are no set $W$ of size $\gamma - 1$ that does not contain either $j$ or $p.$ If we now consider the term $S$ in the expression of $B$ defined by the realization $(\mathcal{M}_i, p,W)$verifying $\lvert \mathcal{M}_i \lvert \geq \gamma +1$, and denote $\mathcal{M}_i' = \mathcal{M}_i \setminus \{j \}$ as well as $W' = W \setminus \{j\}$, it appears that $(\mathcal{M}_i',p,W')$ defines a term $S'$ in the expression of $\mathbb{P} \left[ d_{i,j}^{\gamma-1} \leq x \right]$, written similarly as in Equation~\eqref{eq:G_lipschitz_expression_d_i,j_gamma-1}. 
Once again looking at the expression of $S$ and $S'$, we verify that multiplying the integrand of $S$ by a factor $ a_{i,j} \frac{y }{\mu_p} \nu_j\left( \frac{y \mu_j }{\mu_p}\right)$ gives the integrand of $S'$. Following the same argument as for $A$, we obtain the following bound for terms in $B$, i.e.,
\begin{equation}\label{equ:G_lipschitz_bound_B}
    \Bigg\lvert \sum_{\substack{\mathcal{M}_i \\ \lvert \mathcal{M}_i \lvert \geq \gamma \\ j \in \mathcal{M}_i}} 
    B(\mathcal{M}_i)\Bigg\lvert 
    \leq a_{i,j} \frac{\overline{\nu}}{\underline{\mu}}.
\end{equation}

We do the same munipulation for $C_2$ than for $C_1$, i.e., we rewrite $k=p$ and make the change of variable $z = y \mu_p / \mu_j$, which gives,
\begin{equation*}
\begin{aligned}
    C_2(\mathcal{M}_i) 
    &=  \sum_{\substack{p \in\mathcal{M}_i \\ p \neq j}} \sum_{\substack{W \subset \mathcal{M}_i\setminus \{j\} \\ \lvert W \lvert = \gamma-1 \\ p \notin W}}
    \int_0^{x \mu_p / \Delta} 
    \frac{z }{\mu_p} \nu_p (z) \nu_j \left( \frac{z \mu_j}{\mu_p} \right) \\
    &\times \prod_{\substack{w \in W }} \overline{\mathscr{V}}_w  \left( \frac{z \mu_w}{\mu_p} \right)
    \prod_{\substack{\ell \in \mathcal{M}_i \setminus W\\ \ell \neq j,p}} \mathscr{V}_\ell  \left( \frac{z \mu_\ell}{\mu_p} \right) dz.
    \end{aligned}
\end{equation*}

The same argument as for $B$ holds, and we may only consider the cases where $\lvert \mathcal{M}_i\lvert \geq \gamma+1.$
For terms of $C_2$ defined by realizations $(\mathcal{M}_i, p, W)$ verifying this condition, we denote $\mathcal{M}_i' = \mathcal{M}_i \setminus \{j \}$. We again verify that $(\mathcal{M}_i',p,W)$ defines a term $S'$ of $\mathbb{P} \left[ d_{i,j}^{\gamma-1} \leq x \right]$ (see Equation~\eqref{eq:G_lipschitz_expression_d_i,j_gamma-1} and replace $\gamma-1$ by $\gamma$).
    The rest of the argument is similar as before and we obtain a bound on terms in $C_2$, i.e.,
\begin{equation}\label{equ:G_lipschitz_bound_C_2}
    \Bigg\lvert \sum_{\substack{\mathcal{M}_i \\ \lvert \mathcal{M}_i \lvert \geq \gamma \\ j \in \mathcal{M}_i}} 
    C_2(\mathcal{M}_i)  \prod_{\substack{k\in \mathcal{M}_i}} a_{i,k} \Bigg\lvert
    \leq a_{i,j} \frac{\overline{\nu}}{\underline{\mu}}.
\end{equation}

Combining Equations~\eqref{equ:G_lipschitz_deriv_H}, \eqref{equ:G_lipschitz_bound_A}, \eqref{equ:G_lipschitz_bound_C_1} , \eqref{equ:G_lipschitz_bound_B} and~\eqref{equ:G_lipschitz_bound_C_2}  finally gives a bound on the derivatives of $H_i$, i.e.,
\begin{equation}\label{equ:G_lipschitz_bound_deriv_H}
    \left \lvert \frac{\partial  H_i }{\partial \mu_j} (x, \bm{\mu_{-i}}) \right \lvert
    \leq 4 a_{i,j} \frac{\overline{\nu}}{\underline{\mu}}.
\end{equation}

\subsubsection{Lipschitz Continuity of the Expected Loss}

We use the result of the previous paragraphs to prove the Lipschtiz continuity of agent $i$'s expected loss function $L_i(\bm{\mu}) = Z_i(\bm{\mu}) -  G(\bm{\mu}).$ Since we already established that of the expected gain in Section~\ref{sec:G_sim_lear_lipschitz_G}, it only remains to study the 
expected expenditure function $Z_i(\bm{\mu}) = \mathbb{E}_{\bm{v}, \bm{m}} \left[ d_i^{\gamma} \mathds{1}\{\Delta v_i  > \mu_i d_i^\gamma \}  \right]$, where \\$d_i^\gamma = \gamma\textsuperscript{th}\mbox{-}\max_{j:j\neq i} \left\{ \mathds{1}\{ m_j = m_i\}  \Delta v_j / \mu_j \right\}.$


As values are drawn independently across agents, and since $\nu_i$ is null outside $[0,1]$, we can write:
\begin{equation}\label{equ:G_lipschitz_value_Z}
    Z_i(\bm{\mu}) = \int_0^{\Delta / \mu_i} x \left(1- \mathscr{V}_i\left( \frac{x \mu_i} {\Delta} \right) \right) dH_i(x, \bm{\mu_{-i}}).
\end{equation}
The function $\mu_i \mapsto x \left(1- \mathscr{V}_i\left( x \mu_i/\Delta\right) \right)$ is differentiable almost everywhere with derivative bounded by $\Delta ^2 / \underline{\mu}^2 \overline{\nu}$ since the distribution of valuations $\bm{\mathcal{V}}$ is absolutely continuous.
Leibniz's integral rule then implies that the expenditure function is differentiable with respect to $\mu_i$, i.e.,
\begin{equation*}
    \frac{\partial Z_i}{\partial \mu_i}(\bm{\mu})
    = - \int_0^{\Delta / \mu_i} \frac{x^2}{\Delta}  \nu_i \left( \frac{x \mu_i}{\Delta} \right)   dH_i(x, \bm{\mu_{-i}}).
\end{equation*}
Bounding $x^2 / \Delta  \nu_i \left( x \mu_i/\Delta \right)  $ by $ \Delta \overline{\nu} / \underline{\mu}^2 $ under the absolute continuity of the distribution of valuations $\bm{\mathcal{V}}$
and integrating over the competing prices gives the following bound, i.e.,
\begin{equation}\label{equ:G_lipschitz_bound_dZ_dmu_i}
    \left\lvert\frac{\partial Z_i}{\partial \mu_i} (\bm{\mu}) \right \lvert
   \leq \frac{\Delta \overline{\nu}}{\underline{\mu}^2}.
\end{equation}

Furthermore, an integration by part on Equation~\eqref{equ:G_lipschitz_value_Z} gives the following alternate expression for $Z_i(\bm{\mu})$
\begin{equation*}
    Z_i(\bm{\mu}) = \int_0^{\Delta / \mu_i} \left( \frac{x \mu_i}{\Delta} \nu_i\left( \frac{x \mu_i}{\Delta}\right) - \overline{\mathscr{V}}_i\left( \frac{x \mu_i} {\Delta} \right) \right) H_i(x, \bm{\mu_{-i}}) dx.
\end{equation*}
Let $j \neq i.$ Since $H_i(x, \bm{\mu_{-i}})$ is differentiable almost everywhere with respect to $\mu_j$  with derivative bounded according to Equation~\eqref{equ:G_lipschitz_bound_deriv_H}, Leibniz's integral rule implies that:
\begin{equation*}
    \frac{ \partial Z_i}{\partial \mu_j} (\bm{\mu}) = \int_0^{\Delta / \mu_i} 
    \left( \frac{x \mu_i}{\Delta} \nu_i\left( \frac{x \mu_i}{\Delta}\right)
    - \overline{\mathscr{V}}_i\left( \frac{x \mu_i} {\Delta} \right) \right) \frac{\partial H_i}{\partial \mu_j} (x, \bm{\mu_{-i}}) dx.
\end{equation*}
Using Equation~\eqref{equ:G_lipschitz_bound_deriv_H}, bounding $\overline{\mathscr{V}}_i\left( \frac{x \mu_i} {\Delta} \right)$ by one on one hand, and bounding $x$ by $\Delta / \mu_i$ before integrating over $\nu_i$ on the other hand,
we get the following bound on the derivative of $Z_i$, i.e.,
\begin{equation} \label{equ:G_lipschitz_bound_dZ_dmu_j}
    \left\lvert \frac{ \partial Z_i}{\partial \mu_j} (\bm{\mu}) \right \lvert 
    \leq  8 a_{i,j} \frac{\Delta \overline{\nu}}{\underline{\mu}^2} .
\end{equation}

Together, Equations~\eqref{equ:G_sim_lear_lipschitz_cont_G}, ~\eqref{equ:G_lipschitz_bound_dZ_dmu_i} and \eqref{equ:G_lipschitz_bound_dZ_dmu_j}  imply for all $\bm{\mu}$, $\bm{\mu'}$ in $\bm{U}$:
\begin{equation}\label{equ:G_lipschitz_L}
    \begin{aligned}
         \big\lvert L_i(\bm{\mu}) - L_i \left(\bm{\mu'}\right) \big\lvert 
         &\leq \big\lvert Z_i(\bm{\mu}) - Z_i \left(\bm{\mu'}\right) \big\lvert +  \big\lvert G(\bm{\mu}) - G \left(\bm{\mu'}\right) \big\vert ,\\
        &\stackrel{(a)}{\leq}   \frac{\Delta \overline{\nu}}{\underline{\mu}^2} \left(  \lvert \mu_i - \mu_i'\lvert
        + \sum_{j\neq i} 8 a_{i,j}   \lvert \mu_j - \mu_j'\lvert  \right) \\
        &+  \frac{M \gamma \Delta}{N \underline{\mu}^2} \Vert \bm{\mu} - \bm{\mu'} \Vert_2 ,\\
        &\stackrel{(b)}{\leq}  \left(\frac{ 8\Delta \overline{\nu}}{\underline{\mu}^2} \Vert \bm{a_i} \Vert_2  + \frac{M \gamma \Delta}{N \underline{\mu}^2}\right) \Vert \bm{\mu} - \bm{\mu'} \Vert_2 \\
        &+  \frac{ \Delta \overline{\nu}}{\underline{\mu}^2} \lvert \mu_i - \mu_i'\lvert.
    \end{aligned}
\end{equation}
Inequality $(a)$ follows from Equations~\eqref{equ:G_sim_lear_lipschitz_cont_G}, ~\eqref{equ:G_lipschitz_bound_dZ_dmu_i} and \eqref{equ:G_lipschitz_bound_dZ_dmu_j}, and inequality $(b)$ uses Cauchy Schwartz inequality.

In the case of a single auction, we hence define the Lipschitz constant for $L_i$
such that $ \big\lvert L_i(\bm{\mu}) - L_i \left(\bm{\mu'}\right) \big\lvert \leq \sqrt{N} \mathcal{L}_L \Vert \bm{\mu} - \bm{\mu'} \Vert_2$, i.e.,
$$\mathcal{L}_L =   \frac{ 8\Delta \overline{\nu}}{\underline{\mu}^2}  +  \frac{ \gamma \Delta}{\sqrt{N}\underline{\mu}^2}  .$$

\subsubsection{Upper Bound on the Total Expected Cost}

We now return to the main proof and use the Lipschitz continuity of $\Psi_i^0$ and $L_i$ to bound the total expected cost of strategy $K.$

Similarly as in the proof of Theorem~\ref{thm:G_stat_comp} (c.f. Section~\ref{sec:proof_G_stat_comp_tot_cost})
we bound the expected total cost of agent $i \in \mathcal{N}$ by considering it always gets maximum valuation after time $\underline{\mathscr{T}}$ but never accesses the priority road.
  \begin{equation}\label{equ:G_sim_lear_1}
      \begin{aligned}
            \mathcal{C}_i^{K} 
            &\leq \mathbb{E}_{\bm{v}}\left[ \sum_{t=1}^{\underline{\mathscr{T}}} c_{i,t}^K \right]  +   \mathbb{E}_{\bm{v}} \left[T - \underline{\mathscr{T}}   \right].
      \end{aligned}
  \end{equation}

   We now turn to bound the cost-per-period $c_{i,t}^K.$ Note that its expression in Equation~\eqref{equ:G_stat_comp_2.5} still holds in the multi-agent setting. For $t\leq \underline{\mathscr{T}}$, we again take expectation and condition on $\bm{\mu_t}$, recognizing the dual objective $\Psi_i^0$ and the loss $L_i.$
\begin{equation} \label{equ:G_sim_lear_2}
\begin{aligned}
      \mathbb{E}_{\bm{v}} \left[c_{i,t}^K \big\vert t\leq \underline{\mathscr{T}} \right]
       &= \mathbb{E}_{\bm{v}} \left[ \Psi_i^0  (\bm{\mu_t})  - \mu_{i,t} L_i (\bm{\mu_t}) \vert t\leq \underline{\mathscr{T}} \right],  
\end{aligned}
\end{equation} 

Using the Lipschitz continuity of $\Psi_i^0$ in Equation~\eqref{equ:G_lipschitz_psi} and that of $L_i$ in Equation~\eqref{equ:G_lipschitz_L}, we finally express the cost-per-period in terms of $\bm{\mu^{\star0}}.$
\begin{equation} \label{equ:G_sim_lear_5}
\begin{aligned}
      \mathbb{E}_{\bm{v}} \left[c_{i,t}^K \big\vert t\leq \underline{\mathscr{T}} \right]
       &\leq \Psi_i^0 \big(\bm{\mu^{\star0}} \big) + \left( \mathcal{L}_\Psi + \overline{\mu} \mathcal{L}_L \right) \sqrt{N} s_t^{1/2}. 
\end{aligned}
\end{equation}

Moreover, Lemma~\ref{lem:G_bound_MSE} combined with $\sqrt{x+y} \leq \sqrt{x} + \sqrt{y}$ allows to bound the series in $s_t^{1/2}$ as follows, i.e.,
\begin{equation}\label{equ:G_sim_lear_6}
    \begin{aligned}
        \mathbb{E}_{\bm{v}} \left[ \sum_{t=1}^{\underline{\mathscr{T}}} s_t^{1/2} \right]
        &\leq \sqrt{N} \left( \frac{ \Delta \epsilon^{1/2} }{\sqrt{2 \lambda} \underline{\mu} } +  \overline{\mu} \sum_{t=1}^{T} (1 - 2\lambda \epsilon)^{(t-1)/2} \right) \\
        &\stackrel{(a)}{\leq} \sqrt{N} \left( \frac{ \Delta \epsilon^{1/2} }{\sqrt{2 \lambda} \underline{\mu} } +  \frac{2\overline{\mu}}{\lambda \epsilon}  \right) \\
    \end{aligned}
\end{equation}
Inequality $(a)$ holds by bounding the geometric sum by the series' limit and noting that $ \frac{1}{1 - (1-\lambda \epsilon)^{1/2}} \leq \frac{2}{\lambda \epsilon} $ since $1 - (1-x)^{1/2} \geq x/2$ for $x\in[0,1]$.

\subsubsection{Intermediate Conclusion.}

Combining Equations~\eqref{equ:G_sim_lear_1}, \eqref{equ:G_sim_lear_5} and~\eqref{equ:G_sim_lear_6} finally gives:
\begin{align*}
    &\frac{1}{T} \mathcal{C}_i^K - \Psi_i^0 \big(\bm{\mu^{\star0}}\big) \\
    &\leq \frac{1}{T} \mathbb{E}_{\bm{v}} \left[ \sum_{t=1}^{\underline{\mathscr{T}}} c_{i,t} \right]  
    +  \frac{1}{T} \mathbb{E}_{\bm{v}} \left[T - \underline{\mathscr{T}}   \right] - \Psi_i^0 \big(\bm{\mu^{\star0}}\big) ,\\
    &\leq \left( \mathcal{L}_\Psi + \overline{\mu} \mathcal{L}_L \right) \frac{\sqrt{N}}{T} \mathbb{E}_{\bm{v}} \left[ \sum_{t=1}^{\underline{\mathscr{T}}} s_t^{1/2} \right] +  \frac{1}{T} \mathbb{E}_{\bm{v}} \left[T - \underline{\mathscr{T}}   \right] ,\\
    &\leq \left( \mathcal{L}_\Psi + \overline{\mu} \mathcal{L}_L \right) N 
    \left(  \frac{\Delta \epsilon^{1/2} }{\sqrt{2\lambda } \underline{\mu}} + \frac{2 \overline{\mu}}{\lambda \epsilon T} \right) 
     + \frac{1}{T} \mathbb{E}_{\bm{v}} \left[T - \underline{\mathscr{T}}   \right].
\end{align*}

 Rewriting $\underline{\rho} = \underline{k_1}/T$,  this proves the existence of a constant $C_2$ in $\mathbb{R}_+$ such that
\begin{equation*}
    \frac{1}{T} \mathcal{C}_i^K - \Psi_i^0 \big(\bm{\mu^{\star0}} \big)
    \leq C_2 \left( N \left( \epsilon^{1/2} +   \frac{1}{\epsilon T}\right) + \frac{ \mathbb{E}_{\bm{v}} \left[T - \underline{\mathscr{T}} \right]}{T}\right).
\end{equation*}

\subsubsection{Control of Hitting Time.}

It only remains to show that the term $\frac{1}{T} \mathbb{E}_{\bm{v}} \left[T - \underline{\mathscr{T}} \right]$ converges to zero under Assumption~\ref{ass:G-sim-learn-hitting-time}. 
Note that Assumptions~\ref{ass:G_stat_comp_hitting_time-5} and~\ref{ass:G-sim-learn-hitting-time-4} respectively imply $\mathscr{T}_i^k = T$ and $\mathscr{T}_i^{\underline{\mu}} = T$. 

We prove in the following that the multiplier profile $\bm{\mu}_t$ remains in the set $\prod_{i=1}^N [ -\overline{\mu}, \overline{\mu}]$ at all time $1\leq t \leq T$, effectively implying  $\underline{\mathscr{T}}^{\overline{\mu}} = T$, and in turn $\underline{\mathscr{T}} = T$. 
\smallskip

For $t \in \mathbb{N}$, we define $\overline{\mu}_t = \max_{i \in \mathcal{N}} \mu_{i,t}$ and   $\underline{\mu_t} = \min_{i \in \mathcal{N}} \mu_{i,t}$ . For a certain realization $\bm{\mu}_{t+1}$, we retrospectively consider a "big" agent $b \in \mathcal{N}$ such that $\mu_{b,t+1} = \overline{\mu}_{t+1}$ (not necessarily unique). We similarly consider a "small" agent $s$ satisfying $\mu_{s,t+1} = \underline{\mu_{t+1}}.$
Finally, we denote by $\mathcal{W}_t$ the set of auction winners, by $\mathcal{L}_t$ the set of the losers, and $\sum_{i= 1}^N \mu_{i,1} = N \mu_m$ the constant sum of the multipliers.
To prove the result, we show by induction that the difference $\overline{\mu}_{t} - \underline{\mu_{t}} $ remains bounded for all $1 \leq t \leq T.$ 

We proceed by exhaustion and consider at first the case where it does not happen that agent $b$ wins the auction while agent $s$ loses it.

\begin{proposition}
    Suppose it is not the case that $b \in \mathcal{W}_t$ and $s \in \mathcal{L}_t$. 
    Then $\overline{\mu}_{t+1} - \underline{\mu_{t+1}} \leq \overline{\mu}_{t} - \underline{\mu_{t}} .$
\end{proposition}
\begin{proof}
    By exhaustion, we need to consider the following two cases.
    \begin{enumerate}
        \item If  $b$ and $s$ are both in $\mathcal{L}_t$ or both in $\mathcal{W}_t$, we have:
            \begin{equation*}
                \begin{aligned}
                    \overline{\mu}_{t+1} - \underline{\mu_{t+1}} 
                    &\stackrel{(a)}{=} \mu_{b,t+1} - \mu_{s,t+1} ,\\
                    &\stackrel{(b)}{=} \mu_{b,t} - \mu_{s,t} ,\\
                    &\stackrel{(c)}{\leq} \overline{\mu}_{t} - \underline{\mu_{t}}.
                \end{aligned}
            \end{equation*}
            Equality $(a)$ holds by definition of $b$ and $s$. 
            Equality $(b)$ uses that variations in multiplier only depend on the affiliation of an agent to $\mathcal{W}_t$ or to $\mathcal{L}_t$, i.e.,  all $w \in \mathcal{W}_t$ increase their multiplier by the same value $\epsilon p_t$, while all $\ell \in \mathcal{L}_t$ decrease their multiplier by a common $\epsilon g_t$.
            Inequality $(c)$ finally follows from the definition of $\overline{\mu}_t$ and $\underline{\mu_{t}}.$
            
        \item If $b \in \mathcal{L}_t$ and $s \in \mathcal{W}_t$, we have:
            \begin{equation*}
                \begin{aligned}
                    \overline{\mu}_{t+1} - \underline{\mu_{t+1}} 
                    &\stackrel{(a)}{=} \mu_{b,t+1} - \mu_{s,t+1} ,\\
                    &\stackrel{(b)}{\leq} \mu_{b,t} - \mu_{s,t} - \epsilon \left( 1 + \frac{M \gamma} {N} \right) \frac{\Delta \underline{v}}{\overline{\mu}} ,\\
                    &\stackrel{(c)}{<} \overline{\mu}_{t} - \underline{\mu_{t}}.
                \end{aligned}
            \end{equation*}
            Equality $(a)$ holds by definition of $b$ and $s$, equality $(b)$ uses the minimum variation step of multipliers for the price $p_t = \Delta \underline{v} / \overline{\mu}.$
            Inequality $(c)$ finally follows from the definition of  $\overline{\mu}_t$ and $\underline{\mu_{t}}.$
    \end{enumerate}
\end{proof}
We now turn to analyze the remaining case where agent $b$ wins the auction while agent $s$ loses it.

\begin{proposition}
    Suppose that $b \in \mathcal{W}_t$ and $s \in \mathcal{L}_t$. Then $\overline{\mu}_{t} - \underline{\mu_{t}} \leq \mu_m  \left(  \frac{2} {\underline{v}} \left(1 - \frac{\gamma}{N} \right)^{-1} - \frac{ \underline{v} } {2 }\right) .$
\end{proposition}
\begin{proof}
    We consider an hypothetical framework where $\overline{\mu} = \infty$, and iteratively bound $\overline{\mu}_{t+1}$, $\underline{\mu_{t+1}} $, $\overline{\mu}_{t}$ and $\underline{\mu_{t}} $ .
    
    We first note that $ \mu_{b,t}$ cannot be too far below $\mu_m$, hence it is greater than $\underline{\mu}$:
  \begin{equation} \label{equ:_____1}
      \begin{aligned}
          \mu_{b,t} 
          &\stackrel{(a)}{\geq} \mu_{b,t+1} - \epsilon \frac{ \gamma}{N} \frac{\Delta }{\underline{\mu}} ,\\
          &\stackrel{(b)}{=} \overline{\mu}_{t+1} - \epsilon \frac{ \gamma}{N} \frac{\Delta }{\underline{\mu}} ,\\
          &\stackrel{(c)}{\geq} \mu_m - \epsilon \frac{ \gamma}{N} \frac{\Delta }{\underline{\mu}}, \\
          &\stackrel{(d)}{\geq} \frac{ \underline{v} } {2 } \mu_m  , \\
          &\stackrel{(e)}{>} \underline{\mu}  .
      \end{aligned}
  \end{equation}
  Equality $(a)$ uses that the maximum multiplier increase is $\epsilon \frac{ \gamma}{N} \frac{\Delta \overline{v}}{\underline{\mu}}$; equality $(b)$ follows from the definition of $b$;  inequality $(c)$ uses that the maximum must be greater than the mean; inequality $(d)$ and $(e)$ respectively follow from Assumptions~\ref{ass:G-sim-learn-hitting-time-3} and~\ref{ass:G-sim-learn-hitting-time-2}.

    We then note that for all $\ell \in \mathcal{L}_t$, we must have: 
    \begin{equation}\label{equ:relation_bids_w&l}
        \frac{ \Delta \underline{v}}{\max \{\mu_{\ell,t}, \underline{\mu} \} } 
        \leq  b_{\ell,t} \leq b_{b,t} \leq
        \frac{ \Delta }{\max \{\mu_{b,t}, \underline{\mu} \} }  = \frac{ \Delta}{\mu_{b,t} } .
    \end{equation}
    
    Using the above, we next show that $\max \{\mu_{\ell,t}, \underline{\mu} \}  = \mu_{\ell,t}$ holds for all $\ell \in \mathcal{L}_t$.  
    
\begin{equation}\label{equ:_____2}
    \begin{aligned}
        \max \{\mu_{\ell,t}, \underline{\mu} \} 
        &\stackrel{(a)}{\geq} \underline{v} \mu_{b,t} ,\\
        &\stackrel{(b)}{\geq} \underline{v} \left( \mu_m - \epsilon \frac{ \gamma}{N} \frac{\Delta }{\underline{\mu}} \right) ,\\
        &\stackrel{(c)}{\geq} \frac{\underline{v}}{2} \mu_m  ,\\
        &\stackrel{(d)}{\geq} \underline{\mu} .
    \end{aligned}
\end{equation}
 Inequalities $(a)$ and $(b)$ use the inequalities in Equations~\eqref{equ:relation_bids_w&l} and~\eqref{equ:_____1} respectively; inequalities $(c)$ and $(d)$ respectively follow from Assumptions~\ref{ass:G-sim-learn-hitting-time-3} and~\ref{ass:G-sim-learn-hitting-time-2}.

As Equation~\eqref{equ:_____2} holds in particular for $s$, we propagate this lower-bound to $\underline{\mu_{t}}$ and $\underline{\mu_{t+1}}$:
 \begin{equation} \label{equ:_____2.5}
     \begin{aligned}
        \min\{ \underline{\mu_{t}}, \underline{\mu_{t+1}} \}
        &\stackrel{(a)}{\geq} \mu_{s,t} - \epsilon \left( 1 + \frac{ \gamma}{N} \right)\frac{\Delta }{\underline{\mu}} ,\\
        &\stackrel{(b)}{\geq} \underline{v} \mu_m - \epsilon \left( 1+ \left( 1 + \underline{v} \right) \frac{ \gamma}{N}  \right)  \frac{\Delta }{\underline{\mu}}  ,\\
        &\stackrel{(c)}{\geq}  \frac{\underline{v}}{2} \mu_m  ,\\
        &\stackrel{(d)}{>} \underline{\mu} .
     \end{aligned}
 \end{equation}
 Equality $(a)$ uses that both $ \underline{\mu_{t}}$ and $ \underline{\mu_{t+1}}$  are at most $\epsilon \left( 1 + \frac{ \gamma}{N} \right) \frac{\Delta }{\underline{\mu}}$ away from $\mu_{s,t}$ by definition of $s$; inequality $(b)$ makes use of the the lower-bound in Equation~\eqref{equ:_____2} $(b)$, inequalities $(c)$ and $(d)$ 
 respectively follow from Assumptions~\ref{ass:G-sim-learn-hitting-time-3} and~\ref{ass:G-sim-learn-hitting-time-2}.

We now proceed to upper-bound $ \overline{\mu}_t$ and $ \overline{\mu}_{t+1}$. By conservation of the sum of multipliers, we have:

\begin{equation}\label{equ:_____3}
    \begin{aligned}
        N \mu_m &= \mu_{b,t} +  \sum_{\ell \in \mathcal{L}_t} \mu_{\ell,t}  + \sum_{w \in \mathcal{W}_t \setminus \{ b\} } \mu_{w,t} ,\\
        &\geq \mu_{b,t} \left(1 + (N-\gamma) \underline{v} \right) + (\gamma - 1) \underline{\mu}.
    \end{aligned}
\end{equation}
The inequality follows from Equation~\eqref{equ:_____2} $(a)$,  as well as the definition of $\underline{\mu_t}$ and Equation~\eqref{equ:_____2.5}. 

It then follows that:
\begin{equation}
    \begin{aligned}
        \max\{ \overline{\mu}_t, \overline{\mu}_{t+1} \}
        &\stackrel{(a)}{\leq}  \mu_{b,t} + \epsilon \left( 1 + \frac{ \gamma}{N} \right)\frac{\Delta}{\underline{\mu}}  ,\\ 
        &\stackrel{(b)}{\leq} \frac{ N \mu_m - (\gamma - 1) \underline{\mu} } {1 + (N-\gamma) \underline{v}} + \epsilon \left( 1 + \frac{ \gamma}{N} \right)\frac{\Delta }{\underline{\mu}} , \\ 
        & < \frac{1}{\underline{v}}   \mu_m \left(1 - \frac{\gamma}{N} \right)^{-1} + \epsilon \left( 1 + \frac{ \gamma}{N} \right)\frac{\Delta }{\underline{\mu}} ,\\
         &\stackrel{(c)}{\leq} \frac{2}{\underline{v}}   \mu_m \left(1 - \frac{\gamma}{N} \right)^{-1}.
    \end{aligned}
\end{equation}
Equality $(a)$ uses that both $ \overline{\mu}_{t}$ and $ \overline{\mu}_{t+1}$ are at most $\epsilon \left( 1 + \frac{ \gamma}{N} \right) \frac{\Delta }{\underline{\mu}}$ away from $\mu_{b,t}$ by definition of $b$; inequality $(b)$ follows from Equation~\eqref{equ:_____3} and inequality $(c)$  from Assumption~\ref{ass:G-sim-learn-hitting-time-3}.

This finally proves that $ \overline{\mu}_{t+1} - \underline{\mu_{t+1}} \leq \mu_m  \left(  \frac{2} {\underline{v}} \left(1 - \frac{\gamma}{N} \right)^{-1} - \frac{ \underline{v} } {2 }\right) .$
\end{proof}

Together, the two propositions show the induction and $\overline{\mu}_{t} - \underline{\mu_{t}} \leq \mu_m  \left(  \frac{2} {\underline{v}} \left(1 - \frac{\gamma}{N} \right)^{-1} - \frac{ \underline{v} } {2 }\right) $ for all $1 \leq t\leq T.$ 
Besides, since $\underline{\mu_t} \leq \mu_m \leq \overline{\mu}_t$, we finally get
$$\mu_t \in \left[\mu_m  \left(  1 +  \frac{ \underline{v} } {2 } - \frac{2} {\underline{v}} \left(1 - \frac{\gamma}{N} \right)^{-1}\right) , \mu_m  \left(  1 +  \frac{2} {\underline{v}} \left(1 - \frac{\gamma}{N} \right)^{-1} - \frac{ \underline{v} } {2 } \right) \right].$$ 
Noting that the latter interval is part of  $[-\overline{\mu}, \overline{\mu} ] $ under Assumption~\ref{ass:G-sim-learn-hitting-time-2} concludes the proof.

\subsection{Proof of Theorem~\ref{thm:G_eps_NE}}\label{sec:G_eps_NE_proof}

In this section, we consider that all agents follow strategy $K$ except for agent $i$ in $\mathcal{N}$ following strategy $\beta \in \mathcal{B}^{CI}$ (we refer the reader to the discussion before Lemma~\ref{lem:G_bound_MSE_eps_NE} for the definition of $\mathcal{B}^{CI}$). To simplify notations, we drop the superscript $K_{-i}$ and write $c_i^\beta$, $b_i^\beta$, $z_i^\beta$, and $g_i^\beta$ for the respective costs, bids, expenditures, and gains of agent $i$. This is not to be confused with $c_i^{K}$, $b_i^{K}$, $z_i^{K}$ and $g_i^{K}$, which are the hypothetical variables of agent $i$, had it followed a $K$ strategy (but all other agents still placing bids according to $K$ strategies). Importantly, we use the superscripts $\beta$ or $K$ for the prices $(p_m)_{m=1}^M$ of the different auctions instead of the superscript $\gamma+1.$ 
Since agent $i$ does not follow $K$, $\mu_{i,t}$ is a priory not defined. For convenience, we take the convention $\mu_{i,t} = \mu_i^{\star 0 }$ for all $t \in \mathbb{N}.$

\subsubsection{Lower Bound on the Total Expected Cost under Strategy $\beta.$}
We lower-bound the total expected cost suffered by agent $i$ for following strategy $\beta$ by considering it always gets a null cost after time $\underline{\mathscr{T}}$ and by using that strategy $\beta$ respects the budget constraint on average. 
\begin{equation}\label{equ:G_eps_NE_1}
   \begin{aligned}
     \mathcal{C}_i^{\beta, K_{-i}} 
     &\geq \mathbb{E}_{\bm{v}, \bm{m}} \left[ \sum_{t=1}^{ \underline{\mathscr{T}}} c_{i,t}^{\beta}  \right],\\
     &\stackrel{(a)}{\geq} \mathbb{E}_{\bm{v}, \bm{m}} \left[ \sum_{t=1}^{ \underline{\mathscr{T}}} c_{i,t}^{\beta}  \right] 
     + \mu_i^{\star 0} \mathbb{E}_{\bm{v}, \bm{m}} \left[ \sum_{t=1}^{ T} z_{i,t}^{\beta} - \rho_i - g_{i,t}^{\beta} \right], \\
     &\stackrel{(b)}{\geq} \mathbb{E}_{\bm{v}, \bm{m}} \left[ \sum_{t=1}^{ \underline{\mathscr{T}}}   c_{i,t}^{\beta} 
     +\mu_i^{\star 0} \left( z_{i,t}^{\beta} - g_{i,t}^{\beta} - \rho_i \right) \right]\\
     &- \left(  \frac{ M \gamma\Delta}{N \underline{\mu}} + \rho_i\right) \mathbb{E}_{\bm{v}, \bm{m}} \left[ T - \underline{\mathscr{T}}\right].
   \end{aligned}
   \end{equation}
   Inequality $(a)$ uses that strategy $\beta$ respects the budget constraint on average, and inequality $(b)$ holds since $g_{i,t}^{\beta} \leq \frac{ M \gamma\Delta}{N \underline{\mu}}  .$

 We now bound the Lagrangian cost of agent $i$ for the $t\textsuperscript{th}$ period.
   \begin{equation}\label{equ:G_eps_NE_2}
   \begin{aligned}
        &c_{i,t}^{\beta} + \mu_i^{\star0} \left( z_{i,t}^{\beta} - g_{i,t}^\beta - \rho_i \right) \\
        &= v_{i,t} - \mathds{1} \left\{ b_{i,t}^\beta \geq d_{i,t}^\gamma \right\} \left( \Delta v_{i,t} - \mu_i^{\star0} d_{i,t}^\gamma \right) - \mu_i^{\star0} \left( \rho_i + g_{i,t}^\beta \right) ,\\
        &\geq v_{i,t} - \left( \Delta v_{i,t} - \mu_i^{\star0} d_{i,t}^\gamma \right)^+ - \mu_i^{\star0} \left( \rho_i + g_{i,t}^\beta \right).
   \end{aligned}
   \end{equation}
    Taking expectation 
    , recognizing the dual objective $\Psi_i^0$ and using its Lipschitz continuity gives
\begin{equation}\label{equ:G_eps_NE_3}
   \begin{aligned}
        &\mathbb{E}_{\bm{v}, \bm{m}} \left[c_{i,t}^{\beta} + \mu_i^{\star0} \left( z_{i,t}^{\beta} -  g_{i,t}^\beta - \rho_i \right) \Big\vert t\leq \underline{\mathscr{T}}\right] \\
        & \stackrel{(a)}{\geq} \mathbb{E}_{\bm{v}, \bm{m}} \left[ \Psi^0_i ( \bm{\mu_t}) \vert t\leq \underline{\mathscr{T}} \right] 
        + \mu_i^{\star 0}  \left(\mathbb{E}_{\bm{v}, \bm{m}} \left[ g_{i,t}^K - g_{i,t}^\beta \Big\vert t\leq \underline{\mathscr{T}}\right]  - \rho_i \right),\\
        & \stackrel{(b)}{\geq} \mathbb{E}_{\bm{v}, \bm{m}} \left[ \Psi_i^0 \big( \bm{\mu^{\star0}} \big) -  \frac{\Delta}{\underline{\mu}} \left( \Vert \bm{a_i} \Vert_2 +   \frac{M\gamma  \overline{\mu} }{N \underline{\mu}} \right)    \left\Vert \bm{\mu_t} - \bm{\mu^{\star0}} \right\Vert_2 \Big\vert t\leq \underline{\mathscr{T}}\right] \\
        &+   \mu_i^{\star 0} \left( \frac{\gamma}{N} \mathbb{E}_{\bm{v}, \bm{m}} \left[ p_{m_i,t}^{K} - p_{m_i,t}^{\beta} \Big\vert t\leq \underline{\mathscr{T}}\right] -\rho_i \right) , \\
        &\stackrel{(c)}{\geq} \Psi_i^0 \big( \bm{\mu^{\star0}} \big) -  \frac{\Delta}{\underline{\mu}} \left( \Vert \bm{a_i} \Vert_2 +  \frac{M\gamma  \overline{\mu} }{N \underline{\mu}} \right)    s_t^{1/2} -  \overline{\mu} \left( \frac{\gamma\Delta}{N \underline{\mu}} +\rho_i \right).
   \end{aligned}
   \end{equation}
    Inequality $(a)$ follows from Equation~\eqref{equ:G_eps_NE_2} by adding and subtracting $g_{i,t}^{K}$ and using that $\mu_{i,t} = \mu_i^{\star0}$ by convention.
    Inequality $(b)$ uses the Lipschitz continuity of $\Psi_i^0$ in Equation~\eqref{equ:G_lipschitz_psi} and the expression of $g_t = \frac{\gamma}{N} \sum_{m=1}^M p_{m,t}$; 
    inequality $(c)$ finally uses Jensen's inequality, as well as the bounds on the price $p_m$ of any auction $0 \leq p_m \leq \Delta / \underline{\mu}.$

    Together, Equations~\eqref{equ:G_eps_NE_1} and~\eqref{equ:G_eps_NE_3} finally give 
    \begin{equation}\label{equ:G_eps_NE_8}
        \begin{aligned}
         \mathcal{C}_i^{\beta, K_{-i}} 
        &\geq T \Psi_i^0 \big( \bm{\mu^{\star0}}\big) - \frac{\Delta}{\underline{\mu}} \left( \Vert \bm{a_i} \Vert_2 +  \frac{M\gamma  \overline{\mu} }{N \underline{\mu}} \right)  \mathbb{E}_{\bm{v}, \bm{m}} \left[\sum_{t=1}^{ \underline{\mathscr{T}}} s_t^{1/2} \right] \\
        &- \overline{\mu} T \left( \frac{\gamma\Delta  }{N \underline{\mu}} + \overline{\rho} \right)
        - \left(  \frac{ M \gamma\Delta}{N \underline{\mu}} + \overline{\rho} \right) \mathbb{E}_{\bm{v}, \bm{m}} \left[ T - \underline{\mathscr{T}}\right] .
       \end{aligned}
    \end{equation}

\subsubsection{Extension of the Proof of Theorem~\ref{thm:G_sim_lear_cv} to a Parallel Auction Setting.}

Since Equation~\eqref{equ:G_sim_lear_2} still holds for the cost-per-period of strategy $K$ in the parallel auctions setting, we again use the Lipschitz continuity of the dual and expenditure functions to bound its different terms.
Using the more careful result from Equation~\eqref{equ:G_lipschitz_psi} with $\mu_{i,t} =\mu_i^{\star0}$  gives however
\begin{equation*}
        \Psi_i^0(\bm{\mu_t}) 
        \leq \Psi_i^0 \big(\bm{\mu^{\star0}}\big) + \frac{\Delta}{\underline{\mu}} \left( \Vert \bm{a_i} \Vert_2 +  \frac{M\gamma  \overline{\mu} }{N \underline{\mu}} \right)    \left\Vert \bm{\mu_t} - \bm{\mu^{\star0}} \right\Vert_2.
 \end{equation*}
 Similarly, using the Lipschitz continuity of $L_i$ from Equation~\eqref{equ:G_lipschitz_L}, we get
 \begin{equation*}
    \begin{aligned}
        -\mu_{i,t} L_i (\bm{\mu_t})  
        &\leq \overline{\mu} \left(\frac{ 8\Delta \overline{\nu}}{\underline{\mu}^2} \Vert \bm{a_i} \Vert_2  + \frac{M \gamma \Delta}{N \underline{\mu}^2}\right) \left\Vert \bm{\mu_t} - \bm{\mu^{\star0}} \right\Vert_2 .
    \end{aligned}
\end{equation*}

These two bounds, together with Equation~\eqref{equ:G_sim_lear_2} and Jensen's inequality, finally lead to the following expression, i.e.,
\begin{equation*}
      \begin{aligned}
       \mathbb{E}_{\bm{v}, \bm{m}} \left[c_{i,t}^K \big\vert t\leq \underline{\mathscr{T}} \right] 
       &\leq \Psi_i^0 \big(\bm{\mu^{\star0}}\big) +  \left( \frac{\Delta}{\underline{\mu}}  \left(1 + \frac{8  \overline{\mu} \overline{\nu}}{\underline{\mu}} \right)  \Vert \bm{a_i} \Vert_2 + \frac{2M \gamma \Delta \overline{\mu} }{N \underline{\mu}^2}\right)
       s_t^{1/2}.
      \end{aligned}
  \end{equation*}
  
Similarly as in Section~\ref{sec:G_sim_lear_cv_proof}, we conclude using Equation~\eqref{equ:G_sim_lear_1} that
\begin{equation} \label{equ:G_eps_NE_7}
    \begin{aligned}
         \mathcal{C}_i^K 
        &\leq  T \Psi_i^0 \big(\bm{\mu^{\star0}}\big)  +\left( \frac{\Delta}{\underline{\mu}}  \left(1 + \frac{8  \overline{\mu} \overline{\nu}}{\underline{\mu}} \right)  \Vert \bm{a_i} \Vert_2 +  \frac{2 M \gamma \Delta \overline{\mu} }{N \underline{\mu}^2}\right) \mathbb{E}_{\bm{v}, \bm{m}} \left[\sum_{t=1}^{\underline{\mathscr{T}}} s_t^{1/2} \right] \\
        &+  \mathbb{E}_{\bm{v}, \bm{m}} \left[ T - \underline{\mathscr{T}}\right] .\\
    \end{aligned}
\end{equation}

Moreover, we 
use Lemma~\ref{lem:G_bound_MSE_eps_NE} with $\beta = G \in \mathcal{B}^{CI}$ to derive an equivalent of Equation~\eqref{equ:G_sim_lear_6} by following the same argument, i.e.,
\begin{equation}\label{equ:G_eps_NE_9}
    \begin{aligned}
         \mathbb{E}_{\bm{v}, \bm{m}} \left[\sum_{t=1}^{ \underline{\mathscr{T}}} s_{t}^{1/2} \right]
         &\leq \frac{2 \overline{\mu} \sqrt{N}}{\lambda \epsilon} 
          + \frac{ \Delta \sqrt{N \epsilon} }{ \sqrt{\lambda} \underline{\mu}} T
        + \frac{\Delta } {\lambda \underline{\mu}} \left( \Vert \bm{a_i}\Vert_2 +  \frac{\gamma}{\sqrt{N}}\right) T.
    \end{aligned}
\end{equation}

\paragraph{Conclusion.}
  
   Together, Equations~\eqref{equ:G_eps_NE_8}, \eqref{equ:G_eps_NE_7}  and~\eqref{equ:G_eps_NE_9} finally lead to
   \begin{equation*}
       \begin{aligned}
         \mathcal{C}_i^{K} - \mathcal{C}_i^{\beta, K_{-i}} 
        &\leq \left( \frac{2\Delta}{\underline{\mu}}  \left(1 + \frac{4  \overline{\mu} \overline{\nu}}{\underline{\mu}} \right)  \Vert \bm{a_i} \Vert_2 
        + \frac{M\gamma \overline{\mu}}{N\underline{\mu}} \left(1 + \frac{\Delta}{\underline{\mu}} \right)
        \right) \Bigg( \frac{2 \overline{\mu} \sqrt{N}}{\lambda \epsilon} \\
          &+ \frac{ \Delta \sqrt{N \epsilon} }{ \sqrt{\lambda} \underline{\mu}} T
        + \frac{\Delta } {\lambda \underline{\mu}} \left( \Vert \bm{a_i}\Vert_2 +  \frac{\gamma}{\sqrt{N}}\right) T\Bigg)
        + \overline{\mu} T \left( \frac{\gamma\Delta  }{N \underline{\mu}} + \overline{\rho} \right)\\
        &+ \left( 1 + \frac{ M \gamma\Delta}{N \underline{\mu}} + \overline{\rho} \right) \mathbb{E}_{\bm{v}, \bm{m}} \left[ T - \underline{\mathscr{T}}\right] . 
       \end{aligned}
   \end{equation*}
Rewriting $\overline{\rho} = \overline{k_1}/T$,  this proves the existence of a constant $C$  in $\mathbb{R}_+$ such that
\begin{equation*}
\begin{aligned}
     \frac{1}{T} \left(\mathcal{C}_i^{K} - \mathcal{C}_i^{\beta, K_{-i}} \right)
    &\leq C \Bigg( \left( \Vert \bm{a_i}\Vert_2 + \frac{M\gamma}{N} \right) \bigg( \sqrt{N\epsilon} \left(1 + \frac{1}{\epsilon^{3/2} T}\right) + \Vert \bm{a_i}\Vert_2 \\
    &+ \frac{\gamma}{\sqrt{N}}  \bigg) 
    + \left( \frac{\gamma}{N} + \frac{\overline{k_1}}{T} \right)
    +   \left( \frac{\overline{k_1}}{T} + \frac{M\gamma}{N} \right) 
    \frac{\mathbb{E}_{\bm{v}, \bm{m}} \left[T - \underline{\mathscr{T}}   \right]}{T} \Bigg).
\end{aligned}
\end{equation*}
 \section{Proofs of the Supplementary Results }\label{sec:suppl_res}

In this section, we gather the proofs of secondary results discussed in the main text of the article, mainly discussions on the relaxation of certain assumptions.

\subsection{Sufficient Conditions for Assumption~\ref{ass:G_stat_comp_L_mono}} \label{sec:A_stat_comp_proof_imply_differentiability+concavity}

In this section, we provide mild differentiability and continuity conditions on the valuation and competing bid distributions that guarantee that the technical Assumption
~\ref{ass:G_stat_comp_L_mono} is satisfied.
This verifies that the conditions of Theorem
~\ref{thm:G_stat_comp} are not restrictive.

\begin{assumption}\label{ass:A_stat_comp_imply_differentiability+concavity}
     The following conditions hold:
    \begin{enumerate}[label=\ref{ass:A_stat_comp_imply_differentiability+concavity}.\arabic*]
        \item \label{ass:A_stat_comp_imply_diff_val_density}The valuation density $\nu_i$ is differentiable with bounded derivative $\lvert \nu_i'(v) \lvert \leq \overline{\nu}'$ for all $v \in [0,1]$;

        \item \label{ass:A_stat_comp_imply_diff_comp_bids_density}The distribution of competing bids $\mathcal{D}_i^\gamma$ is absolutely continuous with bounded density $h_i: [0, \Delta/\underline{\mu_i}] \mapsto [\underline{h}, \overline{h}]  \subset \mathbb{R}_{>0}$;

        \item \label{ass:A_stat_comp_imply_diff_varepsilon}The residual gain satisfies $\hat{\varepsilon} \leq \dfrac{\underline{\nu} \:  \underline{h} \: \Delta^4}{3  \: \overline{\nu} \: \overline{\mu}^4}$.
    \end{enumerate}
 \end{assumption}

In the following, we show that Assumption
~\ref{ass:G_stat_comp_L_mono} is 
implied by the absolute continuity of the distribution of valuations $\bm{\mathcal{V}}$
and Assumption~\ref{ass:A_stat_comp_imply_differentiability+concavity} in five different steps. We first bound the first and second derivatives of the expected gain function $G: \mu \mapsto \mathbb{E}_{v, d} \left[\gamma p^{\gamma+1} / N \right]$; then give upper and lower bounds on the derivative of the expenditure function $Z: \mu \mapsto \mathbb{E}_{v, d} \left[d^\gamma \mathds{1}\{\Delta v  > \mu d^\gamma \}  \right]$ and bounds its second order derivative;  we conclude by showing that the dual function $\Psi^0: \mu \mapsto \mathbb{E}_{v, d} \left[v \tau -\mu g - (\Delta v - \mu d^\gamma )^+\right]$ is differentiable. 

\subsubsection{Bounds on the Derivative of $G$}

The function $\mu \mapsto \frac{\gamma}{N} p $ is differentiable with derivative $\mu \mapsto -\frac{\gamma \Delta v}{N \mu^2} \mathds{1} \big\{  d^{\gamma} > \frac{\Delta v}{\mu} > d^{\gamma+1} \big\}$, except in the sets $\left\{ d^{\gamma} = \Delta v / \mu \geq d^{\gamma+1}\right\}$ and $\left\{ d^{\gamma} \geq \Delta v / \mu = d^{\gamma+1} \right\}$ of measure zero since valuations and competing prices are absolutely continuous
under Assumption~\ref{ass:A_stat_comp_imply_diff_comp_bids_density}. Since the derivative is bounded by $ \frac{\gamma \Delta }{N \underline{\mu}^2}$ which is integrable, 
Leibniz's integral rule implies that the gain function $G$ is differentiable, i.e.,
\begin{equation}\label{equ:G_stat_compt_diff_conc_0}
    G^{'} (\mu) =  - \frac{\gamma  }{N } \mathbb{E}_{\bm{v}, \bm{d}} \left[ \frac{\Delta v}{\mu^2} \mathds{1} \left\{  d^{\gamma} > \Delta v / \mu > d^{\gamma+1} \right\} \right].
\end{equation}
In particular, $G^{'}$ is negative on $[\underline{\mu}, \overline{\mu}].$

We now carefully bound $G^{'}.$
Let $H: d^\gamma \times d^{\gamma+1} \mapsto H\left( d^\gamma , d^{\gamma+1} \right) $ denote the cumulative probability function of competing bids.
\begin{equation}\label{equ:G_stat_compt_diff_conc_1}
\begin{aligned}
     \left\vert G^{'} (\mu) \right \vert 
    &=   \frac{\gamma }{N}  \int_{x=0}^{\Delta/ \underline{\mu}} \int_{y=0}^{\Delta/ \underline{\mu}} \int_{z= \mu y / \Delta}^{\mu x / \Delta} \nu(z) dH(x,y) dz ,\\
    & \stackrel{(a)}{\leq} \frac{\gamma  }{N} \frac{\overline{\mu} \overline{\nu}}{\Delta}  \int_{x=0}^{\Delta/ \underline{\mu}} \int_{y=0}^{\Delta/ \underline{\mu}} (x - y) dH(x,y) ,\\
    &= \frac{\gamma  }{N} \frac{\overline{\mu} \overline{\nu}}{\Delta}  \mathbb{E}_{\bm{v}, \bm{d}}\left[ d^\gamma - d^{\gamma+1} \right] ,\\
\end{aligned}
\end{equation}
where inequality $(a)$ uses the absolute continuity of the distribution of valuations $\bm{\mathcal{V}}.$ 

\subsubsection{Bound on the Second Order Derivative of $G$}

Furthermore, the function $\mu \mapsto - \frac{\gamma \Delta v}{N \mu^2} \mathds{1} \left\{  d^{\gamma} > \Delta v / \mu > d^{\gamma+1} \right\} $ is differentiable outside of the same measure zero sets, and its derivative, bounded by $ \frac{3 \gamma \Delta }{N \underline{\mu}^3}$, is still integrable. Leibniz's integral rule hence implies that $G^{'}$ is also differentiable, i.e.,
\begin{equation*}
\begin{aligned}
    \left \vert G^{''} (\mu) \right\vert
    &=   \left \vert \frac{3 \gamma \Delta }{N \mu^3 } \mathbb{E}_{\bm{v}, \bm{d}} \left[ v \mathds{1} \left\{  d^{\gamma} > \Delta v / \mu > d^{\gamma+1} \right\} \right] \right\vert
    \leq  \frac{3 \Delta }{ \underline{\mu}^3 },
\end{aligned}
\end{equation*}
where we used that the fraction of winning agents $\frac{\gamma}{N} $ is smaller than one.

\subsubsection{Bounds on the Derivative of $Z$ and Strong Monotonicity of $L$}

Let $\mathscr{V}$ denote the cumulative distribution function of valuations associated with the density $\nu$. Since $\nu$ is null outside of $[0,1]$, we can write the expenditure function as follows:
\begin{equation*}
    Z(\mu) = \int_0^{\Delta / \mu } x  \left(1 - \mathscr{V} \left( \frac{\mu x}{\Delta} \right)\right) h(x) dx.
\end{equation*}
Since $\mu \mapsto x (1- \mathscr{V}(\mu x / \Delta))$ is differentiable almost everywhere with derivative bounded by $\frac{\Delta^2}{\underline{\mu}^2} \overline{\nu}$, Leibniz's integral rule 
gives:
\begin{equation}\label{equ:A_stat_compt_diff_conc_2}
    \begin{aligned}
         Z^{'}(\mu) 
         &= - \int_0^{\Delta / \mu } \frac{x^2}{\Delta} \nu\left( \frac{\mu x}{\Delta} \right)  h(x) dx.
    \end{aligned}
\end{equation}
Using the bounded densities in Assumption~\ref{ass:A_stat_comp_imply_differentiability+concavity}, we finally  bounds  the derivative, i.e.,
\begin{equation} \label{equ:A_stat_compt_diff_conc_3}
     -\frac{\Delta} {\underline{\mu}^2}\overline{\nu} \leq Z^{'}(\mu) \leq - \underline{\nu} \underline{h} \int_0^{\Delta / \overline{\mu} } \frac{x^2}{\Delta} dx = - \frac{\underline{\nu} \underline{h} \Delta^3}{3 \overline{\mu}^3}.
\end{equation}

Combining Equation~\eqref{equ:G_stat_compt_diff_conc_1} with  Equation~\eqref{equ:A_stat_compt_diff_conc_3}  and the fact that $G^{'} $ is negative on $[\underline{\mu}, \overline{\mu}]$ yields
\begin{equation*} 
    L^{'}(\mu) = Z^{'}(\mu) - G^{'}(\mu) 
    \leq   \frac{\overline{\mu} \overline{\nu}}{\Delta}  \frac{\gamma  }{N} \mathbb{E}_{\bm{v}, \bm{d}}\left[ d^\gamma - d^{\gamma+1} \right]  - \frac{\underline{\nu} \underline{h} \Delta^3}{3 \overline{\mu}^3} \stackrel{(a)}{<} 0,
\end{equation*}
where inequality $(a)$ follows from Assumption~\ref{ass:A_stat_comp_imply_diff_varepsilon}. 
This proves the strong monotonicity of the loss function $L.$

\subsubsection{Bound on the Second Order Derivative of $Z$}

Since $\mu \mapsto \frac{x^2}{\Delta} \nu\left( \frac{\mu x}{\Delta} \right)  h(x)$ is differentiable almost everywhere with derivative bounded by $\frac{\Delta}{\underline{\mu}^3} \overline{\nu}^{'}$, using Leibniz integral rule on Equation~\eqref{equ:A_stat_compt_diff_conc_2} gives the following, i.e.,
\begin{equation*} 
         Z^{''}(\mu) 
         = \frac{\Delta^2}{\mu^4} \nu(1) h\left(\frac{\Delta}{\mu}\right)  - \int_0^{\Delta / \mu } \frac{x^3}{\Delta^2} \nu^{'}\left( \frac{\mu x}{\Delta} \right)  h(x) dx.
\end{equation*}
Using the absolute continuity of the distribution of valuations $\bm{\mathcal{V}}$ and
Assumption~\ref{ass:A_stat_comp_imply_differentiability+concavity}, we finally get the desired bound.
\begin{equation*}
    \left\lvert  Z^{''}(\mu) \right\lvert \leq \frac{\Delta^2}{\underline{\mu}^4} \overline{\nu} \overline{h} + \frac{\Delta}{\underline{\mu}^3} \overline{\nu}^{'}
\end{equation*}

\subsubsection{Differentiability of $\Psi^0$}

On one hand, the function $\mu \mapsto (\Delta v - \mu d^\gamma )^+$ is differentiable with derivative $ \mu \mapsto d^\gamma \mathds{1}\{\Delta v  > \mu d^\gamma \}$, except in the set $\{ (v,d^\gamma) : \Delta v = \mu d^\gamma \} $ of measure zero since valuations and competing prices are both absolutely continuous with respective support in $[0, \overline{\nu}]$ and $[0, \Delta/\underline{\mu}].$ As the derivative is bounded by $\Delta/\underline{\mu}$, which is integrable, Leibniz's integral rule ensures that $\mu \mapsto (\Delta v - \mu d^\gamma )^+$ is differentiable. 

On the other hand, Equation~\eqref{equ:G_stat_compt_diff_conc_0} implies that the function $\mu \mapsto \mu G(\mu)$ is differentiable with derivative $G(\mu) + \mu G'(\mu)$. 
Hence the dual function $\Psi^0$ is differentiable and we have for all $\mu \geq 0$, i.e.,
\begin{equation*}
\begin{aligned}
    \Psi^{0'}(\mu) &= d^\gamma \mathds{1}\{\Delta v  > \mu d^\gamma \}  - (G(\mu) + \mu G'(\mu) ) \\
    &= L(\mu) -\mu G'(\mu).
    \end{aligned}
\end{equation*}

\subsection{Uniqueness of Stationary Multiplier Profile } \label{sec:A_sim_lear_unique_mu^*}

We prove in the following that Assumption~\ref{ass:G_sim_lear_strong_mono} ensures the uniqueness of a stationary multiplier profile up to a multiplicative constant.

 By contradiction, suppose that two different vectors $\bm{\mu}$ and $\bm{\mu'}$ in $\bm{U}\cap \bm{H_{\mu_1}}$ satisfy the definition of a stationary multiplier. We have:
 \begin{align*}
    0 \stackrel{(a)}{>} - \lambda \Vert \bm{\mu} - \bm{\mu'}\Vert_2^2 \stackrel{(b)}{\geq}  \sum_{i=1}^N (\mu_i - \mu_i') \left( L_i(\bm{\mu}) - L_i(\bm{\mu'})\right) 
    \stackrel{(c)}{=} 0. 
 \end{align*}
 Inequality $(a)$ uses that $\bm{\mu} \neq \bm{\mu'}$ and the norm is positive definite;iinequality $(b)$ follows from the strong monotonicity of $\bm{L}$, and equality $(c)$ uses that both $\bm{\mu}$ and $\bm{\mu'}$ are stationary multipliers. This is a contradiction.

\end{document}